\documentclass[11pt]{article}
\pdfoutput=1
\usepackage{fullpage}
\usepackage{amsmath,amsthm,amssymb}
\usepackage{graphicx,float,psfrag,epsfig,amssymb}
\usepackage[top=1in, bottom=1.25in, left=1in, right=1in]{geometry}
\usepackage[usenames,dvipsnames,svgnames,table]{xcolor}
\definecolor{darkgreen}{rgb}{0.0,0,0.9}
\RequirePackage[pagebackref,colorlinks=true,citecolor=OliveGreen,linkcolor=BrickRed,urlcolor=BrickRed,pdfstartview=FitH]{hyperref}
\usepackage{caption,subcaption}
\usepackage[font={small}]{caption} 
\usepackage{float}
\usepackage{multirow}
\usepackage{slashbox}

\usepackage[utf8]{inputenc}
\usepackage{scalerel,stackengine}
\usepackage[mathscr]{euscript}

 \let\chapter\section
 \usepackage[linesnumbered,lined,boxed,commentsnumbered]{algorithm2e}
 \usepackage[noend]{algorithmic}

\graphicspath{ {figs}{figs/BatchExample}{figs/EXP1}{figs/EXP2-1}{figs/EXP3}{figs/EXP4}{figs/EXP5}   }

\newtheorem{propo}{Proposition}[section]
\newtheorem{lemma}[propo]{Lemma}

\newtheorem{assumption}[propo]{Assumption}
\newtheorem{proposition}[propo]{Proposition}

\newtheorem{corollary}[propo]{Corollary}
\newtheorem{remark}[propo]{Remark}
\newtheorem{thm}[propo]{Theorem}
\newtheorem{theorem}[propo]{Theorem}

\newtheorem{definition}[propo]{Definition}
\newtheorem{example}[propo]{Example}

\def\Bsum{\sum_{\ell=1}^{K-1} \sum_{t\in E_{\ell}}}

\def\Sell{\widehat{\Sigma}^{(\ell)}}

\def\Rell{{R}^{(\ell)}}

\def\Mell{M^{(\ell)}}
\def\mli{m^\ell_a}

\def\Lsth{{\hth^\sL}}
\def\mge{\succcurlyeq}

\def\lambdalbd{{\Lambda_0}}
\def\psdcone{{\mathbb{S}}}
\def\lambdamin{{\lambda_{\rm min}}}
\def\lambdamax{{\lambda_{\rm max}}}
\newcommand{\paren}[1]{{(#1)}}

\def\cF{{\cal F}}
\def\fF{{\frak F}}

\def\cA{{\cal A}}

\def\cS{{\cal S}}

\def\sL{{\sf L}}

\newcommand{\iid}{\stackrel{\mathrm{iid}}{\sim}}
\newcommand{\ind}{\stackrel{\mathrm{ind}}{\sim}}

\providecommand{\argmin}{\mathop\mathrm{arg min}}

\def\LS{{\sf LS}}

\def\naturals{{\mathbb N}}
\def\integers{{\mathbb Z}}
\def\reals{{\mathbb R}}

\def\eps{{\varepsilon}}
\def\prob{{\mathbb P}}
\def\E{{\mathbb E}}
\def\Var{{\rm Var}}

\def\VAR{{\sf VAR}}

\def\var{{\rm Var}}

\def\L0{{L_i}}

\def\d{{\rm d}}
\def\<{\langle}
\def\>{\rangle}

\def\hth{{\widehat{\theta}}}
\def\dth{{\widehat{\theta}^{{\sf d}}}}
\def\offth{{\widehat{\theta}^{{\sf off}}}}  
\def\onth{{\widehat{\theta}^{{\sf on}}}}

\def\hSigma{\widehat{\Sigma}}

\def\hsigma{\widehat{\sigma}}

\def\supp{{\rm supp}}

\def\F{{\sf F}}
\def\ind{{\mathbb I}}

\def\F{{\sf F}}
\def\normal{{\sf N}}

\def\P{{\mathbb{P}}}

\def\sT{{\sf T}}

\def\id{{\rm I}}

\def\RE{{\rm RE}}

\def\bvsigma{{\bar{\varsigma}}}

\def\sign{{\rm sign}}

\def\v*{v_i}
\def\T*{T_i}

\def\u*{u_i}
\def\F*{F_i}

\def\mumin{{\mu_{\rm min}}}
\def\mumax{{\mu_{\rm max}}}

\def\tA{{\widetilde{A}}}

\def\htheta{\widehat{\theta}}

\def\som{{s_\Omega}}

\def\cL{\mathcal{L}}

\def\th{{\rm th}}
\def\hth{{\widehat{\theta}}}

\newcommand\norm[1]{\lVert{#1}\rVert}
\newcommand\abs[1]{\lvert{#1}\rvert}

\def\l1u{W}

\def\ols{\widehat{\theta}^{{\sf OLS}}}

\newcommand{\ajcomment}[1]{}

\makeatletter
\newcommand{\labitem}[2]{%
\def\@itemlabel{\text{#1}}
\item
\def\@currentlabel{#1}\label{#2}}
\makeatother

\DeclareMathAlphabet{\mathpzc}{OT1}{pzc}{m}{it}
\def\mya{\tau}

\title{Online Debiasing for Adaptively Collected High-dimensional Data with Applications to Time Series Analysis}
\author{ 
Yash Deshpande\thanks{Institute for Data, Systems and Society, Massachusetts Institute of Technology, Email: \url{yash@mit.edu}}\and
Adel Javanmard\thanks{Data Sciences and Operations Department, Marshall School of Business, University
of Southern California, Email: \url{ajavanma@usc.edu}}\and
Mohammad Mehrabi\thanks{Data Sciences and Operations Department, Marshall School of Business, University
of Southern California, Email: \url{mehrabim@usc.edu}} \thanks{The names of the authors are in alphabetical order. } 
}

\begin{document}
\maketitle
\begin{abstract}
Adaptive collection of data is 
commonplace in applications
throughout science and engineering.
From the point of view of statistical
inference however, adaptive data collection
induces memory and correlation in the samples,
and poses significant challenge. 

We consider
the high-dimensional linear regression,
where the samples are collected adaptively, and the sample
size $n$ can be 
smaller than $p$, the number of covariates.
In this setting, there are 
two distinct sources of bias:
the first due to regularization imposed
for consistent estimation, e.g. using the LASSO, and the second
due to adaptivity in collecting the samples. 
We propose 
\emph{`online debiasing'}, a general procedure for
estimators such as the LASSO, which
addresses both 
sources of bias. In two concrete contexts
 $(i)$  time series analysis and
 $(ii)$ batched data collection, we demonstrate
 that online debiasing optimally
 debiases the LASSO estimate
 when the underlying
parameter $\theta_0$ has sparsity
 of order $o(\sqrt{n}/\log p)$. In this
 regime, the debiased estimator can be 
 used to compute $p$-values and
  confidence intervals of optimal size.
\end{abstract}

\section{Introduction}\label{sec:intro}
Modern data collection, experimentation 
and modeling
are often adaptive in nature. For example, 
clinical trials are run in phases, wherein the data
from a previous phase
inform and influence the design of future phases. In
commercial recommendation engines, 
algorithms collect data by eliciting feedback from
their users; data which is ultimately used to 
improve the algorithms underlying the recommendations and so influence the future data. 
In such applications, adaptive data collection is
often carried out for objectives correlated
to, but 
distinct from statistical inference.
In clinical trials, an ethical 
experimenter might prefer to assign more patients a 
treatment that they might benefit from, 
instead of the control treatment. 
In e-commerce, recommendation engines aim to minimize the revenue loss. 
In other applications, collecting data is potentially costly, and 
practitioners may choose to collect samples
that are a priori deemed most informative. 
Since such objectives are intimately related to 
statistical estimation, it is not surprising that
adaptively collected data can be used to derive
statistically consistent estimates, often using
standard estimators. The question of statistical inference however, is
more subtle: on the one hand, consistent estimation
indicates that the collected samples are informative enough.
On the other hand, adaptive collection
induces endogenous correlation in the samples, resulting 
in bias in the estimates. In this paper, 
we address the following natural question raised
by this dichotomy:
\vspace{0.1cm}
\begin{center}
 \emph{Can adaptively collected data be used for ex post statistical inference?}
\end{center}
\vspace{0.1cm}
We will focus on the linear model, where the samples
$(y_1, x_1), (y_2, x_2), \dots$, $(y_n, x_n)$ satisfy:
\begin{align}
y_i &= \<x_i, \theta_0\> + \eps_i, \quad \eps_i \iid\normal(0, \sigma^2). 
\label{eq:linmodel}
\end{align}
Here $\theta_0 \in \reals^p$ is an unknown parameter vector
relating the covariates $x_i$ to the response $y_i$, and
the noise $\eps_i$ are i.i.d. $\normal(0, \sigma^2)$ 
random variables. 
In vector form, we write Eq.\eqref{eq:linmodel} as
\begin{align}
y &= X \theta_0 + \eps,
\end{align}
where $y = (y_1, y_2, \dots, y_n)$,
$\eps = (\eps_1, \eps_2, \dots, \eps_n)$ and the design matrix 
 $X \in\reals^{n\times p}$ has rows $x_1^\sT, \dots ,x_n^\sT$.
  When the samples are adaptively collected, the data point $(y_i, x_i)$ is
obtained \emph{after viewing the previous data points} 
$(y_1, x_1), \dots$, $(y_{i-1}, x_{i-1})$\footnote{Formally, we
assume
 a filtration $(\fF_i)_{i\le n}$ to which the sequence
$(y_i, x_i)_{i\le n}$ is adapted, and with respect to 
which the sequence $(x_i)_{i\le n}$ is predictable}. 

In the `sample-rich' regime when $p< n$, the standard approach
would be to compute the least squares estimate
$\htheta^\LS = (X^\sT X)^{-1}X^\sT y$, and assess
the uncertainty in $\htheta^\LS$ using a
central limit approximation $ (X^\sT X)^{1/2} (\htheta^\LS -\theta_0) 
\approx \normal(0, \id_p)$ \cite{lai1982least}.  
 However, while the estimator $\htheta^\LS$ 
 is consistent under fairly weak conditions, 
 adaptive data collection 
complicates the task of characterizing its distribution. 
One hint for this is the observation that, in stark contrast 
with the non-adaptive setting, 
$\htheta^\LS = \theta_0 + (X^\sT X)^{-1}X^\sT \eps$ is in general \emph{a biased estimate} of $\theta_0$. 
Adaptive data collection creates correlation between
the responses $y_i$ (therefore $\eps_i$) and covariate
vectors  $x_{i+1}, x_{i+2}, \dots, x_n$ observed in the future. 
In the context of multi-armed bandits, where the 
estimator $\htheta^\LS$ for model \eqref{eq:linmodel} 
reduces to sample averages,   
\cite{xu2013estimation,villar2015multi} observed such bias empirically, and \cite{nie2017why,shin2019bias} characterized and
developed upper bounds on the bias. 
While bias is an important problem, 
estimates may 
also show higher-order
distributional defects that 
complicate inferential tasks. 

This phenomenon is exacerbated in the high-dimensional
or `feature-rich' regime when $p > n$. Here the design
matrix $X$ becomes rank-deficient, and 
consistent parameter estimation requires $(i)$ additional structural 
assumptions on $\theta_0$ and $(ii)$ regularized estimators beyond $\htheta^\LS$,
such as the LASSO \cite{Tibs96}.  
 Such estimators are non-linear, non-explicit
 and, consequently it is difficult
to characterize their distribution
even with strong random design assumptions 
\cite{BayatiMontanariLASSO,javanmard2014hypothesis}.
In analogy to the low-dimensional regime, it is
relatively easier to develop 
consistency guarantees for estimation using the LASSO
 when $p > n$. Given the sample 
  $(y_1, x_1), \dots (y_n, x_n)$ one can compute
the LASSO estimate $\htheta^\sL = \htheta^\sL(y, X; \lambda_n)$
\begin{align}
\htheta^\sL &= \arg\min_{\theta} \Big\{ \frac{1}{2n} \norm{y - X\theta}_2^2 + \lambda_n \norm{\theta}_1 \Big\}  
,  \label{eq:lasso}
\end{align}
If $\theta_0$ is sparse with at most $s_0\ll p$
non-zero entries and the design $X$ satisfies
some technical conditions, the LASSO estimate, for an appropriate
choice of $\lambda_n$ has estimation error 
 $\|\htheta^\sL - \theta_0\|^2_2$ of order $\sigma^2 s_0 (\log p)/n$, with high probability
 \cite{basu2015regularized,bastani2015online}. 
 In particular the estimate is consistent provided the sparsity
 satisfies $s_0 = o( n/ \log p )$.
 This estimator is biased though because of two
 distinct reasons. The first is the regularization imposed 
 in Eq.\eqref{eq:lasso}, which disposes $\htheta^\sL$ 
 to have small $\ell_1$ norm. The second is the correlation induced between
 $X$ and $\eps$ due to adaptive data collection. 
To address the first source,  
\cite{ZhangZhangSignificance,javanmard2014confidence,van2014asymptotically}
proposed  a so-called ``\emph{debiased estimate}" of the form
\begin{align}
\offth &= \htheta^\sL + \frac{1}{n} M X^\sT (y - X\htheta^\sL), 
\label{eq:classicaldebias}
\end{align}
where $M$ is chosen as an `approximate inverse' of 
the sample covariance $\hSigma = X^\sT X/n$. The intuition
for this idea is the following decomposition that
follows directly from Eqs.\eqref{eq:linmodel}, \eqref{eq:classicaldebias}:\footnote{{The notation $\offth$ stands for ``offline" debiasing. We use this notation/terminology to highlight its main difference from the ``online" debiasing that will be introduced later in this paper.}}
\begin{align}
\offth- \theta_0 = (I_p -  M \hSigma ) (\htheta^\sL - \theta_0) + \frac{1}{n} M X^\sT \eps. \label{eq:classicaldecomp}
\end{align}
When the data collection is non-adaptive, $X$ and $\eps$ are independent and therefore, conditional on the design $X$,  $M X^\sT \eps /n$ is
distributed as $\normal(0, \sigma^2 Q / n)$
where $Q = M\hSigma M^\sT$. Further,
 the bias in $\offth$ is isolated to
the first term, which intuitively should be of smaller
order than the second term, provided both $\htheta^\sL - \theta_0$ and
 $M \hSigma - \id_p$ are small in an appropriate sense.  
 This intuition suggests that, if the second
term
dominates
the first term in $\offth$, we can produce confidence intervals
for $\theta_0$ 
in the usual fashion using the debiased estimate $\offth$ \cite{javanmard2014confidence,javanmard2014hypothesis,van2014asymptotically}.
For instance, with $Q = M\hSigma M^\sT$, the interval  $\big[\offth_1 - 1.96 \sigma \sqrt{Q_{11}/n}, \offth_1 + 1.96 \sigma \sqrt{ Q_{11}/n}\big]$ 
forms a standard $95\%$ confidence interval for the parameter
$\theta_{0, 1}$.
In the so-called `random design' setting --when the rows of $X$ are 
drawn i.i.d. 
from a broad class of distributions--
this approach to 
inference via the debiased estimate 
$\offth$ enjoys several optimality
guarantees: the resulting
confidence intervals have minimax optimal
size \cite{javanmard2014inference,javanmard2014confidence,cai2017confidence}, 
and are semi-parametrically efficient \cite{van2014asymptotically}.

\emph{This line of argument breaks down
when the samples are adaptively
collected, as the debiased estimate $\offth$ still 
suffers the second source of bias.} Indeed, this is
exactly analogous to 
 $\htheta^\LS$ in low dimensions.
Since $M$, $X$ and the noise $\eps$ are 
correlated, we can no longer assert that
the term $MX^\sT \eps/n$ is unbiased. Indeed,  
characterizing its distribution 
can be quite difficult, given the
intricate correlation between $M$, $X$ and $\eps$ 
induced by the data
collecting policy and the procedure for choosing $M$.
We illustrate the failure of offline debiasing
in two scenarios of interest in this paper:
$(i)$ batched data collection and
$(ii)$ autoregressive time series. 

\subsection{Why offline debiasing fails?} \label{sec:offlinefailure}

\subsubsection*{Batched data collection} 
\label{sec:example1}
Consider a stylized 
model of adaptive data collection
wherein the experimenter (or analyst)
collects data in two phases or batches. 
 In the first phase, 
the experimenter collects an initial set of samples 
$(y_1, x_1), \dots, (y_{n_1}, x_{n_1}) $
of size $n_1 < n$
 where
the responses follow Eq.\eqref{eq:linmodel} and 
the covariates are i.i.d. from a distribution 
$\P_x$. Following this, she computes an intermediate 
estimate $\htheta^1$ of $\theta_0$ and then
collects additional samples $(y_{n_1+1}, x_{n_1+1}), \dots , (y_{n}, x_n)$ of size $n_2 = n - n_1$, where the covariates $x_i$
are drawn independently from the law of $x_1$, conditional on the event $\{\<x_1, \htheta^1\> \ge \varsigma\}$, where $\varsigma$
is a threshold, that may be data-dependent.  
This is a typical scenario where the
response $y_i$
represents an instantaneous reward that the experimenter
wishes to maximize, as in multi-armed bandits
\cite{lai1985asymptotically,bubeck2012regret}. 
   For instance, clinical trials may be designed
to be response-adaptive and allocate patients to treatments
that they are likely to benefit from based on prior
data \cite{zhou2008bayesian,kim2011battle}. The multi-armed
bandit problem is a standard formalization of this 
trade-off, and a variety
of bandit algorithms are designed to 
operate in distinct phases of `explore--then exploit'\cite{rusmevichientong2010linearly,deshpande2012linear,bastani2015online,perchet2016batched}. The 
model we describe above is a close approximation of data collected
from one arm in a run of such an algorithm. 
With the full samples $(y_1, x_1), \dots, (y_n, x_n)$ 
at hand, the experimenter would like to 
perform inference on a fixed coordinate $\theta_{0, a}$
of the underlying parameter.  

As a numerical example, we 
consider $\theta_0 \in\{0, 1\}^{600}$ with exactly $s_0 = 10$
non-zero entries. We obtain the first batch
$(y_1, x_1), \dots, (y_{500}, x_{500} )$ of observations
with $y_i = \<x_i, \theta_0\> + \eps_i$, 
$x_i \iid \normal(0, \Sigma)$ and $\eps_i\iid\normal(0, 1)$
where we use the covariance $\Sigma$ as below:
\begin{align*}
\Sigma_{a, b} &= \begin{cases}
1 &\text{ if } a = b, \\
0.1 &\text{ if } \lvert a-b\rvert = 1\\
0 &\text{ otherwise.}
\end{cases} 
\end{align*}
Based on this data, we construct an intermediate
estimator $\hth^1$  on $(y^{(1)}, X_1)$ using two different
strategies: 
$(i)$ debiased LASSO and $(ii)$ ridge regression with cross-validation. With
this estimate we now sample new covariates
$x_{501}, \dots, x_{1000}$ independently from the law of
$x\vert_{\<x, \hth^1\> \ge \<\hth^1, \Sigma \hth^1\>^{1/2}}$ and the corresponding
outcomes $y_{501}, \dots, y_{1000}$ are generated according to Eq.\eqref{eq:linmodel}. Unconditionally,
$\<x, \hth^1\>\sim\normal(0, \<\hth^1, \Sigma\hth^1\>)$, so this
choice of threshold corresponds to sampling
covariates that correlate with $\hth^1$ at least one standard
deviation higher than expected unconditionally. This procedure yields two batches of data, each of $n_1=n_2 = 500$
data points, combining to a set of $1000$ samples. 

From the full dataset $(y_1, x_1), \dots, (y_{1000}, x_{1000})$
we compute the LASSO estimate $\Lsth = \Lsth(y, X; \lambda)$
with $\lambda = 2.5\lambdamax(\Sigma)\sqrt{(\log p)/n}$. Offline debiasing
yields the following prescription to debias
$\Lsth$:
\begin{align*}
\offth&= \Lsth + \frac{1}{n}\Omega(\hth^1)X^\sT (y - X\Lsth), 
\end{align*}
where $\Omega(\hth)$ is the population precision matrix:
\begin{align*}
\Omega(\hth^1)^{-1} &= \frac{1}{2}\E\{x x^\sT\} + \frac{1}{2}{\E\left\{x x^\sT \Big\vert \<x, \hth^{\,1}\>\ge \|\Sigma^{1/2}\hth^{\,1}\|\right\}}\,.
\end{align*}
We generate the dataset for $100$ Monte Carlo iterations and compute
the offline debiased estimate $\offth$ for each
iteration. 
Figure \ref{fig:batchexampleonlineofflinedebiasedlasso0}  shows
the 
histogram of the entries $\offth$ on the support of $\theta_0$  for the 
two choices of $\hth^1$. As we see $\offth$ still has considerable bias,
due to adaptivity in the data collection.

\begin{figure}[]
\centering
\begin{subfigure}{.45\linewidth}
\centering
\includegraphics[scale=0.5]{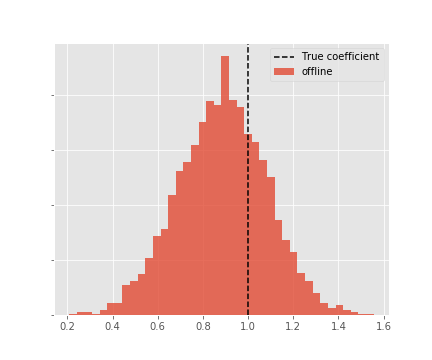}
\caption{with $\hth^1$ the debiased LASSO on first batch}
\end{subfigure}
\begin{subfigure}{0.5\linewidth}
\centering
\includegraphics[scale =0.5]{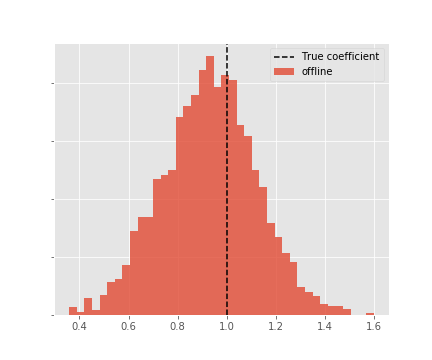}
\caption{with $\hth^1$ the ridge estimate on first batch}
\end{subfigure}
\caption{{\small Histograms of the offline debiased estimate
$\offth$ restricted to the support
of $\theta_0$. The dashed line indicates the true coefficient size. Recall that the second batch is chosen based on an intermediate estimator $\hth^1$ computed on the first batch. (Left) $\hth^1$ is debiased LASSO on the first batch,
(Right) $\hth^1$ is ridge estimate on the first batch. As we observe the offline debiasing (even with access to the precision matrix $\Omega$ of the random designs) has a significant bias and dose not admit  a Gaussian distribution.}  \label{fig:batchexampleonlineofflinedebiasedlasso0} }
\end{figure}
\subsubsection*{Autoregressive time series}\label{sec:example2}
A vector autoregressive ($\VAR$) time series model 
posits that data points $z_t$ 
evolve according to the dynamics:
\begin{align}
z_{t} &= \sum_{\ell=1}^d A^{(\ell)} z_{t-\ell} + \zeta_t
\end{align}
where $A^{(\ell)} \in \reals^{p\times p}$ are time invariant coefficients and $\zeta_t$ is the noise term satisfying $\E(\zeta_t)=0$ (zero-mean), $\E(\zeta_t\zeta_t^\sT) = \Sigma_\zeta$ (stationary covariance), and $\E(\zeta_t \zeta_{t-k}^\sT) =0$ for $k>0$ (no serial correlation). 
Given the data $z_1,\dotsc, z_T$, the  task of interest is to perform statistical inference on the model parameters, i.e., coefficient matrices $A^{(1)}, \dotsc, A^{(d)}$. Clearly, the samples $z_t$ are `adaptively collected', in the sense
that there is serial correlation in the samples. Indeed, the data point $z_t$ depends on the previous data points $z_{t-1}, z_{t-2}, \dots, z_1$.

As in the batched data example, we will carry out a simple illustration. 
We generate data from a $\VAR(d)$ model with $p=15$, $d=5$, $T=60$, and diagonal $A^{(\ell)}$ matrices with value $b=0.15$ on their diagonals. We also generate $\zeta_t \iid \normal(0, \Sigma_\zeta)$. Note that this is a high-dimensional setting as the number of parameters $dp^2$ exceeds the sample size $(T-d)p$. 
We keep the covariance of the noise terms $\zeta_t$ as below:
\begin{align*}
 \Sigma_{\zeta, ij}=0.5^{\ind(i\neq j)} 
\end{align*}
To estimate the parameters, we define the covariate
vectors $x_t= (z_{t+d-1}^\sT,\dots,z_t^\sT)^\sT\in \reals^{dp}$,
 obtained by concatenating $d$ consecutive data points and $\eps = (\zeta_{d+1, i}, \zeta_{d+2, i} \dots, \zeta_{T, i})$.
We focus on the noise component of the offline debiased estimate, i.e.,
\begin{align}
W^{{\rm off}} &= \frac{1}{\sqrt{n}} M \sum_{t=1}^n x_t \eps_t\,,\
\end{align}
with $M$ denoting the decorrelating matrix in the debiased estimate as per~\eqref{eq:classicaldebias}.

In Figure \ref{fig:fixed-coord}, we show the QQ-plot, PP-plot and histogram of $W^{{\rm off}}_1$  (corresponding to the entry $(1,1)$ of matrix $A_1$) for 1000 different realizations of the noise $\zeta_t$.  As we observe, even the noise component $W^{{\rm off}}$ is biased because the offline construction of $M$ depends on all features $x_t$ and hence endogenous noise $\zeta_t$. Recall that for the setting with an i.i.d sample, the noise component is zero mean gaussian for any finite sample size $n$. This further highlights the challenge of high-dimensional statistical inference with adaptively collected samples and demonstrate why the classical debiasing approach will not work in this case.  


\begin{figure}[]
\begin{subfigure}{.33\linewidth}
\centering
\includegraphics[scale=0.20]{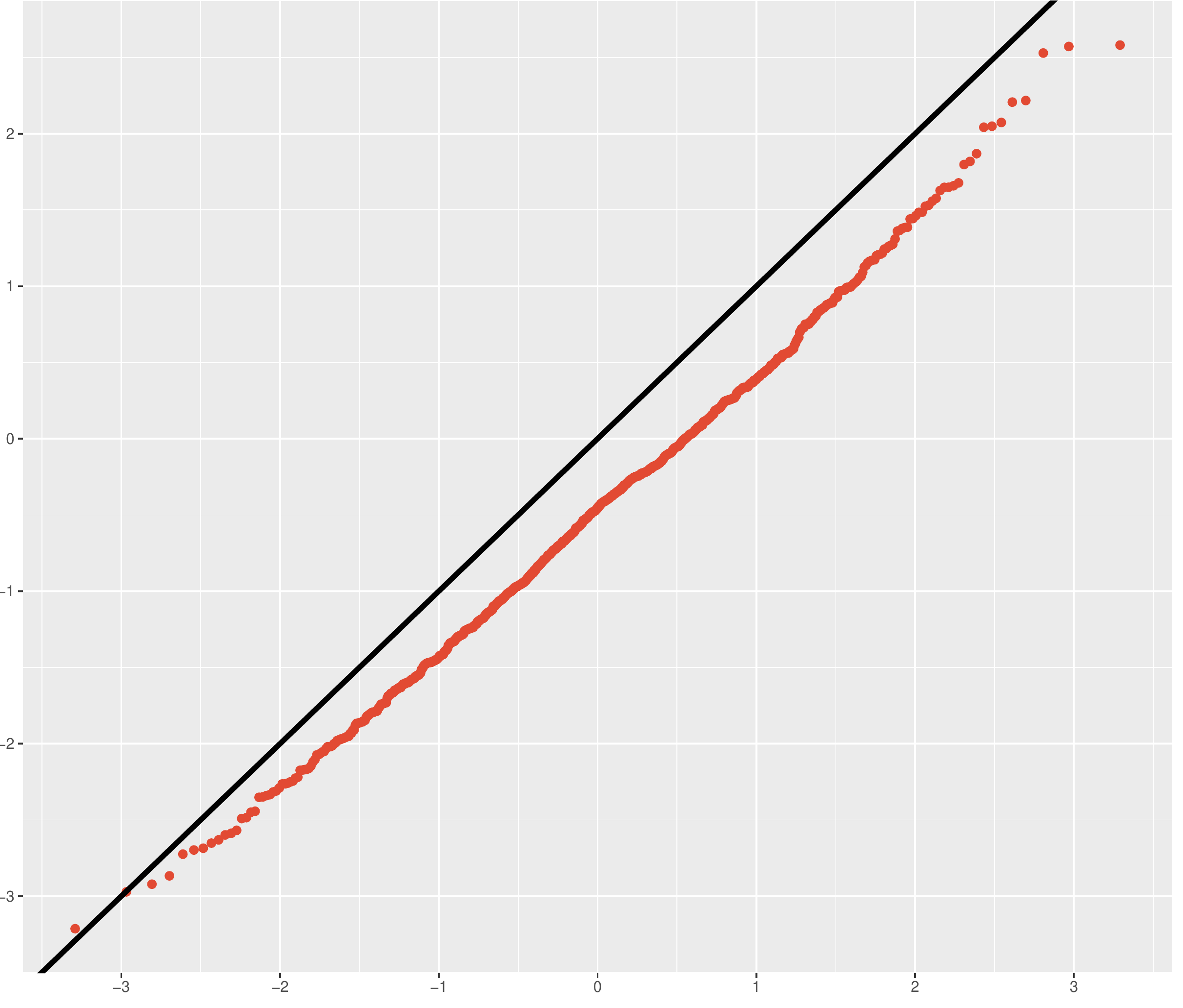}
\put(-150,55){\rotatebox{90}{\scriptsize{Sample}}}
\put(-93,-5){\rotatebox{0}{\scriptsize{Theoretical}}}
\caption{}
\label{fig:fixed-coord:sub1-0}
\end{subfigure}%
\begin{subfigure}{.33\linewidth}
\centering
\includegraphics[scale=0.20]{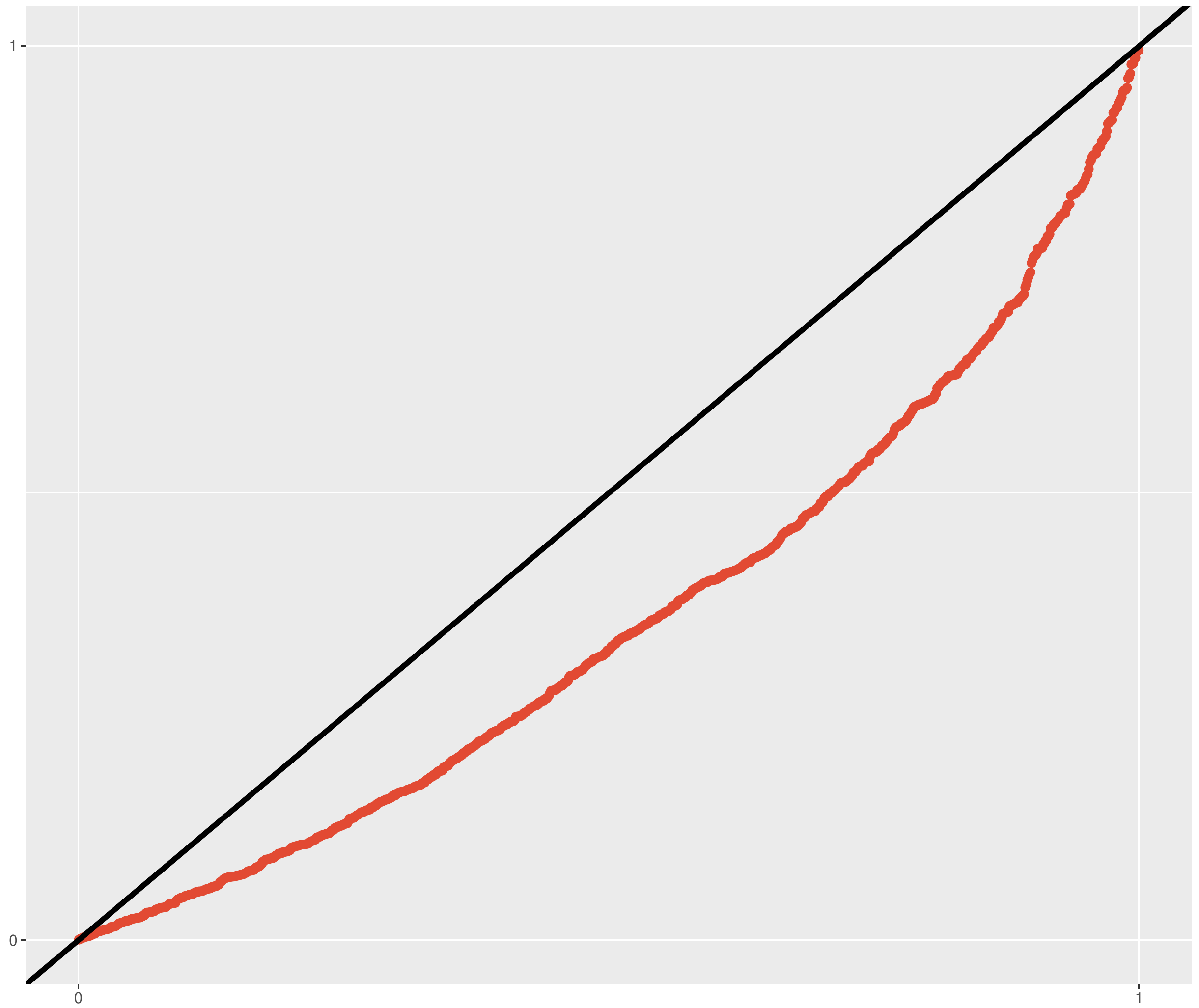}
\put(-150,55){\rotatebox{90}{\scriptsize{Sample}}}
\put(-93,-5){\rotatebox{0}{\scriptsize{Theoretical}}}
\caption{}
\label{fig:fixed-coord:sub2-0}
\end{subfigure}
\begin{subfigure}{0.33\linewidth}
\centering
\includegraphics[scale=0.20]{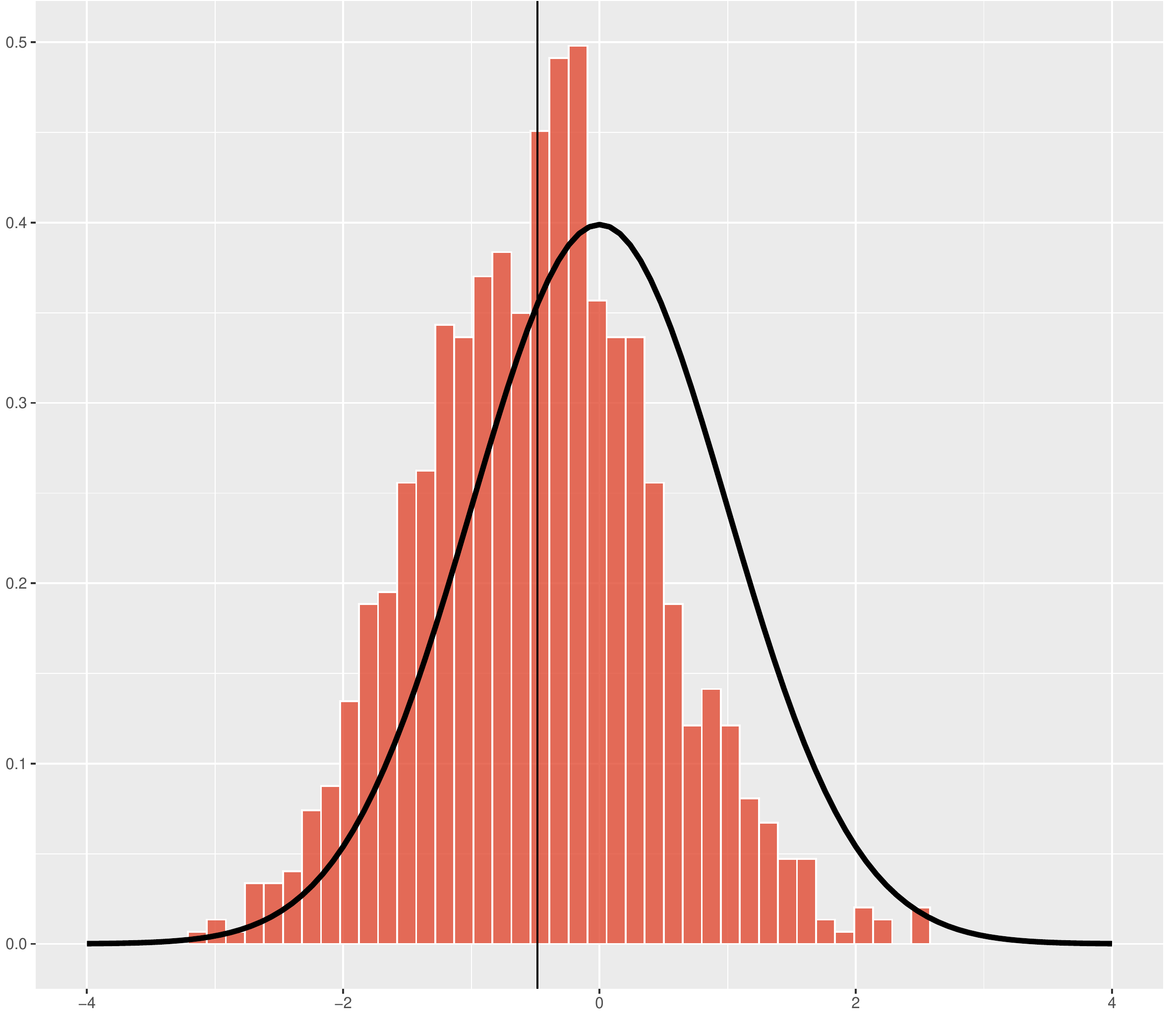}
\put(-93,-5){\rotatebox{0}{\scriptsize{Noise Terms}}}
\put(-147,55){\rotatebox{90}{\scriptsize{Density}}}
\caption{}
\label{fig:fixed-coord:sub3-0}
\end{subfigure}
\caption{{\small Empirical behavior of noise term associated with the offline debiased estimate of a fixed coordinate of Gaussian $\VAR(d)$ model. In this example, $d=5, p=15, T=60, \rho=0.5$, $\Sigma_{\zeta}(i,j)=\rho^{|i-j|}$, and $A^{(\ell)}$ matrices are diagonal with value $b=0.15$ on their diagonals. Plots \ref{fig:fixed-coord:sub1-0}, \ref{fig:fixed-coord:sub2-0}, and \ref{fig:fixed-coord:sub3-0} show the QQ plot, PP plot, and the histogram of the offline debiased noise terms (red) over 1000 independent experiments, respectively and black curve/lines denote the ideal standard normal distribution. As we observe, even the noise component of the offline debiased estimator deviates from the standard normal distribution; This implies the failure of offline debiasing method for statistical inference purposes when the samples are correlated. The vertical black line in (c) indicates the mean of the noise component of the offline debiased estimator.}} 
\label{fig:fixed-coord}
\end{figure}
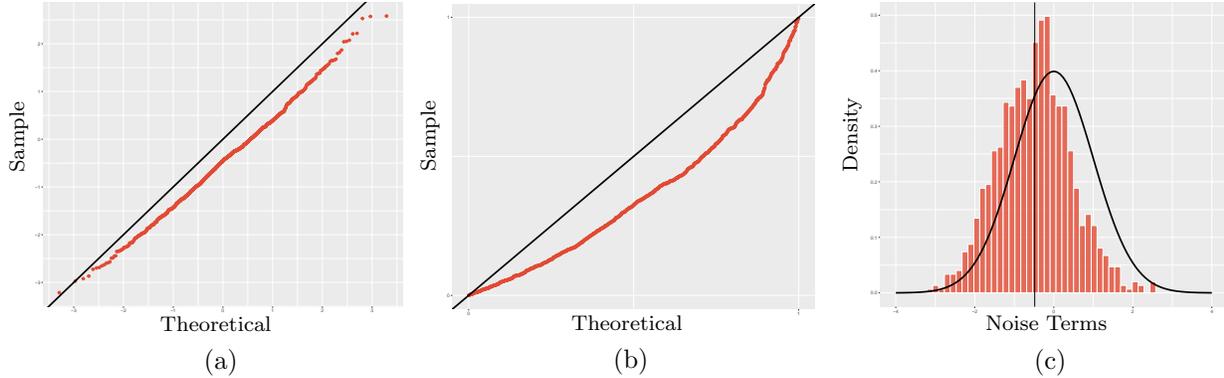



\section{Online debiasing}

We propose 
\emph{online debiased} estimator $\onth = \onth(y, X; (M_i)_{i\le n}, \lambda)$ that takes the form
\begin{align}
\onth &\equiv \htheta^\sL + \frac{1}{n} \sum_{i=1}^n M_i x_i (y_i - x_i^\sT \htheta^\sL). \label{eq:onlinedebias}
\end{align}
The term `online' comes from the first crucial constraint of \emph{predictability} imposed on the sequence $(M_i)_{i\le n}$. 
\begin{definition}[Predictability]\label{def:pred}
Without loss of generality, 
there exists a filtration $(\fF_i)_{i\ge 0}$
so that, for $i =1, 2, \dots, n$, 
$(i)$ $\eps_i$ are adapted to $\fF_i$ and
$\eps_i$ is independent of $\fF_{j}$ for $j < i$. 
We assume that the sequences $(x_i)_{i\ge 1}$
and $(M_i)_{i\ge 1}$ are \emph{predictable}
with respect to $\fF_i$, i.e. 
for each $i$, $x_i$ and $M_i$ are measurable with respect
to $\fF_{i-1}$.
\end{definition}
With predictability, 
the data points $(y_i, x_i)$ are adapted
to the filtration $(\fF_i)_{i\le n}$ and, moreover, the covariates $x_i$ are predictable
with respect to $\fF_i$. Intuitively, the $\sigma$-algebra
$\fF_i$ contains all information in the data, as well
as potential external randomness, that is used to 
query the new data covariate $x_{i+1}$. Predictability
ensures that only this information may be used
to construct the matrix $M_{i+1}$. 
Analogous to  Eq.\eqref{eq:classicaldecomp} we can 
decompose $\onth$ into two components:
\begin{align}
\onth &= \theta_0 + \frac{1}{\sqrt{n}}\big(B_n (\htheta^\sL - \theta_0) + W_n \big) \\
\text{ where } B_n &\equiv \sqrt{n} \Big(  I_p - \frac{1}{n}\sum_i M_i x_i x_i^\sT \Big), \nonumber \\
\text{ and } W_n &\equiv \frac{1}{\sqrt{n}}\sum_i M_i x_i \eps_i \nonumber. 
\end{align}
Predictability of $(M_i)_{i\le n}$ ensures that
$W_n$ is unbiased and  the bias in
$\onth$ is contained entirely in the first
term $B_n (\htheta^\sL - \theta_0)$. Suppose that, 
analogous to offline debiasing, we prove that
the bias term $B_n(\hth^\sL - \theta_0)$ is of
smaller order than the variance term $W_n$. 
We are then left with the problem of characterizing
the asymptotic distribution of 
the sequence $W_n$. As the sequence $\sqrt{n}W_n = \sum_i M_i x_i \eps_i$ is \emph{a martingale} with respect to the filtration ${\fF_i}$, one might expect that $W_n$ is asymptotically
Gaussian. The following `stability' property, identified first by
Lai and Wei \cite{lai1982least} in this context, is crucial to ensure that this intuition
is correct. 

\begin{definition}[Stability]\label{def:stability}
Consider a square integrable triangular martingale array $\{Z_{i,n}\}_{i\le n, n\ge 1}$ adapted to a filtration
$\fF_i$ and its quadratic variation $V_n = \sum_{i\le n} \E\{(Z_{i,n} - Z_{i-1, n})^2 \vert \fF_{i-1}\}$. Note that $V_n$ is non-negative random variable, 
measurable with respect to $\fF_{n-1}$. We say that the martingale array $\{Z_{i, n}\}_{i\ge 1}$ is stable if there exists a constant $v_\infty > 0 $ where $ \lim_{n\to \infty} V_n = v_\infty$ in probability.  
\end{definition}

An important contribution of our paper is to develop online
debiasing estimators $\onth$ whose underlying martingales
are stable. The specifics of construction of predictable sequence $(M_i)_{i\le n}$ and deriving the distributional characterization of the debiased estimator $\onth$ depend on the context of the problem at hand. In this paper, we instantiate this idea in two concrete contexts: $(i)$  time series analysis (Section \ref{sec:timeseries}) and $(ii)$ batched data collection (Section 
\ref{sec:batch}). For both of these settings, 

\begin{enumerate}
	\item We first establish estimation
results for the LASSO estimate, showing that
even with adaptive data collection, the LASSO estimate
enjoys good estimation error (Theorems \ref{propo:estimation} and \ref{thm:batchlassoerr}). These results draw significantly
on prior work in high-dimensional estimation 
\cite{basu2015regularized,buhlmann2011statistics}.  

\item Next, we propose constructions for the 
online debiasing sequence $(M_i)_{i\le n}$,
using an optimization program that trades off variance
with bias, \emph{while ensuring stability}. This optimization program is a novel modification of the approximate inverse construction in \cite{javanmard2014confidence}.
The important change is the inclusion of an $\ell_1$ constraint in the
program, 
which ensures stability of the underlying martingales, and allows the use of a martingale CLT theorem to characterize the distribution of the online debiased estimator. 

\item We establish a distributional characterization
of the resulting online debiased estimate $\onth$ (Theorems \ref{pro:SS} and \ref{thm:batchdistchar}).  
Informally, this demonstrates that coordinates of $\onth$ are approximately
Gaussian with a covariance computable from data. 
\end{enumerate}

 In Section \ref{sec:inference}, we demonstrate how the online debiased estimate $\onth$
can be used to compute standard inferential primitives like
confidence intervals and p-values. 
Section \ref{sec:numerical} contains numerical experiments
that demonstrate the validity our proposals on both synthetic and real data. 
In Section \ref{sec:discussion} we develop computationally efficient
iterative descent methods to construct the online debiasing
sequence $(M_i)_{i\le n}$. In the interest of reproducibility, we make an {{\sf R}} implementation of our algorithm publicly available at {\small \url{http://faculty.marshall.usc.edu/Adel-Javanmard/OnlineDebiasing}}.  

Our proposal of online debiasing approach builds on the insight in 
\cite{deshpande2018accurate}, which has studied a similar problem for low-dimensional settings ($p<n$). We provide a detailed discussion of this this work in Section~\ref{sec:comparison}, highlighting the main distinctions and the inefficacy of that method for high-dimensional setting to further motivate our work and contributions.


\paragraph{Notation} 
Henceforth, we use the shorthand $[p] \equiv \{1,\dotsc, p\}$ for an integer $p\ge 1$, and $a\wedge b \equiv \min(a,b)$, $a\vee b \equiv \max(a,b)$. We also indicate the matrices in upper case letters and use lower case letters for vectors and scalars. 
We write $\|v\|_p$ for the standard $\ell_p$ norm of a vector $v$, $\|v\|_p = (\sum_i |v_i|^p)^{1/p}$ and $\|v\|_0$ for the number of nonzero elements of $v$. We also denote by $\supp(v)$, the support of $v$ that is the positions of its nonzero entries. For a matrix $A$, $\|A\|_p$ represents its $\ell_p$ operator norm 
and $\|A\|_\infty = \max_{i,j}|A_{ij}|$ denotes the maximum absolute value of its entries.  In particular, $\|A\|_1$ is the $\ell_1-\ell_1$ norm of matrix $A$ (the maximum $\ell_1$ norm of its columns). For two matrices $A$, $B$, we use the shorthand $\<A,B\> \equiv {\rm trace}(A^\sT B)$. 
In addition $\phi(x)$ and $\Phi(x)$ respectively represents the probability density function and the cumulative distribution function of standard normal variable. Also, we use the term \emph{with high probability} to imply that the probability converges to one as $n\to \infty$.


\section{Online debiasing for high-dimensional time series}
\label{sec:timeseries}

\def\Sigmazeta{{\Sigma_\zeta}}

The Gaussian \emph{vector autoregressive model}   of order $d$ 
(or $\VAR(d)$ for short)  \cite{shumway2006time}, 
posits that data points $z_t$ 
follow the dynamics:
\begin{align}
z_{t} &= \sum_{\ell=1}^d A^{(\ell)} z_{t-\ell} + \zeta_t,\label{eq:varddef}
\end{align}
where $A^{(\ell)} \in \reals^{p\times p}$ and $\zeta_t \iid \normal(0, \Sigmazeta)$. $\VAR$ models are extensively used across science and 
engineering
(see \cite{fujita2007modeling,stock2001vector,holtz1988estimating,seth2015granger} for notable examples in macroeconomics, genomics and neuroscience). Given the data $z_1, \dots, z_T$, the fundamental
task is
to estimate the parameters of the $\VAR$
model, viz. the matrices $A^{(1)}, \dots A^{(d)}$. 
The estimates of the parameters can be used in a variety
of ways depending on the context: to detect or test for 
stationarity, forecast future data, or suggest causal links.
Since each matrix is $p\times p$, this
forms a putative total of $dp^2$ parameters, 
which we estimate from a total of $(T-d)p$ linear
equations (Eq.\eqref{eq:varddef} with $t = d+1, \dots, T$). 
For the $i^\th$ coordinate of $z_t$, Eq.\eqref{eq:varddef}
reads
\begin{align}
z_{t, i} &= \sum_{\ell=1}^d \<z_{t-\ell}, A^{(\ell)}_{i}\> + \zeta_{t, i},
\label{eq:varrowwise}
\end{align}
where $A^{(\ell)}_i$ denotes the $i^\th$ row of the matrix $A^{(\ell)}$.
This can be interpreted in the linear regression form
Eq.\eqref{eq:linmodel} in dimension $dp$ with $\theta_0 \in \reals^{dp}, X\in \reals^{(T-d) \times dp} $, $y, \eps \in \reals^{T-d}$ identified
as:
\begin{align}
\label{eq:Y-X}
\theta_0 &= (A^{(1)}_i, A^{(2)}_i, \dots, A^{(d)}_i)^\sT, \nonumber \\
X &= \begin{bmatrix}
z_d^\sT  &z_{d-1}^\sT &\dots &z_1^\sT  \\
z_{d+1}^\sT & z_{d}^\sT &\dots &z_2^\sT\\
\vdots & \vdots &\ddots &\vdots  \\
z_{T-1}^\sT & z_{T-2}^\sT &\dots &z_{T-d}^\sT
\end{bmatrix}, \nonumber \\
y &= (z_{d+1, i}, z_{d+2, i}, \dots, z_{T, i}), \nonumber \\
\eps &= (\zeta_{d+1, i}, \zeta_{d+2, i}, \dots, \zeta_{T, i}). 
\end{align}
We omit the dependence on the coordinate $i$, and also denote the rows of $X$ by $x_1, \dotsc, x_n\in \reals^{dp}$, with $n = T-d$.
Given sufficient data, or when $T$ is large in comparison
with $dp$, it is possible to estimate the parameters
using least squares \cite{shumway2006time,lai1982least}. 
In \cite{basu2015regularized}, Basu and Michailidis
consider the problem of estimating the parameters when 
number of time points $T$ is small in
comparison with the total number of parameters $dp$, 
with the proviso that 
the matrices $A^{(\ell)}$ are sparse. 
Their estimation results  
build on similar ideas as {\cite[Theorem~6.1]{buhlmann2011statistics}},
relying on proving a restricted eigenvalue property for the design $X^\sT X/ n$. This result
hinges on stationary properties of the model \eqref{eq:varddef},
which we summarize prior to stating the estimation result.

\begin{definition}[Stability and invertibility of \VAR($d$) Process \cite{basu2015regularized}]\label{def:timeseriesstab}
A $\VAR(d)$ process with an associated 
reverse characteristic polynomial 
\begin{align}\label{eq:charpoldef}
&{\cal A}(\gamma)= I -\sum\limits_{\ell=1}^{d}{A}^{(\ell)} \gamma^\ell\,,
\end{align}
is called stable and invertible if $\det({\cal A}(\gamma))\neq 0$ for all $\gamma\in \mathbb{C}$ with $|\gamma|=1$. Based on this characteristic
polynomial, we also define the following spectral parameters:
\begin{align*}
\mumin(\cA) &= \min_{\abs{\gamma} =1} \lambdamin( \cA^*(\gamma) \cA(\gamma)  )\\
\mumax(\cA) &= \max_{\abs{\gamma} = 1} 
\lambdamax( \cA^*(\gamma) \cA(\gamma)  )
\end{align*}
\end{definition}

\begin{theorem}[Estimation Bound] \label{propo:estimation}
Recall the relation $y = X\theta_0 + \eps$, where $X,y, \theta_0$ are given by \eqref{eq:Y-X} and let $\hth^\sL$ be the Lasso estimator
\begin{equation} \label{estimation:optimization}
    \hth^\sL=\argmin_{\theta\in \reals^{dp}}\Big \{ %
    \frac{1}{2n}\norm{y-X \theta}_2^2+\lambda_n \norm{\theta}_1  \Big\}\,. 
\end{equation} 
Assume that 
$|\supp(\theta_0)|\leq s_0$, and define 
\begin{align*}
\omega &= \frac{d\lambdamax(\Sigmazeta)}{\lambdamin(\Sigmazeta)} \cdot
\frac{\mumax(\cA)}{\mumin(\cA)} \\
 \alpha &= \frac{\lambdamin(\Sigmazeta)}{\mumax(\cA)}.
\end{align*}

 There exists a universal
constant $C>0$, such that for any $n \ge C \alpha\, \omega^2 s_0 \log (dp) $
and 
$\lambda_n = \lambda_0\sqrt{\log(dp)/n}$, with 
$\lambda_0 \ge 4\lambdamax(\Sigmazeta) (1 \vee \mumax(\cA))/\mumin(\cA)$
the following happens. With probability at least 
$1-(dp)^{-6}$, the estimate satisfies: 
\begin{align*}
 \norm{\hth^\sL-\theta_0}_{1} \le C \frac{\lambda_0}{\alpha} \sqrt\frac{s_0^2\log(dp)}{n}.
\end{align*}

\end{theorem}

In short, given the standardized setting where $\lambda_0, \alpha$ are order one, 
the $\ell_1$ estimation error rate is of order $s_0 \sqrt{\log(dp)/n}$, which
is the same obtained in data without temporal dependence. Our proof is
similar to that of Basu and Michailidis \cite{basu2015regularized}, 
and relies on 
establishing a now-standard restricted eigenvalue property for the
design  $X^\sT X/n$. The spectral characteristics of the time
series quantified in Definition \ref{def:timeseriesstab} play
an important part in establishing this. We refer the reader to
Appendix \ref{app:timeseries} for the proof, as well as a discussion of the
differences with the proof of \cite{basu2015regularized}.

\subsection{Constructing the online debiased estimator} \label{sec:construct-online-TS}

\begin{figure}[t]
\centering
\includegraphics[width=0.9\linewidth]{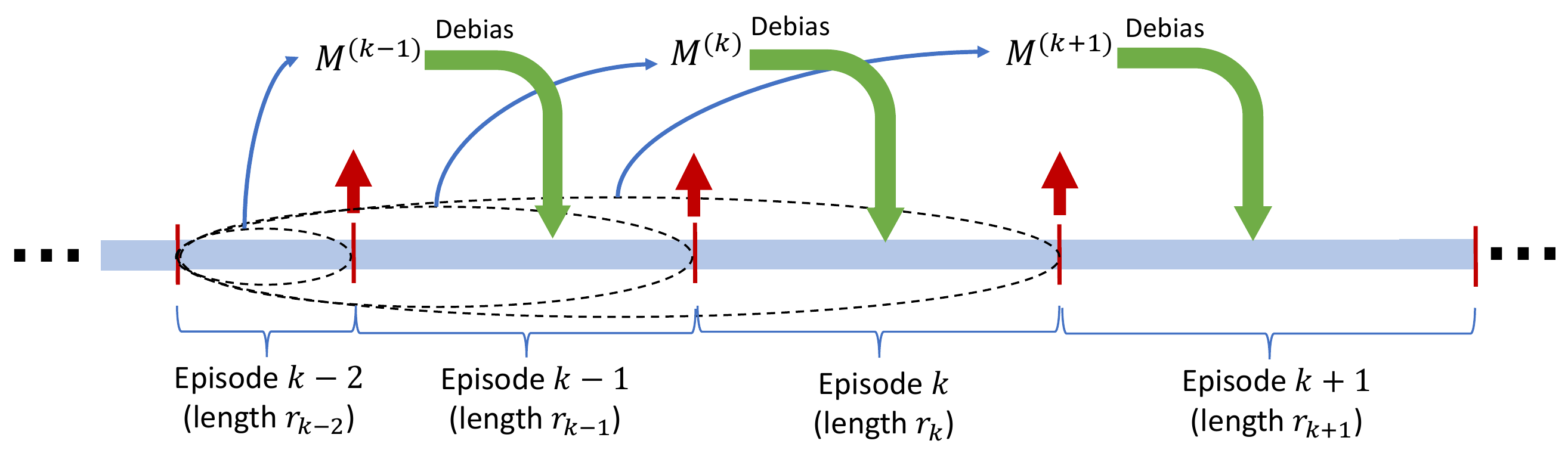}
\caption{Schematic for constructing the debiasing matices
$M^{(\ell)}$. We divide time into $K$ episodes $E_0, \dots, E_{K-1}$; in episode $\ell$, $M_i$ is held constant at $M^{(\ell)}$, which is a function of $x_t$ in \emph{all prior} episodes.  }
\end{figure}

Our task now is to construct a predictable sequence of
debiasing matrices $\{M_i\}_{i\le n}$. One simple approach
is the `sample-splitting' approach: construct a generalized
inverse $M$ based on the first $n/2$ data points using, for 
example, the program of \cite{javanmard2014confidence} and
let the sequence $\{M_i\}_{i\le n}$ be defined by
\begin{align*}
M_i &= \begin{cases}
0 &\text{ if } i \le n/2  \\
M &\text{ if } n/2 <i \le n. 
\end{cases}
\end{align*}
It is easy to see that this is a valid predictable sequence. However, 
due to sample-splitting, it does not make an efficient use of the 
data and loses power. More importantly, it is not clear that the underlying
martingale (the noise component of the debiased estimator $\sqrt{n}W_n = \sum_i M_i x_i \eps_i$) will be stable in the sense of Definition \ref{def:stability}.
Our proposal generalizes sample-splitting via an episodic structure 
and, importantly, regularizes to ensure stability.

We partition the time indices $[n]$ into $K$ episodes $E_0, \dotsc, E_{K-1}$, with $E_\ell$ of length $r_{\ell}$, so that $\sum_{\ell=0}^{K-1} r_\ell = n$.
Over an episode $\ell$, we keep the debiasing matrix $M_i = \Mell$
to be fixed over time points in the episode. Moreover, $M^{(\ell)}$ is
constructed using all the time points in \emph{previous} episodes
$E_0, \dots, E_{\ell-1}$ in the following way.
Let $n_\ell = r_0 + \dotsc+r_{\ell-1}$, for $\ell = 1, \dotsc, K$; hence, $n_{K} = n$. 
Define the sample covariance of the features in the first $\ell$ episodes. 
$$\Sell = \frac{1}{n_\ell} \sum_{t\in E_0\cup \dotsc \cup E_{\ell-1}} x_t x_t^\sT\,,$$ 
The matrix $\Mell$ has rows $(\mli)_{a\in [dp]}$ 
as the solution of the optimization:
\begin{align}\label{eq:opt}
\begin{split}
\text{minimize}\quad &m^\sT \Sell m\\
\text{subject to}\quad &\|\Sell m - e_a\|_\infty \le \mu_\ell,\quad \|m\|_1\le L\,, 
\end{split}
\end{align}
for appropriate values of $\mu_\ell, L>0$. 
We then construct the online debiased estimator for 
coordinate $a$ of $\theta_0$ as follows:
 \begin{align}\label{eq:debias}
\onth= \hth^\sL + \frac{1}{n} \sum_{\ell=1}^{K-1} \sum_{t\in E_{\ell}} \Mell x_t (y_t - \<x_t, \hth^\sL\>)\,.
\end{align} 

In Section~\ref{subsec: Distributional:char}, we show that the constructed online debiased estimator $\onth$ is asymptotically unbiased and admits a normal distribution. To do that we provide a high probability bound on the bias of $\onth$ (See Lemma~\ref{lem:onlinebias}). This bound is in terms of the batch sizes $r_\ell$, from which we propose the following guideline for choosing them: $r_0 \sim \sqrt{n}$ and $r_\ell \sim \beta^\ell$, for a constant $\beta >1$, and $\ell \ge 1$. 


Before proceeding into the distributional characterization of the online debiased estimator for $\theta_0$ (entries of  coefficient matrices $A^{(\ell)}$), we revisit the numerical example from Section~\ref{sec:example1} in which the (offline) debiased estimator of~\cite{javanmard2014confidence} does not display an unbiased normal distribution. However, as we will observe the constructed online debiased estimator empirically  admits an unbiased normal distribution.

\paragraph{Revisiting the numerical example from Section~\ref{sec:example2}}
In Section~\ref{sec:example2}, we considered a $\VAR(d)$ model with $p=15$, $d=5$, $T=60$, and diagonal $A^{(\ell)}$ matrices with value $b=0.15$ on their diagonals. 
The covariance matrix $\Sigma_{\zeta}$ of the noise terms $\zeta_t$ is chosen as $\Sigma_{\zeta}(i,j)=\rho^{\ind(i\neq j)}$ with $\rho=0.5$ and $i,j\in[p]$. The population covariance matrix of vector $x_t= (z_{t+d-1}^\sT,\dots,z_t^\sT)^\sT$ is 
a $dp$ by $dp$ matrix $\Sigma$ consisting of $d^2$ blocks of size $p\times p$ with $\Gamma_z(r-s)$ as block $(r,s)$. The analytical formula to compute $\Gamma_z(\ell)$ is given by~\cite{basu2015regularized}:
\[
\Gamma_z(\ell)=\frac{1}{2\pi}\int\limits_{-\pi}^{\pi}\cA^{-1}(e^{-j\theta})\Sigma_{\zeta}(\cA^{-1}(e^{-j\theta}))^*e^{j\ell \theta}d\theta\,,
\]
where $\cA(\gamma)$ is given in equation \eqref{eq:charpoldef}. Figure \ref{fig:heatmaps} shows the heat maps of magnitudes of the elements of $\Sigma$ and the precision matrix $\Omega= \Sigma^{-1}$ for the on hand $\VAR(5)$ process.   
As evident from Figure \ref{fig:fixed-coord}, the noise component of offline debiased estimator is biased. Here, we look into the noise component of the online debiased  estimator given by
\begin{align}
W^{{\rm on}} &= \frac{1}{\sqrt{n}} \sum_{\ell=1}^{K-1} M^{(\ell)} \sum_{t\in E_{\ell}} x_t \eps_t\,,
\end{align}
with $M^{(\ell)}$ constructed from the solutions to optimization \eqref{eq:opt} for $\ell = 1, \dotsc, K-1$.
 Also, recall that $\eps = (\zeta_{d+1, i}, \zeta_{d+2, i} \dots, \zeta_{T, i})$ by equation~\eqref{eq:Y-X}.

In Figure \ref{fig:fixed-coord2}, we show the QQ-plot, PP-plot and histogram of $W^{{\rm on}}_1$ and $W^{{\rm off}}_1$  (corresponding to the entry $(1,1)$ of matrix $A_1$) for 1000 different realizations of the noise $\zeta_t$.  As we observe, even the noise component $W^{{\rm off}}$ is biased because the offline construction of $M$ depends on all features $x_t$ and hence on endogenous noise $\zeta_t$. 
However, the online construction of decorrelating matrices $M^{(\ell)}$, makes the noise term a martingale and hence $W^{{\rm on}}$ converges in distribution to a zero mean normal vector, allowing for a distributional characterization of the online debiased estimator.

\begin{figure}[!h]
	\centering
	\begin{subfigure}{0.5\textwidth}
		\centering
		\includegraphics[scale =0.35]{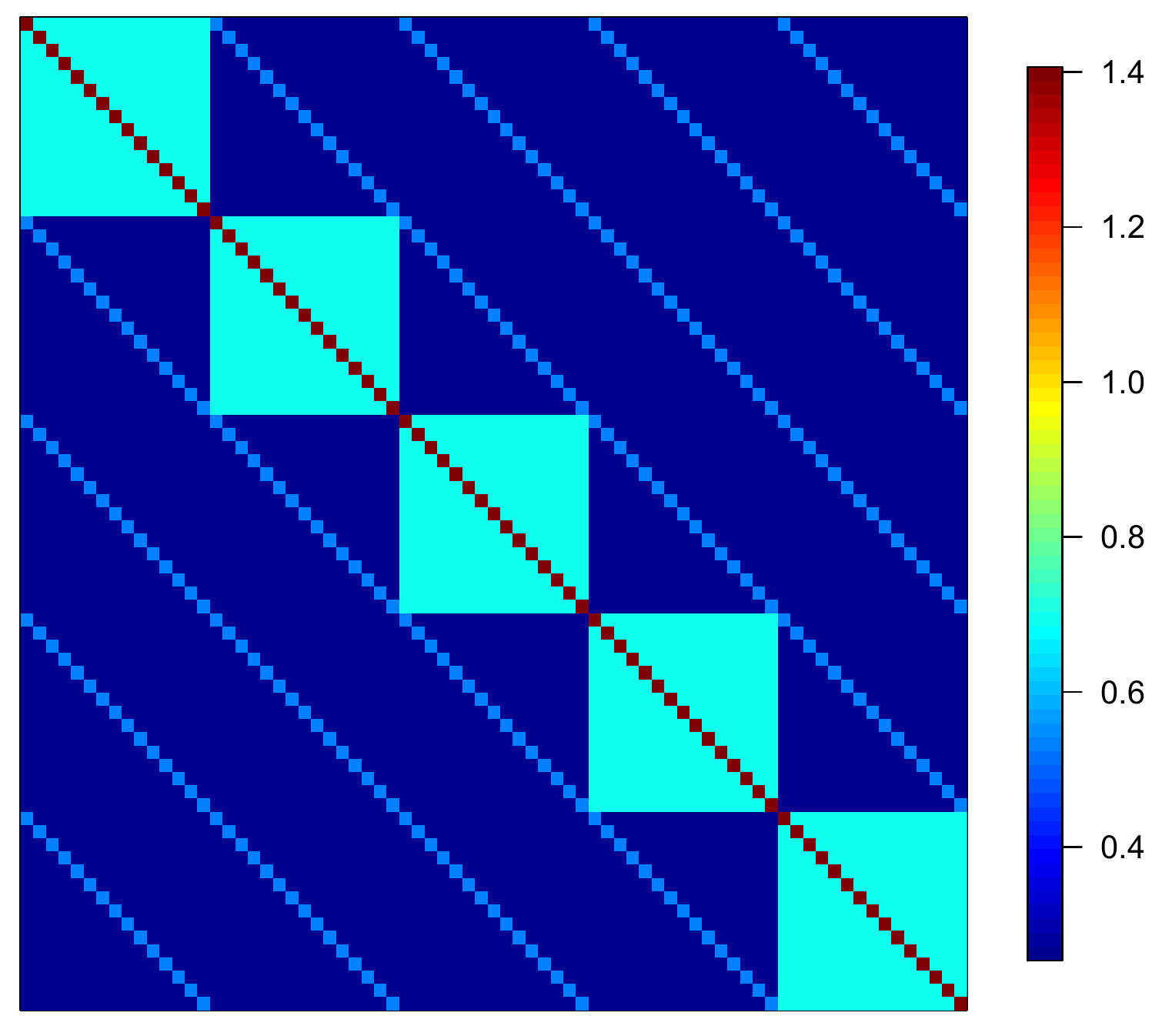}
		\caption{Heat map of $\Sigma$}
		\label{fig:heatmap-gamma}
	\end{subfigure}%
   \begin{subfigure}{.5\textwidth}
   	\centering
   	\includegraphics[scale=0.35]{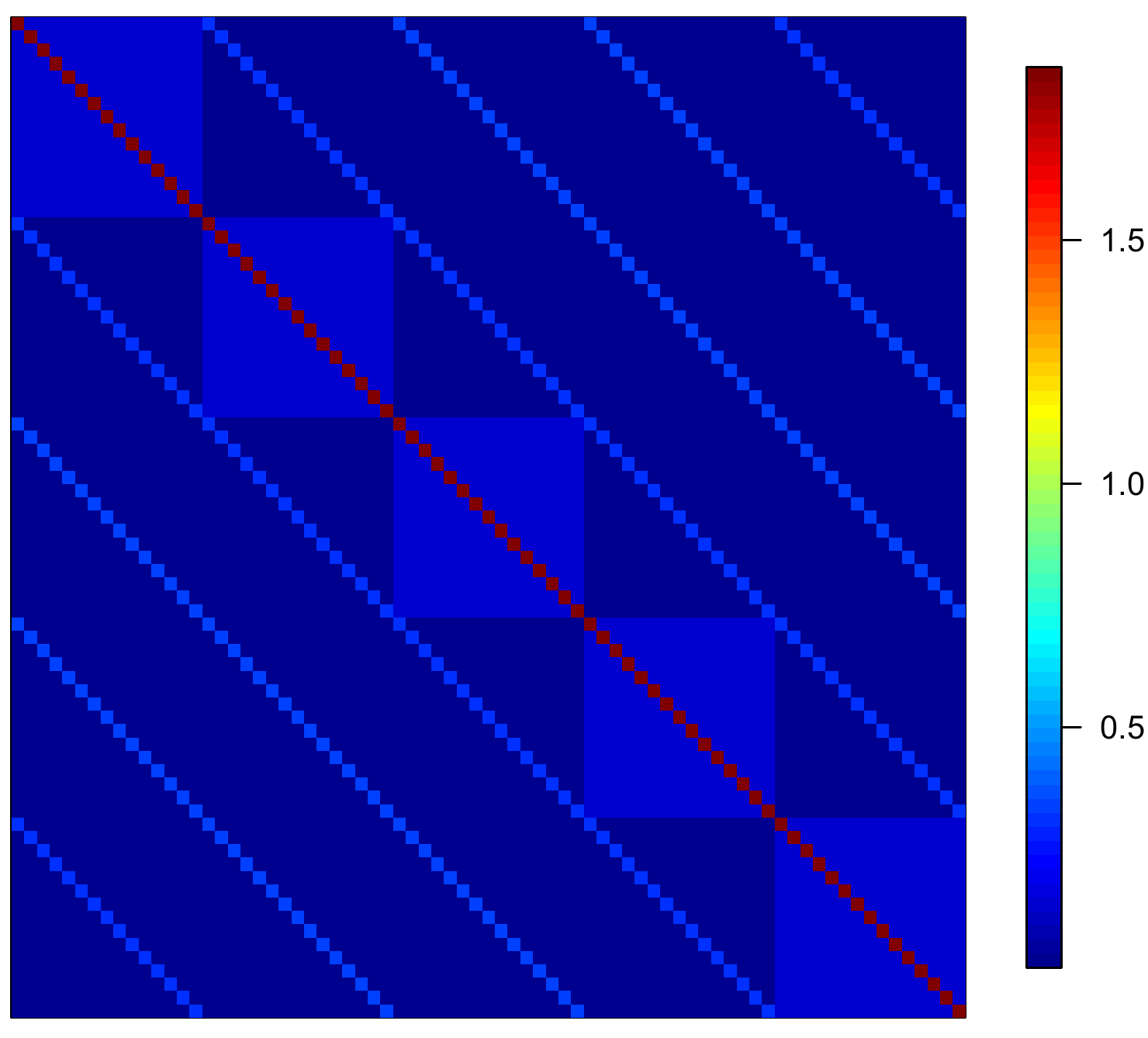}
   	\caption{Heat map of $\Omega$}
   	\label{fig:heatmap-omega}
   \end{subfigure}
  \caption{{\small Heat maps of magnitudes of elements of covariance matrix $\Sigma \equiv \E(x_ix_i^T)$ (left plot), and precision matrix $\Omega = \Sigma^{-1}$ (right plot). In this example. $x_i$'s are generated from a $\VAR(d)$ model with covariance matrix of noise $\Sigma_{\zeta}(i,j)=\rho^{|i-j|}$ with values $d=5$, $p=15$, $T=60$, $\rho=0.5$, and diagonal $A^{(i)}$ matrices with $b=0.15$ on diagonals.}}
  \label{fig:heatmaps}
\end{figure}

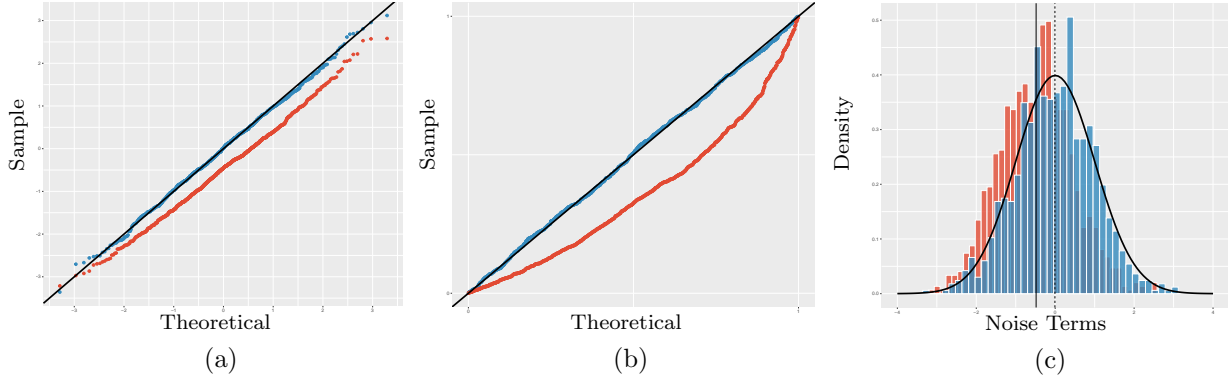
\begin{figure}[]
\begin{subfigure}{.33\linewidth}
\centering
\includegraphics[scale=0.20]{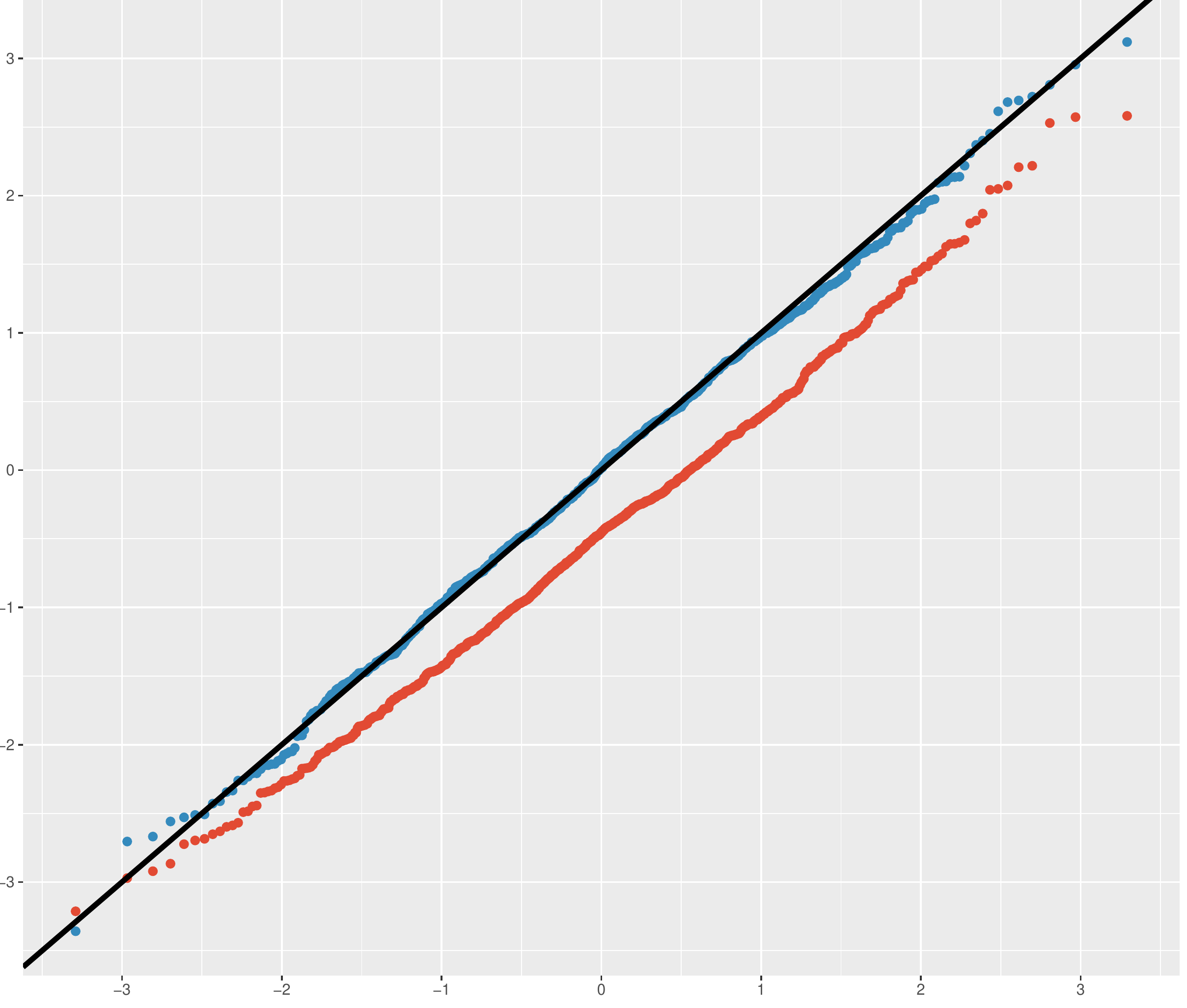}
\put(-150,55){\rotatebox{90}{\scriptsize{Sample}}}
\put(-93,-5){\rotatebox{0}{\scriptsize{Theoretical}}}
\caption{}
\label{fig:fixed-coord:sub1}
\end{subfigure}%
\begin{subfigure}{.33\linewidth}
\centering
\includegraphics[scale=0.20]{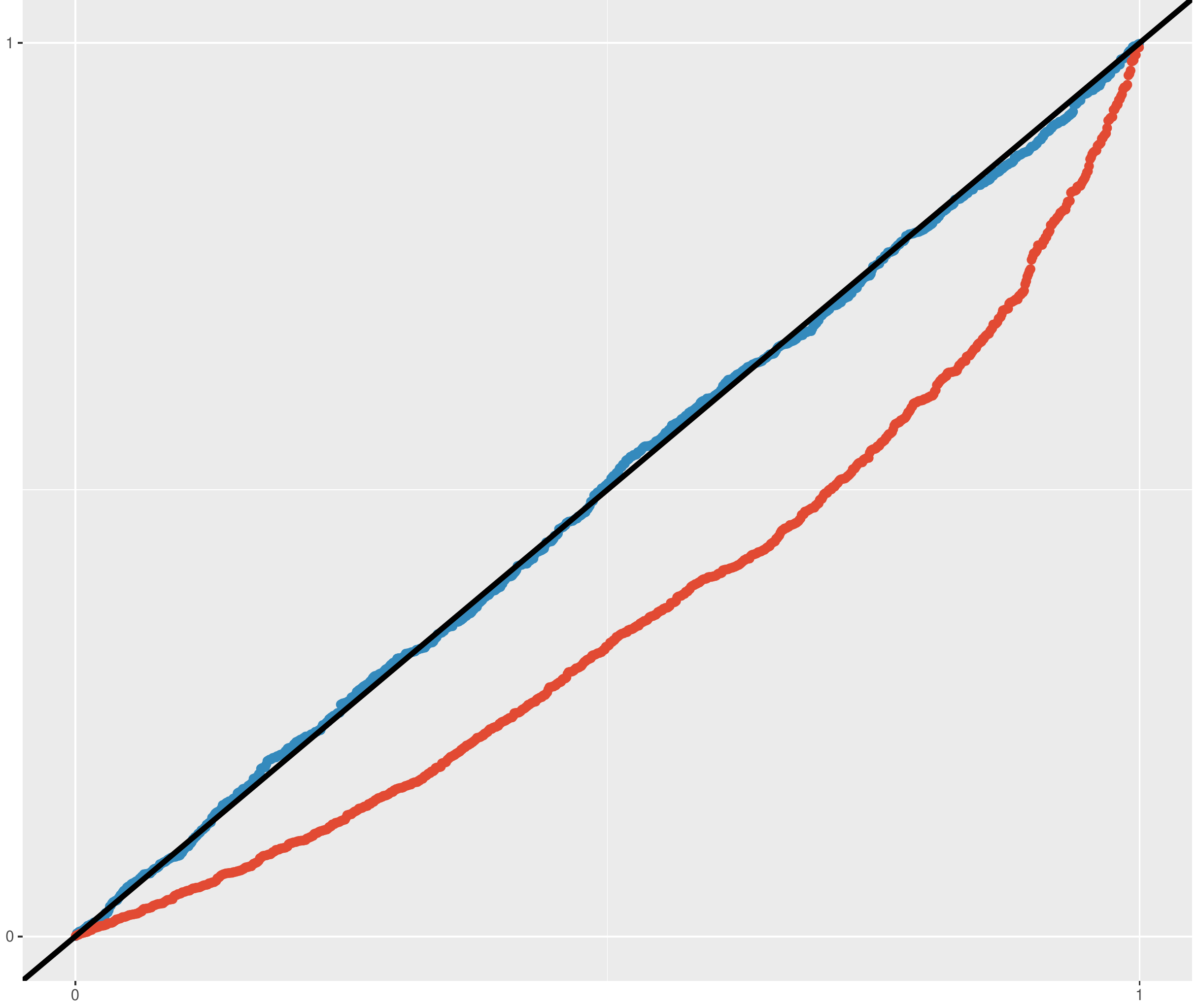}
\put(-150,55){\rotatebox{90}{\scriptsize{Sample}}}
\put(-93,-5){\rotatebox{0}{\scriptsize{Theoretical}}}
\caption{}
\label{fig:fixed-coord:sub2}
\end{subfigure}
\begin{subfigure}{0.33\linewidth}
\centering
\includegraphics[scale=0.20]{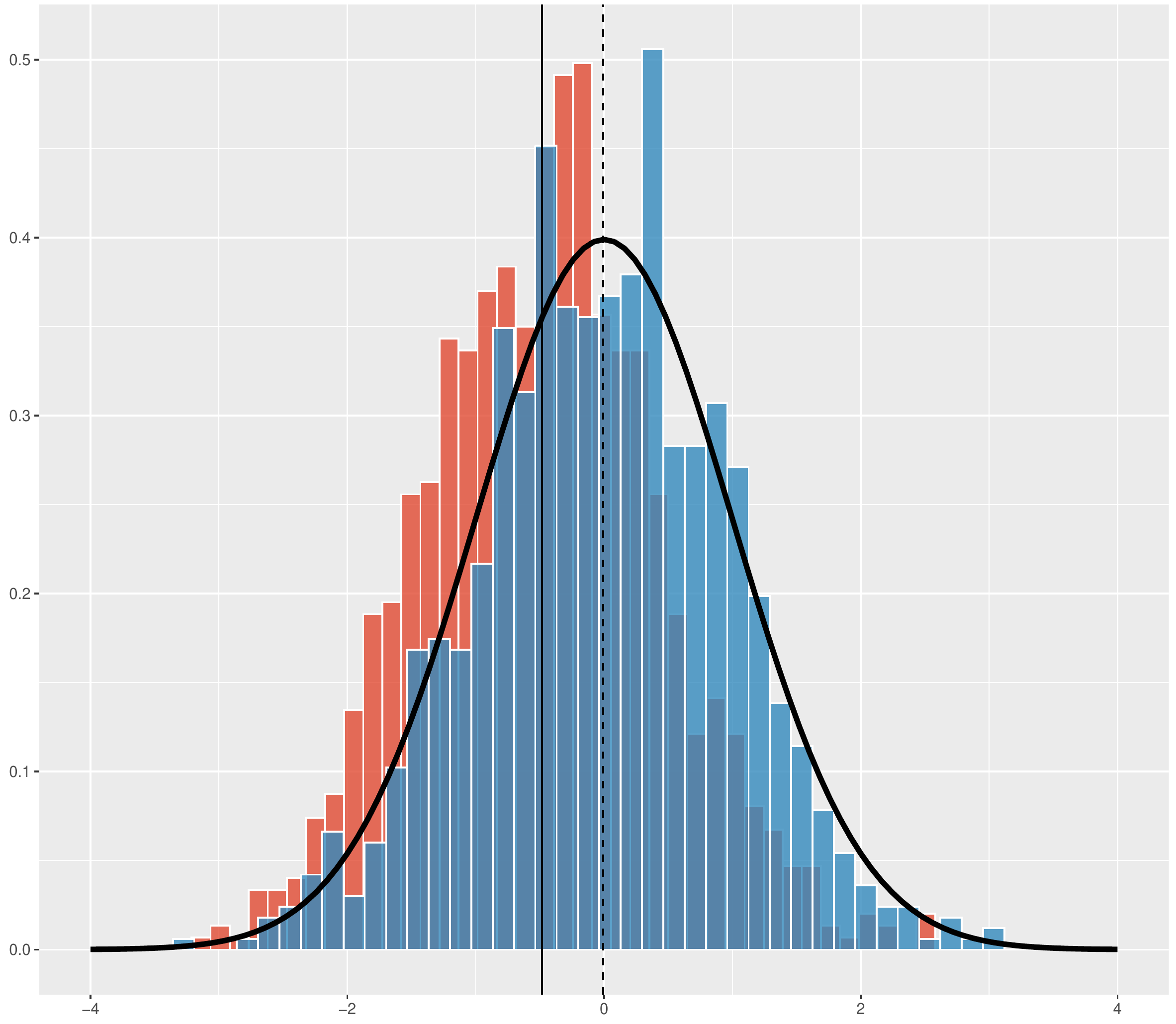}
\put(-93,-5){\rotatebox{0}{\scriptsize{Noise Terms}}}
\put(-150,55){\rotatebox{90}{\scriptsize{Density}}}
\caption{}
\label{fig:fixed-coord:sub3}
\end{subfigure}
\caption{{Plots \ref{fig:fixed-coord:sub1}, \ref{fig:fixed-coord:sub2}, and \ref{fig:fixed-coord:sub3} show the QQ plots, PP plots, and the 
histogram of online debiased noise terms (blue) and offline debiased noise
terms (red) over 1000 independent experiments, respectively and black 
curve/lines denote the ideal standard normal distribution. The solid and 
dash vertical lines in plot (c) indicate the location of the mean of offline and online debiased
noise terms, respectively. }} 
\label{fig:fixed-coord2}
\end{figure}

\subsection{Distributional characterization of online debiasing }\label{subsec: Distributional:char}
We start our analysis of the online debiased estimator $\onth$ by considering a bias-variance decomposition.
Using $y_t =  \<x_t, \theta_0\> + \varepsilon_t$ in the definition~\eqref{eq:debias}: 
\begin{align}
\onth -\theta_0 &= \Lsth -\theta_0 + \frac{1}{n} \sum_{\ell=1}^{K-1} \sum_{t\in E_{\ell}} \Mell x_t x_t^\sT(\theta_0 - \Lsth) + \frac{1}{n} \sum_{\ell=1}^{K-1} \sum_{t\in E_{\ell}} \Mell x_t \varepsilon_t \nonumber\\
&= \Big(I - \frac{1}{n} \sum_{\ell=1}^{K-1} \sum_{t\in E_{\ell}} \Mell x_t x_t^\sT\Big) (\Lsth-\theta_0) + \frac{1}{n} \sum_{\ell=1}^{K-1} \sum_{t\in E_{\ell}} \Mell x_t \varepsilon_t\,.\label{eq:debiasingdecomp}
\end{align}
With the shorthand $\Rell =  (1/r_\ell)\sum_{t\in E_\ell} x_t x_t^\sT$ for the sample covariance of features in episode $\ell$ and the bias $B_n$
and variance term $W_n$ below
\begin{align}
B_n &\equiv \sqrt{n}\Big(I - \frac{1}{n} \sum_{\ell=1}^{K-1} r_{\ell}\Mell \Rell \Big)\,,\label{eq:bias-TS}\\
W_n&\equiv  \frac{1}{\sqrt{n}} \sum_{\ell=1}^{K-1} \Mell \Big(\sum_{t\in E_{\ell}}  x_t \varepsilon_t\Big)\,,\label{eq:noise-TS}
\end{align}
we arrive at the following decomposition
\begin{align}\label{eq:TS-decomposition}
\onth &= \theta_0 + \frac{1}{\sqrt{n}}
\big(B_n (\htheta^\sL - \theta_0) + W_n \big)\,.
\end{align}



Our first set of results concern the bias of $\onth$, establishing
that this is asymptotically smaller than that of the LASSO estimate. 
The analysis of the bias focuses mostly on the term $B_n$, which 
in turn, is controlled by the parameter $\mu_\ell$ in 
the optimization~\eqref{eq:opt}. We would like to 
choose $\mu_\ell$ small enough to reduce the bias, but large enough so that the optimization~\eqref{eq:opt} is still feasible. 
The following lemma shows that, with high probability, 
$\mu_\ell$ of order $\omega \sqrt{\log (dp) /n_\ell}$ is sufficient
to make the optimization feasible.

\begin{lemma}\label{lem:biasmatrixdeviation}
Let $\Omega = \Sigma^{-1} = (\E\{x_tx_t^\sT\})^{-1}$ be the precision matrix of the time series. 
There exists universal constants $C, C'$ such that the following happens. 
Suppose that $n_\ell \ge C\omega^2 \log(dp)$ where $\omega$ is defined
in Theorem \ref{propo:estimation}. Then with probability $1-(dp)^{-6}$:
\begin{align*}
 \max_{i, j} \abs{\Omega\Sell - \ind(i=j)}
 &\le C'\omega \sqrt{ \frac{\log(dp)}{n_\ell}  }.
 \end{align*} 
\end{lemma}
The proof of Lemma~\ref{lem:biasmatrixdeviation} is given in 
Appendix~\ref{proof:lem:biasmatrixdeviation}. The following
theorem uses Lemma \ref{lem:biasmatrixdeviation} to control
the bias of the online debiased estimator.

\begin{theorem}(Bias control)\label{thm:TSbiasbound}
Consider the $\VAR(d)$ model \eqref{eq:varddef} and let $\onth$ be the debiased estimator \eqref{eq:debias} where the decorrelating matrices $M^{(\ell)}$ are computed according to Eq.\eqref{eq:opt}, with $\mu_\ell = c_1\omega\sqrt{(\log (dp)/n_\ell}$ and  $L\ge \|\Omega\|_1$. Further assume that 
the base estimator is $\hth^\sL$ computed with 
$\lambda = \lambda_0 \sqrt{\log(dp)/n}$ where 
$\lambda_0 \ge 4\lambdamax(\Sigmazeta)(1\vee \mumax(\cA))/\mumin(\cA)$.

 Then, 
 under 
 the sample size condition $n\ge C \omega^2  s_0 \log(dp)$, 
we have
 \begin{align}\label{mydecompose}
 \sqrt{n}(\onth  -\theta_0) & =  W_n + \Delta_n, 
  \end{align}
  where $\E\{W_n\} = 0$ and 
  \begin{align}\label{mydecompose-B}
  \P\Big\{ \norm{\Delta_n}_\infty \ge C_1 \frac{\lambda_0(\omega + L\gamma)}{\alpha} \frac{s_0\log(dp)}{\sqrt{n}} \Big\} &\le (dp)^{-4}, 
  \end{align}

  The parameters $\omega,\alpha$ are defined in Theorem \ref{propo:estimation}, and $\gamma = d\lambdamax(\Sigmazeta)/\mumin(\cA)$.
Further, the bias satisfies
\begin{align*}
  \lVert\E \{\onth - \theta_0\}\rVert_\infty \le  \frac{C_1\lambda_0(\omega + L\gamma)}{\alpha} \frac{s_0\log(dp)}{n} 
  + \frac{C_2\norm{\theta_0}_1}{(dp)^6}
\end{align*}
 \end{theorem} 
 
 We refer to Appendix~\ref{proof:thm:TSbiasbound} for the proof of Theorem~\ref{thm:TSbiasbound}.

Note that he above theorem bounds the bias term $\Delta_n$ for finite sample size $n$. To study these bounds in an asymptotic regime, we make the following assumption to simplify our presentation.
 \begin{assumption}\label{assmp:TS}
Suppose that 
\begin{enumerate}
\item The parameters $\lambdamin(\Sigmazeta)$,  $\lambdamax(\Sigmazeta)$, $\mumin(\cA)$ and $\mumax(\cA)$ are bounded away
from $0$ and $\infty$, as $n,p\to \infty$.
\item With $\Omega = \Sigma^{-1} = (\E\{x_tx_t^\sT\})^{-1}$ the precision matrix of the data points $\{x_t\}$, and $s_0$ the sparsity of $\theta_0
 = (A_i^{(1)}, \dots, A_i^{(d)})^\sT$, we assume that $ \|\Omega\|_1   =  o(\sqrt{n}/\log(dp))$.
\end{enumerate}
\end{assumption}
Under Assumption \ref{assmp:TS} the spectral quantities $\omega, \gamma, \alpha$ and (therefore) $\lambda_0$ are order one. We can also ignore the lower
order term $\norm{\theta_0}_1/(dp)^6$ in the high-dimensional regime. Indeed, the denominator $(dp)^6$ can be changed to $(dp)^c$ for arbitrary large $c>0$, by adjusting constant $C_1$ and the tail bound in Eq.\eqref{mydecompose-B}. Therefore, as far as $\|\theta_0\|_1$ grows polynomially at $p$, then this term vanishes asymptotically. The
 theorem, hence, shows that the bias of the online debiased estimator is of order $L s_0 (\log p)/n$. On the other hand, recall
  the filtration $\cF_t$ generated by 
  $\{\varepsilon_1, \dotsc, \varepsilon_t\}$ and rewrite \eqref{eq:noise-TS} as $W_{n} = \sum_t v_t \varepsilon_t$, where $v_t =\Mell x_t/\sqrt{n}$ (Sample $t$ belongs to episode $\ell$). 
 We use Assumption \ref{assmp:TS} in Lemma~\ref{lem:stab-W-TS} below, to show that for each coordinate $i\in [dp]$, the conditional variance $\sum_{t=1}^n \E(\varepsilon_t^2 v_{t,i}^2| \cF_{t-1}) = (\sigma^2/n) \sum_{t=1}^n \<\mli,z_t\>^2$  is 
  of order one. Hence $\|\Delta_n\|_\infty$ is asymptotically 
  dominated by the noise variance when  $s_0 = o\left(\tfrac{\sqrt{n}}{L \log (dp)}\right)$. 
 
 Another virtue of Lemma \ref{lem:stab-W-TS} is that it shows the martingale sum $W_n$ is stable in an appropriate sense. This is a key technical step that allows us to characterize the distribution of the noise term $W_n$ by applying the martingale CLT (e.g., see ~\cite[Corollary 3.2]{hall2014martingale}) and conclude that the unbiased component $W_n$ admits a Gaussian limiting distribution. 
 \begin{lemma}(Stability of martingale $W_n$)
\label{lem:stab-W-TS}
Let $\onth$ be the debiased estimator \eqref{eq:debias} with $\mu_\ell = \mya\sqrt{(\log p)/n_\ell}$ and $L= L_0 \|\Omega\|_1$, for an arbitrary constant $L_0\ge1$.
Under Assumption \ref{assmp:TS}, and for any fixed sequence of integers $a(n)\in [dp]$,\footnote{We index the sequence with the sample size $n$ that is diverging. Since we are in high-dimensional setting $p\ge n$ is also diverging.} we have 
\begin{align}\label{eq:conditionalvar}
V_{n,a} \equiv \frac{\Sigmazeta_{i, i}}{n}\Bsum \<m^\ell_a, x_t\>^2
 &=  \Sigmazeta_{i, i}\cdot \Omega_{a,a} + o_P(1). 
\end{align}
In addition, we have
\begin{align}\label{eq:summand}
\max \Big\{\frac{1}{\sqrt{n}}\lvert \<m^\ell_a, x_t\> \eps_t \rvert:\, \ell \in [K-1],\, t\in [n-1]\Big\} & = o_P(1).
\end{align}
\end{lemma}
We refer to Appendix~\ref{proof:lem:stab-W-TS} for the proof of Lemma~\ref{lem:stab-W-TS}. With Lemma~\ref{lem:stab-W-TS} in place, we can apply a martingale central limit theorem \cite[Corollary 3.2]{hall2014martingale} to obtain the following result.
\begin{corollary}\label{cor:noise-TS}
Consider the $\VAR(d)$ model \eqref{eq:varddef} for time series and let $\onth$ be the debiased estimator \eqref{eq:debias} with $\mu_\ell = C_1\omega\sqrt{(\log p)/n_\ell}$ and $L= L_0 \|\Omega\|_1$, for an arbitrary constant $L_0\ge1$. For any fixed sequence of integers $a(n)\in [dp]$, define the conditional variance $V_n$ as
\begin{align*}
V_{n,a} &\equiv \frac{\Sigmazeta_{i, i}}{n}\Bsum \<m^\ell_a, x_t\>^2\,.
\end{align*}
 Under Assumption \ref{assmp:TS}, for any fixed coordinate $a\in [dp]$, and for all $x\in \reals$ we have
 \begin{align}
 \lim_{n\to\infty} \prob \Big\{\frac{W_{n,a}}{\sqrt{V_{n,a}}} \le x \Big\} = \Phi(x)\,,
 \end{align}
 where $\Phi$ is the standard Gaussian cdf.
\end{corollary}

For the task of statistical inference, Theorem~\ref{thm:TSbiasbound} and Corollary~\ref{cor:noise-TS} suggest to consider the scaled residual $\sqrt{n}(\onth_a - \theta_{0,a})/\sqrt{V_{n,a}}$ as the test statistics. Our next proposition characterizes its distribution. The proof is straightforward given the result of Theorem~\ref{thm:TSbiasbound} and Corollary~\ref{cor:noise-TS} and is deferred to Appendix~\ref{proof:pro:SS}. In its statement we omit explicit constants that can be easily derived from Theorem~\ref{thm:TSbiasbound}.

 \begin{theorem}\label{pro:SS}
  Consider the $\VAR(d)$ model \eqref{eq:varddef} for time series and let $\onth$ be the debiased estimator \eqref{eq:debias} with $\mu_\ell = C_1\omega\sqrt{(\log p)/n_\ell}$, $\lambda = \lambda_0 \sqrt{\log(dp)/n}$, and $L= L_0 \|\Omega\|_1$, for an arbitrary constant $L_0\ge1$. Suppose that Assumption~\ref{assmp:TS} holds and $s_0 = o\left(\tfrac{\sqrt{n}}{\|\Omega\|_1\log(dp)}\right)$, then the following holds true for any fixed sequence of integers $a(n)\in [dp]$. For all $x\in \reals$, we have
 \begin{align}\label{eq:DC-TS0}
 \lim_{n\to\infty} \bigg| \prob\bigg\{\frac{\sqrt{n} (\onth_a - \theta_{0,a})}{\sqrt{V_{n,a}}} \le x\bigg\}  - \Phi(x)\bigg| = 0\,.
 \end{align}
 \end{theorem}


\section{Batched data collection}
\label{sec:batch}

Recall the stylized setting of
adaptive data collection in batches from Section \ref{sec:offlinefailure},
where 
the samples naturally separate into two batches:
the first $n_1$ data points where the covariates are i.i.d from a distribution $\prob_x$, and the second batch of
$n_2$ data points, where the covariates $x_i$ are drawn independently from the law of $x_1$, conditional on the event $\{\<x_1,\hth^1\> \ge \varsigma\}$, where $\varsigma$ is a potentially data-dependent threshold. 
The following theorem is a version of Theorem~6.1 in \cite{buhlmann2011statistics} and is proved in an analogous
manner. It demonstrates that even with adaptive data 
collection consistent estimation using the LASSO
is possible.

\begin{theorem}[{\cite[Theorem~6.1]{buhlmann2011statistics}}]\label{thm:batchlassoerr}
Suppose that the true parameter $\theta_0$ is $s_0$-sparse
and the distribution $\P_x$
is such that with probability one
the following two conditions hold:
$(i)$ the
covariance $\E\{x x^\sT\}$ and 
$\E\{x x^\sT | \<x, \htheta^1\> \ge \varsigma\}$
are
$(\phi_0, \supp(\theta_0))$-compatible
and $(ii)$ $x$ as well as $x\vert_{\<x, \htheta^1\> \ge \varsigma}$
are $\kappa$-subgaussian. Suppose that 
$n \ge C_1(\kappa^4 /\phi_0^2) s_0^2 \log p$. Then, the LASSO
estimate  $\htheta^\sL(y, X; \lambda_n)$ with $\lambda_n  = C_2\kappa \sigma \sqrt{(\log p)/n}$
satisfies, with probability exceeding $1 - p^{-3}$, 
\begin{align*}
\norm{\htheta^\sL - \theta_0}_1 &\le  \frac{C's_0\lambda_n}{\phi_0} 
=  \frac{C \kappa \sigma }{\phi_0}  s_0 \sqrt{\frac{\log p}{n}}.
\end{align*}
\end{theorem}
\begin{remark}(Estimating the noise variance)
For the correct estimation rate using the LASSO, Theorem \ref{thm:batchlassoerr}
requires knowledge of the noise level $\sigma$, which is used
to calibrate the regularization $\lambda_n$. Other estimators like the scaled LASSO \cite{sun2012scaled} or the square-root LASSO \cite{belloni2011square}
allow to estimate $\sigma$ consistently when it is unknown. 
This can be incorporated into the present setting, as done in 
\cite{javanmard2014confidence}. For simplicity, we focus
on the case when the noise level is known. However, the results hold as far as a consistent estimate of $\sigma$ is used. Formally, a consistent estimator refers to an estimate $\hsigma = \hsigma(y,X)$ of the noise level satisfying, for any $\eps>0$,
\begin{align}
\lim_{n\to\infty} \sup_{\|\theta_0\|_0\le s_0} \prob\left(\Big|\frac{\hsigma}{\sigma} - 1\Big|\ge \eps\right) = 0 \,.
\end{align} 
\end{remark}
\begin{remark}At the expense of increasing
the absolute constants  in Theorem \ref{thm:batchlassoerr}, the probability $1-p^{-3}$ can be made $1-p^{-C}$ for any
arbitrary constant $C> 1$. 
\end{remark}

Let $X_1$ and $X_2$ denote the design
matrices of the two batches and, similarly, 
$y^\paren{1}$ and $y^\paren{2}$ the two responses vectors.
In this setting, we use an online debiased estimator as follows:
\begin{align}
  \onth &= \htheta^\sL + \frac{1}{n} M^\paren{1} X_1^\sT (y^\paren{1} - X_1 \htheta^\sL) + \frac{1}{n} M^\paren{2} X_2^\sT (y^\paren{2} - X_2 \htheta^\sL),
  \label{eq:batchonlinedebias}
  \end{align}  
  where we will construct $M^\paren{1}$ as a function of $X_1$
  and $M^\paren{2}$ as a function of $X_1$ as well as $X_2$. 
  The proposal in Eq.\eqref{eq:batchonlinedebias} follows from the general recipe in Eq.\eqref{eq:onlinedebias}
  by setting 
  \begin{itemize}
 \item $M_i = M^\paren{1}$ for $i = [n_1]$ and $M_i = M^\paren{2}$ for $i = n_1+1, \dotsc, n$.
\item Filtrations $\fF_i$ constructed as follows. 
For $i < n_1$, $y_1, \dots, y_i$, $x_1, \dots x_{n_1}$
and $\eps_1, \dots, \eps_i$ are measurable with respect to $\fF_i$.
For $i \ge n_1$,  $y_1, \dots, y_i$, $x_1, \dots, x_n$
and $\eps_1, \dots \eps_i$ are measurable with respect to $\fF_i$.
 \end{itemize}
  By construction, this choice satisfies the 
  predictability condition, given by Definition~\ref{def:pred}. 

  Note that Eq.\eqref{eq:batchonlinedebias}
  nests an intuitive `sample splitting' approach. Indeed,
  debiasing $\htheta^\sL$ using exactly one of the two
  batches is equivalent to setting one of 
  $M^\paren{1}$ or $M^\paren{2}$ to $0$. While sample splitting can
  be shown to work under appropriate conditions, 
  our approach is more efficient with use of the data
  and gains power in comparison.
We construct $M^\paren{1}$ and $M^\paren{2}$ using a modification
of the program used in \cite{javanmard2014confidence}. 
Let $\hSigma^\paren{1} = (1/n_1)X_1^\sT X_1$ and $\hSigma^\paren{2} = (1/n_2)X_2^\sT X_2$ be the sample covariances of each batch; 
 let $M^\paren{1}$ have rows $(m^\paren{1}_a)_{1\le a\le p}$ and similarly
for $M^\paren{2}$. Using parameters $\mu_\ell,L  > 0$ that
we set later, we choose $m^\paren{\ell}_a$, the $a^\th$
row of $M^\paren{\ell}$, as a solution to
the  program
\begin{align}
&\text{minimize }  \quad\<m, \hSigma^{(\ell)} m\>  \nonumber \\
&\text{subject to } \quad \norm{\hSigma^{(\ell)} m - e_a}_\infty \le \mu_\ell, \;\;
\norm{m}_1 \le L.  \label{eq:batchMrowdef}\end{align}
Here $e_a$ is the $a^\th$ basis vector:
a vector which is one at the $a^\th$ coordinate 
and zero everywhere else.  

The intuition for the program \eqref{eq:batchMrowdef} 
is simple. The first constraint ensures that  $\hSigma^{(\ell)} m$ is 
close, in $\ell_\infty$ sense to the $e_a$, the $a^\th$ basis
vector and as we will see in Theorem~\ref{thm:batchbiasbound} it controls the bias term $\Delta$ of $\onth$. The objective is 
a multiple of the variance of the martingale
term $W$ in $\onth$ (cf. Eq.~\eqref{mydec}). We  
wish to minimize this as it directly affects the power of 
the test statistic or the length of valid confidence intervals constructed based on $\onth$. 
The  $\ell_1$ constraint on $m$, which is missing in \cite{javanmard2014confidence}, is crucial for our adaptive data setting.  
 This constraint ensures that  the
value of the program $\<m^{(\ell)}_a, \hSigma^{(\ell)} m^{(\ell)}_a\>$ is
stable, and does not fluctuate 
much from sample to sample (this is formalized as the `stability condition' in Lemmas~\ref{lem:batchstability} and \ref{lem:stab-W-TS}). It
is this stability that ensures that the martingale part of the residual displays
a central limit behavior. 

Note that in the non-adaptive
setting, inference can be performed
conditional on design $X$, and fluctuation
in $\<m^{(\ell)}_a, \hSigma^{(\ell)} m^{(\ell)}_a\>$
is conditioned out. In the adaptive
setting, this is not possible: one effectively cannot
condition on the design without
conditioning on the noise realization $\eps$, and therefore
we perform inference unconditionally on $X$.

\subsection{Online debiasing: a distributional characterization}

We begin the analysis of the online
debiased estimator $\onth$ by a
decomposition that mimics the classical debiasing.
\begin{align}
\onth &= \theta_0 + \frac{1}{\sqrt{n}}
\big(B_n (\htheta^\sL - \theta_0) + W_n \big), \label{eq:batchdebiasdecomp} \\
B_n &= \sqrt{n} \Big( I_p - \frac{n_1}{n} M^\paren{1} \hSigma^{(1)}  - \frac{n_2}{n} M^\paren{2} \hSigma^{(2)} \Big) \label{eq:batchBndef} \\
W_n &= \frac{1}{\sqrt n} \sum_{i\le n_1} M^\paren{1} x_i \eps_i
+ \frac{1}{\sqrt n} \sum_{n_1 < i \le n} M^\paren{2} x_i \eps_i. \label{eq:batchWndef}
\end{align}




\begin{assumption}(Requirements of design) \label{assmp:batchCov}
Suppose that the distribution $\P_x$ and the intermediate 
estimate $\htheta^1$, that is used in collecting the second batch, satisfy the following:
\begin{enumerate}
  \item There exists a constant $\lambdalbd>0$ so 
  that the eigenvalues of $\E\{xx^\sT\}$ and 
   $\E\{xx^\sT|\<x, \htheta^1\> \ge \varsigma\}$ are bounded
   below by $\lambdalbd$. 
\item The laws of $x$ and $x \vert_{\<x, \htheta^1\> \ge \varsigma}$ 
are $\kappa$-subgaussian for a constant $\kappa > 0$. 
\item The precision matrices $\Omega = \E\{x x^\sT\}^{-1}$
and $\Omega^\paren{2}(\hth^1) = \E\{x x^\sT | \<x, \hth^{1}\> \ge \varsigma\}^{-1}$
satisfy $\norm{\Omega}_1\vee \norm{\Omega^\paren{2}(\hth^1)}_1 \le L$.
\item The conditional covariance $\Sigma^\paren{2}(\theta) = \E\{x x^\sT | \<x, \theta\> \ge \varsigma\}$ is $K$-Lipschitz in its argument $\theta$, i.e. 
$\lVert\Sigma^{(2)}(\theta') - \Sigma^{(2)}(\theta)\rVert_\infty \le K \lVert\theta - \theta'\rVert_1$. 
\end{enumerate}
\end{assumption}
The first two conditions of Assumption \ref{assmp:batchCov}
are for ensuring that the 
base LASSO estimator $\htheta^\sL$ has small estimation
error. In addition, our debiasing makes use of the third
and fourth constraints on the precision matrices of the sampling
distributions. 
In the above, we will typically allow $L= L_n$ to
diverge with $n$. 

In the following Example we show that
Gaussian random designs satisfy
all the conditions of Assumption \ref{assmp:batchCov}. We refer to Section~\ref{sec:examples} for its proof.

\begin{example}\label{ex:batchgaussianrequirements}
Let $\P_x = \normal(0, \Sigma)$ and $\htheta$ be 
any vector such that $\norm{\htheta}_1\norm{\htheta}_\infty \le  L_\Sigma \lambdamin(\Sigma) \norm{\hth}/2 $
and $\norm{\Sigma^{-1}}_1 \le L_\Sigma/2$.
Then the distributions of $x$ and $x \vert_ {\<x, \htheta\> \ge \varsigma}$, with $\varsigma = \bvsigma \<\hth, \Sigma\hth\>^{1/2}$
for a constant $\bvsigma \ge 0$ satisfy the conditions of Assumption \ref{assmp:batchCov} 
with 
\begin{align*}
\lambdalbd = {\lambda_{\min}(\Sigma)},
\quad \kappa = 3\lambda_{\max}^{1/2}(\Sigma) (\bvsigma\vee \bvsigma^{-1}),\quad K = \sqrt{8}(1+\bvsigma^2)\frac{\lambdamax(\Sigma)^{3/2}}{\lambdamin(\Sigma)^{1/2}}, \quad L=L_\Sigma.
\end{align*}
\end{example}
Under Assumption \ref{assmp:batchCov}  we provide
a \emph{non-asymptotic} bound on the bias of 
the online debiased estimator $\onth$.
\begin{theorem}(Non-asymptotic bound on bias)\label{thm:batchbiasbound}
 Under Assumption \ref{assmp:batchCov}, there exists
 universal constants $C_1, C_2, C_3$ so that,
 when 
  $n\ge C_1 \kappa^4 s_0^2 \log p / \phi_0^2$ and 
 $n_1 \wedge n_2 \ge C_1( \lambdalbd/ \kappa^2 +\kappa^2/\Lambda_0) \log p$, we have that
 \begin{align}\label{mydec}
 \sqrt{n}(\onth  -\theta_0) & =  W_n + \Delta_n, 
  \end{align}
  where $\E\{W_n\} = 0$ and 
  \begin{align}
  \P\Big\{ \norm{\Delta_n}_\infty \ge  \frac{C_2\kappa^2}{\lambdalbd^{3/2}} \frac{\sigma s_0\log p}{\sqrt{n}}\Big\} &\le p^{-3} . 
  \end{align}
  Further  we have
  \begin{align}
  \lVert\E\{\onth -\theta_0\}\rVert_\infty \le   \frac{C_2\kappa^2 }{\lambdalbd^{3/2}} \frac{\sigma s_0\log p}{n} + \frac{C_3\norm{\theta_0}_1}{p^2}\,. 
  \end{align}
 \end{theorem} 
The proof of Theorem~\ref{thm:batchbiasbound} is given in Appendix~\ref{proof:thm:batchbiasbound}. 
Note that, in the high-dimensional setting of $n \ll p$, 
 the term $\norm{\theta_0}_1/p^2$ will be of lower order as compared to
 $s_0\log p /n$. Therefore, when 
the parameters $\lambdalbd, \sigma, \kappa$ are of 
order one, 
the theorem shows that the bias of the online
 debiased estimator is of order $s_0 \log p/ n $, 
 This may be compared with the LASSO
 estimator $\htheta^\sL$ whose bias is typically
 of order $\lambda \asymp \sigma\sqrt{\log p/n}$.
 In particular, in the regime when $s_0 = o(\sqrt{n/\log p})$, 
 this bias is asymptotically dominated by the variance, 
 which is of order $\sigma/\sqrt{n}$.

 In order to establish asymptotic Gaussian behavior of the online debiased estimate $\onth$, 
 we consider a specific asymptotic regime for the 
 problem instances.

 \begin{assumption}(Asymptotic regime)\label{assmp:batchasymp}
We consider problem instances indexed by the sample size $n$,
where $n, p, s_0$ satisfy the following:
\begin{enumerate}
\item $\lim\inf_{n\to\infty} \frac{n_1\wedge n_2}{n} \ge  c$, for a positive universal constant $c\in (0,1]$. In other words, both batches contain at least a fixed fraction of data points.
\item The parameters satisfy:
\begin{align}
\lim_{n\to \infty} \frac{1}{\lambdalbd} s_0 \sqrt{\frac{\log p}{n}} \left(L^2K \vee \sqrt{\frac{\log p}{\Lambda_0}}\right) = 0\,.
\end{align}
\end{enumerate}

 \end{assumption}

The following proposition establishes that in the asymptotic regime, the unbiased
component $W_n$ has a Gaussian limiting
distribution. The key underlying technical
idea is to ensure that the martingale sum in $W_n$
is stable in an appropriate sense. 

\begin{proposition}\label{prop:batchvarianceclt}
Suppose that Assumption \ref{assmp:batchCov} holds
and consider
the asymptotic regime of Assumption \ref{assmp:batchasymp}. 
Let $a = a(n) \in[p]$ be a fixed
sequence of coordinates. Define the conditional
variance $V_{n, a}$ of the $a^\th$ coordinate
as
\begin{align}\label{eq:batch-Vn}
 V_{n, a} &= \sigma^2 \Big(\frac{n_1}{n} \<m^{(1)}_a, \hSigma^{(1)} m^{(1)} _a\>  + \frac{n_2}{n} \<m^{(2)}_a, \hSigma^{(2)} m^{(2)}_a\> \Big)\,.
 \end{align} 
Then, for any bounded continuous $\varphi:\reals\to\reals$ 
\begin{align*}
\lim_{n\to\infty}\E\Big\{ \varphi\Big(\frac{W_{n, a}}{\sqrt{V_{n, a}}} \Big) \Big\}
&= \E\{\varphi (\xi)\}, 
\end{align*}
where $\xi\sim\normal(0, 1)$.
The same holds for $\varphi$ being a step
function $\varphi(z) = \ind(z \le x)$ for any $x\in \reals$.
In particular, 
\begin{align*}
\lim_{n\to\infty} 
\P\Big\{   \frac{W_{n, a}}{\sqrt{V_{n, a}}} \le x \Big\} &=
\Phi(x),
\end{align*}
where $\Phi$ is the standard Gaussian cdf. 
\end{proposition}
The proof of Proposition~\ref{prop:batchvarianceclt} is deferred to 
Appendix~\ref{proof:prop:batchvarianceclt}.
The combination of Theorem \ref{thm:batchbiasbound}
and Proposition \ref{prop:batchvarianceclt} immediately
yields the following distributional characterization for
$\onth$. 
\begin{theorem}\label{thm:batchdistchar}
Under Assumptions \ref{assmp:batchCov} and \ref{assmp:batchasymp}, 
the conclusion of Proposition \ref{prop:batchvarianceclt}
holds with $\sqrt{n}(\onth_a - \theta_0)$ in place of $W_n$. In particular, 
\begin{align}\label{eq:batch-dist}
\lim_{n\to\infty} 
\P\Big\{   \sqrt{\frac{n}{V_{n, a}}}(\onth_{a} - \theta_{0, a}) \le x \Big\} &=
\Phi(x),
\end{align}
where $V_{n, a}$ is defined as in Proposition \ref{prop:batchvarianceclt}. 

\end{theorem}

 To compare the sample size requirements
made for $\ell_1$-consistent estimation and those in Assumption \ref{assmp:batchasymp},
it is instructive to simplify to the case
when $\kappa, \phi_0, \lambdalbd$ are of
order one. Then $\ell_1$-consistency (Theorem \ref{thm:batchlassoerr} in Appendix~\ref{sec:batchproofs}) requires  
that $n_1 \vee n_2  = \Omega(s_0^2 \log p)$, i.e. at least
one of the batches is larger than $s_0^2\log p$. However, Theorem \ref{thm:batchdistchar} makes the same assumption on $n_1 \wedge n_2$, 
or both batches exceed $s_0^2 \log p$ in size. 
For online debiasing, this is the case of interest. Indeed if
$n_1 \gg n_2$ (or vice versa), we can apply offline debiasing 
to the larger batch to obtain a debiased estimate. Conversely, when
$n_1$ and $n_2$ are comparable as in Assumption \ref{assmp:batchasymp}, this `sample-splitting' approach leads to loss of power corresponding
to a constant factor reduction in the sample size. This is the setting addressed in Theorem \ref{thm:batchdistchar} via online debiasing.

\subsubsection{Revisiting the numerical example from Section~\ref{sec:example1}.}\label{sec:comparison}
 In the batched data example discussed in Section~\ref{sec:example1}, we observed that the classical offline debiasing fails in providing unbiased estimate of the true parameters. Here, we will repeat the same experiment and numerically characterize the distribution of the proposed online debiased estimator. 

Figure \ref{fig:batchexampleonlineofflinedebiasedlasso} (left panel) shows
the 
histogram of the entries  of online debiased estimator $\onth$ on the support of $\theta_0$ (blue) along with the corresponding histogram of entries of the debiased estimator $\offth$ (red). As we see for both
choices of $\hth^1$ (debiased LASSO and ridge estimate on the first batch), the online debiased estimator
$\onth$ is appropriately centered around the true coefficients.  

One can also split samples in the following way.
Since the second batch of data
was adaptively collected while the first batch was not, we can
compute a debiased estimate using only the first, non-adaptive batch: 
\def\offthone{\htheta^{\sf off, 1}}
\begin{align}\label{eq:off1}
 \offthone &\equiv \Lsth(y^{(1)}, X_1) + \frac{1}{n}
 \Omega X_1^\sT(y^{(1)} - X_1 \Lsth(y^{(1)}, X_1)). 
 \end{align} 
Figure \ref{fig:batchexampleonlineofflinedebiasedlasso} (right panel)
shows the histogram of the entries of $\offthone$ restricted to the support
of $\theta_0$, and the comparison with $\onth$. As can be expected, 
both $\offthone$ and $\onth$ are appropriately centered around the 
true coefficient 1. However, as is common with sample-splitting,  
$\offthone$ displays a larger variance and correspondingly loses
power in comparison with $\onth$ since it uses only half of the data. 
The power loss becomes even more pronounced when there are
more than two phases of data collection, or if the phases
are particularly imbalanced.

\medskip

\noindent{\bf Comparison with ridge-type debiasing approach of \cite{deshpande2018accurate}.}
{
This work studies a similar problem, namely performing statistical inference using adaptively collected data using a debiasing approach. To compare with our setting, there are two important points to note:
\begin{enumerate}
\item The method of \cite{deshpande2018accurate} is tailored to low-dimensional setting where the number of covariates $p$ is less than the sample size ($p<n$). More specifically, denoting by $\lambda_{\min}(n)$ the minimum eigenvalue of $X^\sT X$, \cite{deshpande2018accurate} considers a setting where $\lambdamin(n) \to \infty$ almost surely. Note that for the batched data example, this amounts to $\sqrt{n}-\sqrt{p} \to \infty$.
\item The work \cite{deshpande2018accurate} proposes a different method of debiasing which albeit being valid in low-dimensional setting it comes with fundamental challenges to be generalized to high-dimensional setting.
Letting $\ols$ the least square estimator, \cite{deshpande2018accurate} constructs a debiased estimator $\dth$ as follows:
\begin{align}\label{ridge-online}
\dth = \ols + W_n(y- X\ols)\,,
\end{align}
where the matrix $W_n$ is constructed recursively as $W_n = [W_{n-1}| w_n]$ and $X_n = [X_{n-1} | x_{n}]$ with
\begin{align}\label{ridge-online}
w_n = \arg\min_{w\in \reals^p} \|I - W_{n-1} X_{n-1} - w x_n^\sT\|_F^2 + \lambda \|w\|_2^2\,.
\end{align}
Therefore, the decorrelating matrix $W_n$ is constructed in an online way as it is a predictable sequence according to Definition~\ref{def:pred}. Note that $w_i$ corresponds to $M_i x_i$ in our notation. 
\end{enumerate}
 
One can potentially think of using the ridge-type debiased estimator~\eqref{ridge-online} in high-dimensional setting with using $\htheta^\sL$ instead of $\ols$. In Figure~\ref{fig:batchexampleonlineofflinedebiasedlasso}, we include the histogram of such estimate (gray histogram under the name ``ridgeOnline"). As we see the corresponding histogram is biased and deviates from a normal distribution which implies that this approach does not extend to high-dimensional setting.

Some intuition for this may be seen by following the argument
of \cite{deshpande2018accurate}.
Considering the bias-variance decomposition of $\dth - \theta_0 = {\sf b} + {\sf v}$ with ${\sf b} = (I - W_n X_n) (\ols - \theta_0)$ and $ {\sf v} = W_n \eps_n$, the above optimization aims at minimizing a weighted sum of the bias and the variance of $\dth$ in an online manner. The analysis of~\cite{deshpande2018accurate} controls bias as follows
\[
\|{\sf b}\| \le  \|I - W_n X_n\|_{\rm op}\; \|\ols - \theta_0\|_2 \le
\|I - W_n X_n\|_{F}\; \|\ols - \theta_0\|_2\,.
\]
 
  However, in high-dimension this bound is vacuous. Since $W_nX_n\in \reals^{p\times p}$ is of rank at most $n<p$,  $I-W_nX_n$ has eigenvalue 1 with multiplicity at least $p-n$. Therefore $\|I-W_nX_n\|_F\ge p-n \to \infty$ and $ \norm{I - W_n X_n}_{\rm op} \ge 1$. 
 Thus, even a refinement of \cite{deshpande2018accurate} would only yield an insufficient bias bound of the type
  \begin{align*}
   \norm{{\sf b}}_2 &\le \norm{\hth^\sL - \theta_0}_2 \approx \sigma\sqrt{\frac{s_0 \log p}{n}}\,,
   \end{align*} 
%
which dominants the variance component $\Var({\sf v}) = O(1/\sqrt{n})$.
Our scheme of online debiasing overcomes this obstacle by adapting to the geometry of the high-dimensional regime. In particular, it yields the bias bound of order $\|\E\{\onth -\theta_0\}\|_\infty = O(s_0(\log p)/n)$ which is dominated by the noise term, provided that $s_0 = o(\sqrt{n}/\log p)$.

}

\begin{figure}[]
\centering{
\begin{subfigure}{0.9\linewidth}
\centering
\includegraphics[scale=0.6]{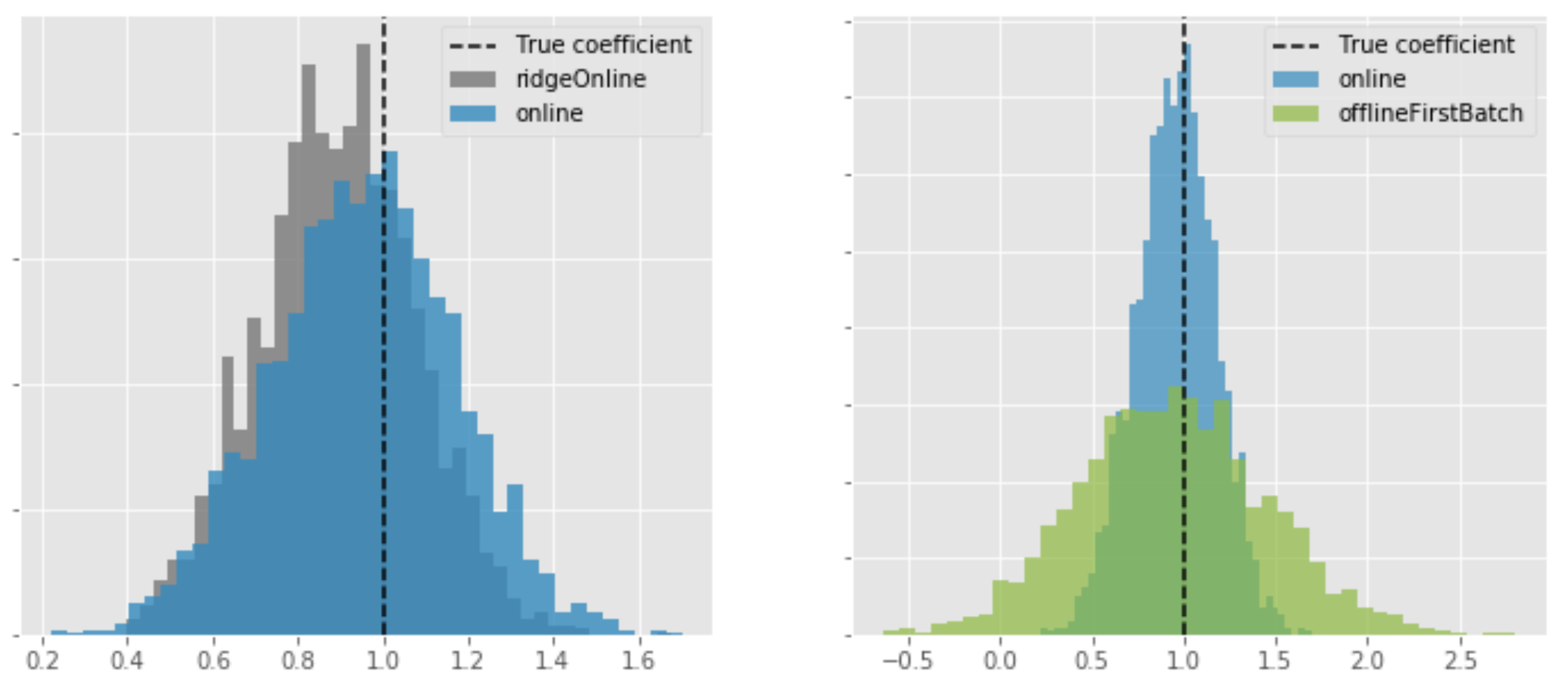}
\caption{with $\hth^1$ the debiased LASSO on first batch}
\end{subfigure}
\vspace{0.2cm}

\begin{subfigure}{0.9\linewidth}
\centering
\includegraphics[scale =0.6]{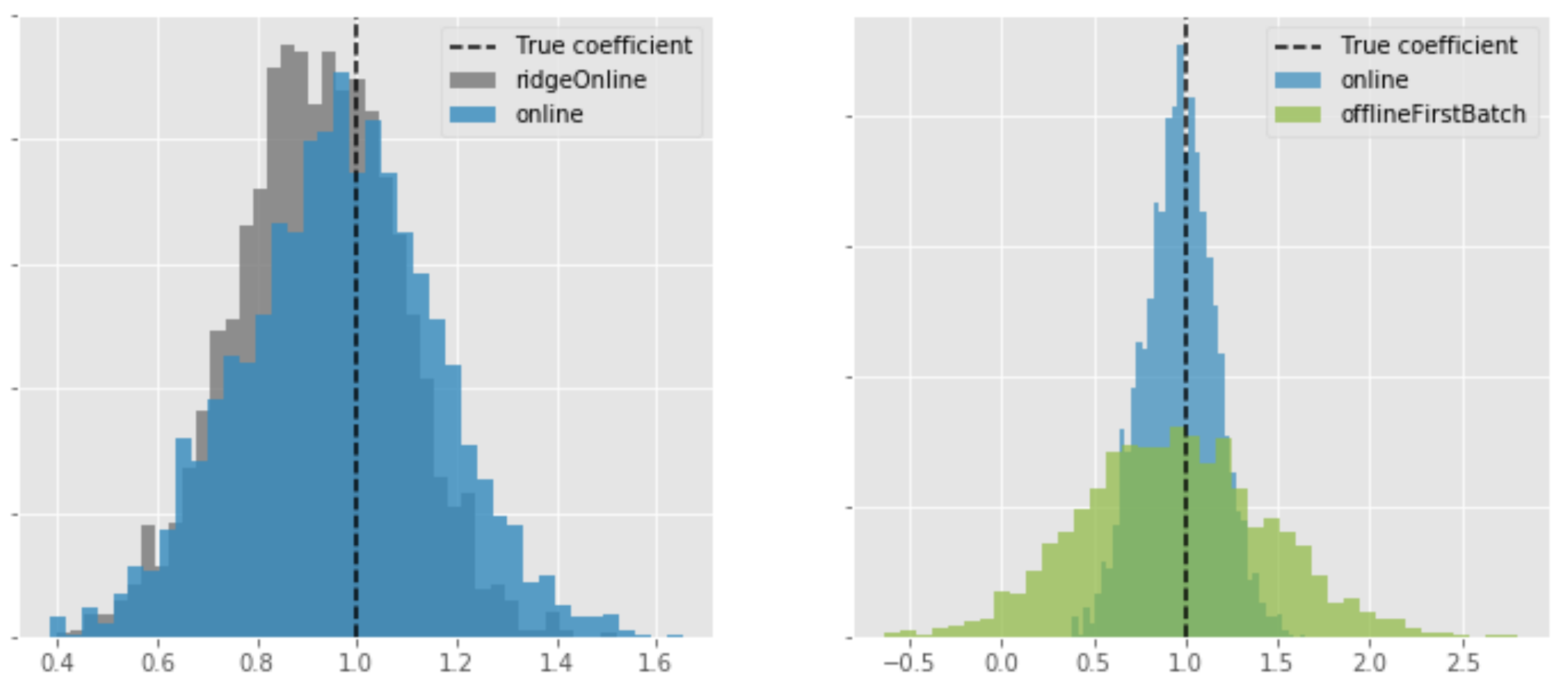}
\caption{with $\hth^1$ the ridge estimate on first batch}
\end{subfigure}
}
\caption{{\small (Left) Histograms of the online debiased estimate $\onth$ and the ridge debiased estimator~\cite{deshpande2018accurate}, restricted to the support
of $\theta_0$. (Right) Histograms of the offline debiased estimate
\emph{only using the first batch}, $ \offthone$ given by~\eqref{eq:off1} and the online debiased
estimate $\onth$. The dashed line indicates the true coefficient size. Offline debiasing $\offthone$ using only the first batch  works well (green histograms called offlineFirstBatch), but then loses power in comparison. Online debiasing is cognizant 
of the adaptivity and debiases without losing power even in the
presence of adaptivity.}  \label{fig:batchexampleonlineofflinedebiasedlasso} }
\end{figure}

\section{Statistical inference}\label{sec:inference}
An immediate use of distributional characterizations~\eqref{eq:DC-TS0} or \eqref{eq:batch-dist} is to construct confidence intervals and also provide valid p-values for hypothesis testing regarding the model coefficients. Throughout, we make the sparsity assumption $s_0 = o(\sqrt{n}/\log p_0)$, with $p_0$ the number of model parameters (for the batched data collection setting $p_0 = p$, and for the $\VAR(d)$ model $p_0 = dp$).

\begin{description}
\item[Confidence intervals:] For fixed coordinate $a\in [p_0]$ and significance level $\alpha\in (0,1)$, we let
\begin{align}
J_a(\alpha) &\equiv [\onth_a - \delta(\alpha,n), \onth_a + \delta(\alpha,n)]\,,\\
\delta(\alpha,n) &\equiv \Phi^{-1}(1-\alpha/2) \sqrt{V_{n,a}/n}\,,
\nonumber
\end{align} 
where $V_{n,a}$ is defined by Equation \eqref{eq:conditionalvar} for the $\VAR(d)$ model and by Equation~\eqref{eq:batch-Vn} for the batched data collection setting.

As a result of Proposition~\ref{pro:SS}, the confidence interval $J_a(\alpha)$ is asymptotically valid because
\begin{align}
\begin{split}\label{eq:CI-ind}
\lim_{n\to\infty} \prob (\theta_{0,a}\in J_a(\alpha)) &= \lim_{n\to\infty} \prob\Big\{\frac{\sqrt{n}(\onth_a - \theta_{0,a})}{\sqrt{V_{n,a}}} \le \Phi^{-1} (1-\alpha/2) \Big\}\\
& - \lim_{n\to\infty} \prob\Big\{\frac{\sqrt{n}(\onth_a - \theta_{0,a})}{\sqrt{V_{n,a}}} \le \Phi^{-1} (1-\alpha/2) \Big\}\\
&= \Phi(\Phi^{-1}(1-\alpha/2)) - \Phi(-\Phi^{-1}(1-\alpha/2)) = 1 - \alpha\,.
\end{split}
\end{align}
Further, note that the length of confidence interval $J_a(\alpha)$ is of order $O(\sigma/\sqrt{n})$ (using Lemma~\ref{lem:batchstability} for the batched data collection setting and  Lemma \ref{lem:stab-W-TS} for the time series). It is worth noting that this is the minimax optimal rate~\cite{javanmard2014hypothesis,javanmard2014inference} and is of the same order of the length of confidence intervals obtained by the least-square estimator for the classical regime $n>p$ with i.i.d samples. 
\item[Hypothesis testing:] Another consequence of Proposition~\ref{pro:SS} is that it allows for testing hypothesis of form $H_0: \theta_{0,a} = 0$ versus the alternative $H_A:\theta_{0,a}\neq 0$ and provide valid $p$-values. Recall that $\theta_0$ denotes the model parameters, either for the batched data collection setting or the $\VAR(d)$ model (which encodes the entries $A^{(\ell)}_{i,j}$ in model~\eqref{eq:varddef}).
Such testing mechanism is of crucial importance in practice as it allows to diagnose the significantly relevant  covariates to the outcome. In case of time series,  it translates to understanding the effect of a covariate $z_{t-\ell,j}$ on a covariate $z_{t,i}$, and to provide valid statistical measures ($p$-values) for such associations.
We construct two-sided $p$-values for testing $H_0$, using our test statistic as follows:
\begin{align}\label{eq:p-value}
P_a = 2\left(1 - \Phi\left(\frac{\sqrt{n}|\onth_a|}{ \sqrt{V_{n,a}}}\right) \right)\,.
\end{align}
Our testing (rejection) rule given the p-value $P_a$ is:
\begin{align}
R(a) = \begin{cases}
1 \quad &\text{ if } P_a\le \alpha \quad (\text{reject } H_{0})\,,\\
0\quad &\text{otherwise} \quad \;(\text{fail to reject } H_0)\,.
\end{cases} 
\end{align}
Employing the distributional characterizations~\eqref{eq:batch-dist} or \eqref{eq:DC-TS0}, it is easy to verify that the constructed p-value $P_a$ is valid in the sense that under the null hypothesis it admits a uniform distribution: $\prob_{\theta_{0,a} = 0}(P_a\le u)$ $= u$ for all $u\in [0,1]$.

\item[Group inference] In many applications, one may want to do inference for a group of model parameters, $\theta_{0,G} \equiv (\theta_{0,a})_{a\in G}$ simultaneously, rather than the individual inference. This is the case particularly, when the model covariates are highly correlated with each other or they are likely to affect the outcome (in time series application, the future covariate vectors) jointly.   

To address group inference, we focus on the time series setting. The setting of batched data collection can be handled in a similar way. We first state a simple generalization of Proposition~\ref{pro:SS} to a group of coordinates with finite size as $n, p \to \infty$. The proof is very similar to the proof of Proposition~\ref{pro:SS} and is omitted.
\begin{lemma}\label{lem:GroupInf}
Let $G = G(n)$ be a sequence of sets $G(n)\subset [dp]$ with $|G(n)| = k$ fixed as $n,p\to \infty$. Also, let the conditional variance $V_n\in\reals^{dp\times dp}$ 
be defined by \eqref{eq:conditionalvar} for the $\VAR(d)$ model, that is:
\begin{align}
V_n \equiv \frac{\sigma^2}{n}\Bsum (\Mell x_t) (\Mell x_t)^\sT\,.
\end{align}
Under the assumptions of Proposition~\ref{pro:SS}, for all $u = (u_1, \dotsc, u_k)\in \reals^k$ we have
 \begin{align}\label{eq:DC-TS-G}
 \lim_{n\to\infty} \bigg| \prob\left\{\sqrt{n} (V_{n,G})^{-1/2} (\onth_G - \theta_{0,G}) \le u \right\}  - \Phi_k(u)\bigg| = 0\,,
 \end{align}
 where $V_{n,G} \in \reals^{k\times k}$ is the submatrix obtained by restricting $V_n$ to the rows and columns in $G$. Here $(a_1, \dotsc, a_k)\le (b_1, \dotsc, b_k)$
 indicates that $a_i\le b_i$ for $i\in [k]$ and $\Phi_k(u) = \Phi(u_1) \dotsc \Phi_k(u)$.
\end{lemma}
Much in the same way as individual inference, we can use Lemma~\ref{lem:GroupInf} for simultaneous inference on a group of parameters. Concretely, let $\cS_{k,\alpha}\subseteq \reals^k$ be any Borel set with $k$-dimensional Gaussian measure at least $1-\alpha$. Then for a group $G\subset[dp]$, with size $|G| = k$, we construct the confidence set $J_G(\alpha) \subseteq \reals^k$ as follows
\begin{align}
J_G(\alpha)\equiv \onth_G + \frac{1}{\sqrt{n}}(V_{n,R})^{1/2} \cS_{k,\alpha}\,.
\end{align}
Then, using Lemma~\ref{lem:GroupInf} (along the same lines in deriving~\eqref{eq:CI-ind}), we conclude that $J_G(\alpha)$ is a valid confidence region, namely 
\begin{align}
\lim_{n\to \infty} \prob(\theta_{0,G}\in J_G(\alpha)) = 1-\alpha\,.
\end{align}
\end{description} 

 \section{Numerical experiments}\label{sec:numerical}
 In this section, we evaluate the performance of online debiasing framework on synthetic data. In the interest of reproducibility, an {{\sf R}} implementation of our algorithm
is available at {\small \url{http://faculty.marshall.usc.edu/Adel-Javanmard/OnlineDebiasing}}.

 Consider the $\VAR(d)$ time series model \eqref{eq:varddef}. In the first setting, we let $p=20$, $d=3$, $T=50$ and construct the covariance matrix of noise terms $\Sigma_\zeta$ by putting $1$ on its diagonal and $\rho=0.3$ on its off-diagonal. To make it closer to the practice, instead of considering sparse coefficient matrices, we work with \emph{approximately} sparse matrices. Specifically, the entries of $A^{(i)}$ are generated independently from a Bernoulli distribution with success probability $q=0.1$, multiplied by $b\cdot \text{Unif}(\{+1,-1\})$ with $b=0.1$, and then added to a Gaussian matrix with mean $0$ and standard error $1/p$. In formula, each entry is generated independently from  
 $$b\cdot\text{Bern}(q)\cdot\text{Unif}(\{+1,-1\})+\mathcal{N}(0,1/p^2)\,.$$ 
We used $r_0=6$ (length of first episode $E_0$) and $\beta=1.3$ for lengths of other episodes $E_\ell\sim \beta^\ell$. 
For each $i\in [p]$ we do the following. Let $\theta_0 = (A^{(1)}_i, A^{(2)}_i, \dots, A^{(d)}_i)^\sT \in \reals^{dp}$ encode the $i^{{\rm th}}$ rows of the matrices $A^{(\ell)}$ and  compute the noise component of $\onth$ as
\begin{align}
W_n&\equiv  \frac{1}{\sqrt{n}} \sum_{\ell=0}^{K-1} \Mell \Big(\sum_{t\in E_{\ell}}  x_t \varepsilon_t\Big)\,,
\end{align}
the rescaled residual $T_n\in \reals^{dp}$  with $T_{n,a} = \sqrt{\frac{n}{V_{n,a}}} (\onth_a-\theta_{0,a})$, and $V_{n,a}$ given by Equation \eqref{eq:conditionalvar} and $\sigma=1$. Left and right plots of Figure \ref{fig:all-coord-plots} denote the QQ-plot, PP-plot and histogram of noise terms $W_n$ and rescaled residuals $T_n$ of \emph{all coordinates} (across all $i\in [p]$ and $a\in [dp]$) stacked together, respectively.

\begin{figure}
\centering
\begin{subfigure}{0.5\textwidth}
  \centering
  \includegraphics[scale =0.23]{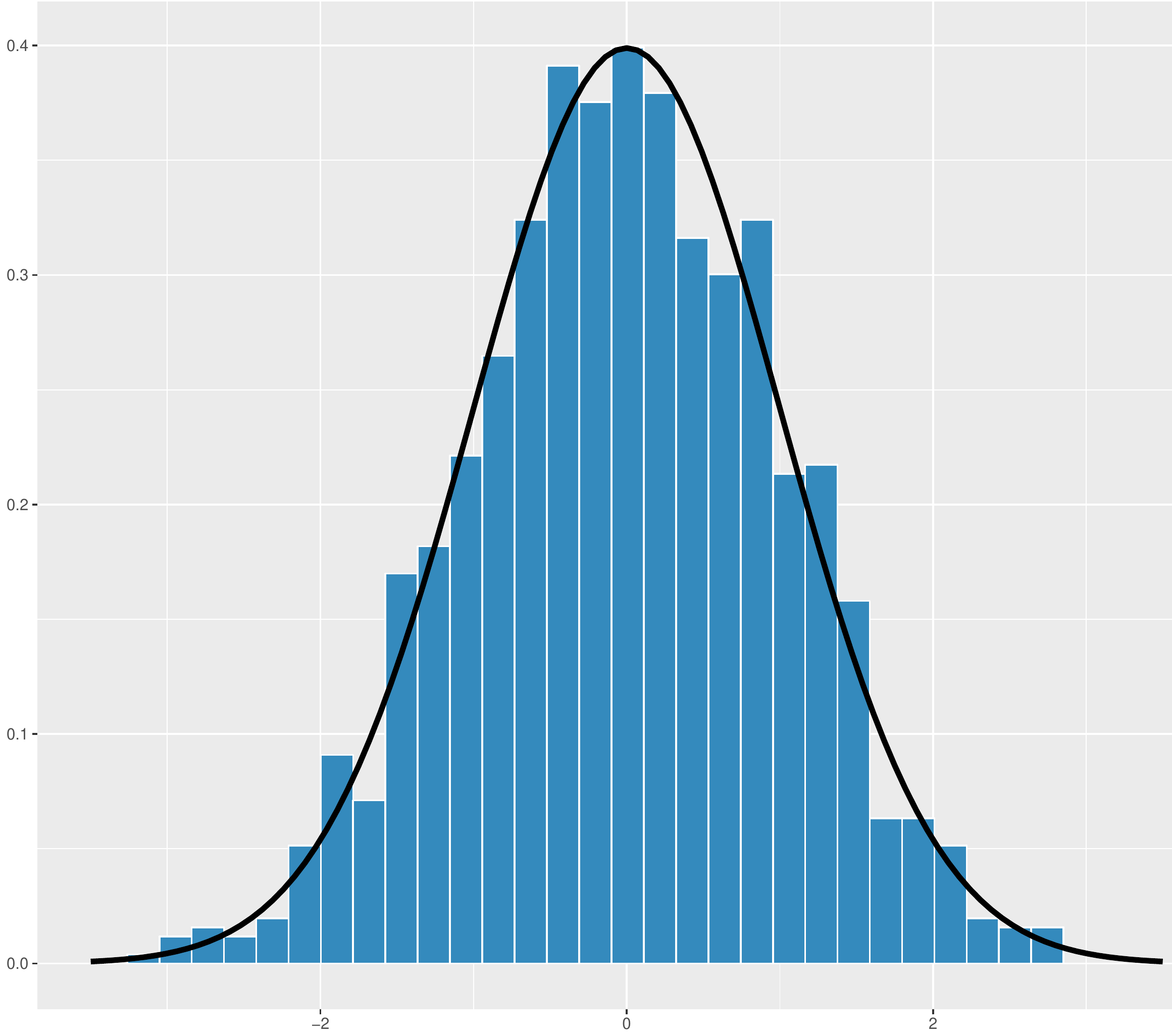}
  \put(-165,60){\rotatebox{90}{\scriptsize{Density}}}
 \put(-105,-7){\rotatebox{0}{\scriptsize{Noise Terms}}}
  \caption{Histogram of Noise Terms $W_n$}
  \label{fig:all-coord:sub1}
\end{subfigure}%
\begin{subfigure}{.5\textwidth}
  \centering
  \includegraphics[scale=0.23]{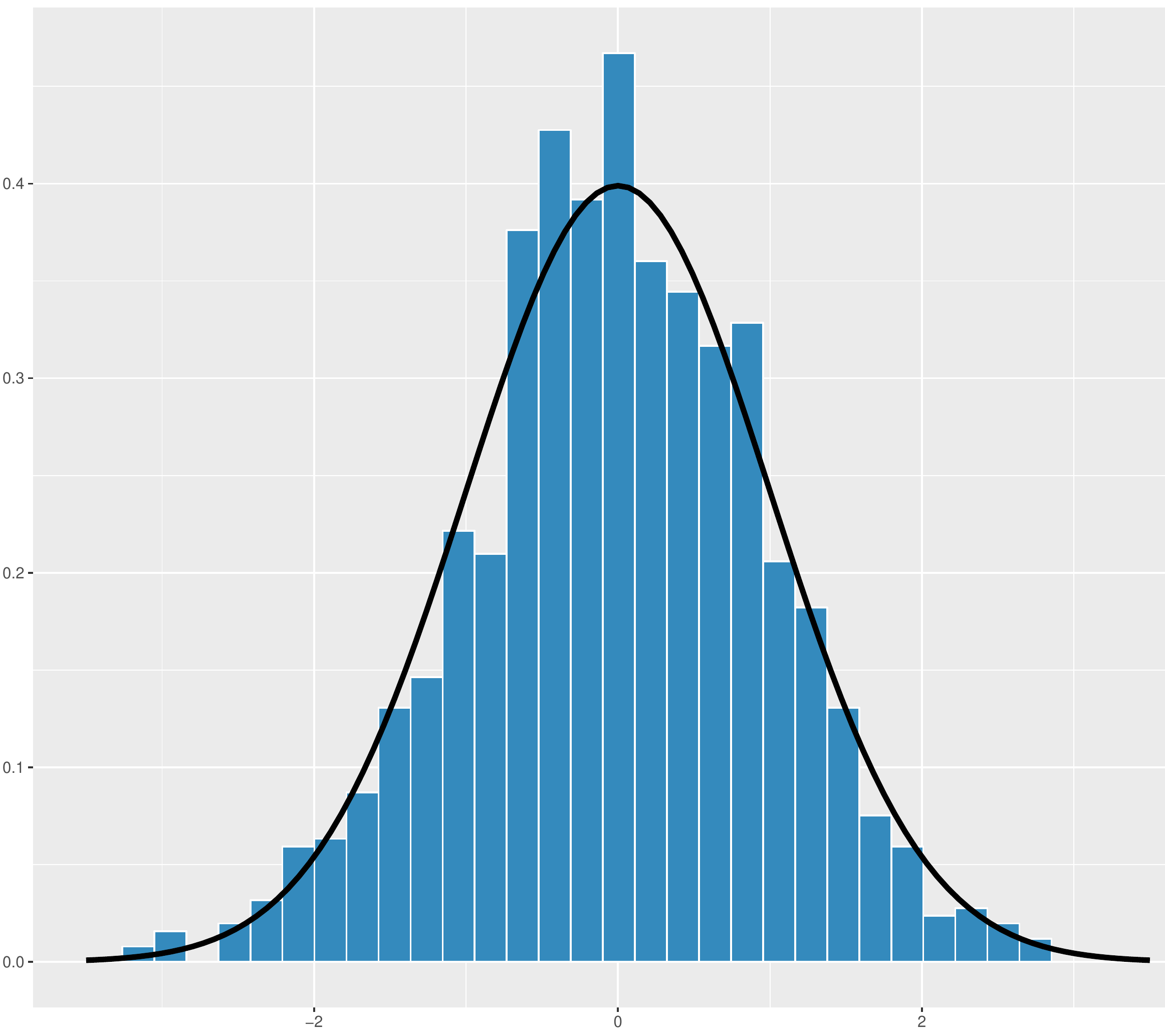}
 \put(-165,60){\rotatebox{90}{\scriptsize{Density}}}
 \put(-105,-7){\rotatebox{0}{\scriptsize{Rescaled Residuals}}}
  \caption{Histogram of Residuals $T_n$}
  \label{fig:all-coord:sub2}
\end{subfigure}

\begin{subfigure}{.5\textwidth}
  \centering
  \includegraphics[scale =0.23]{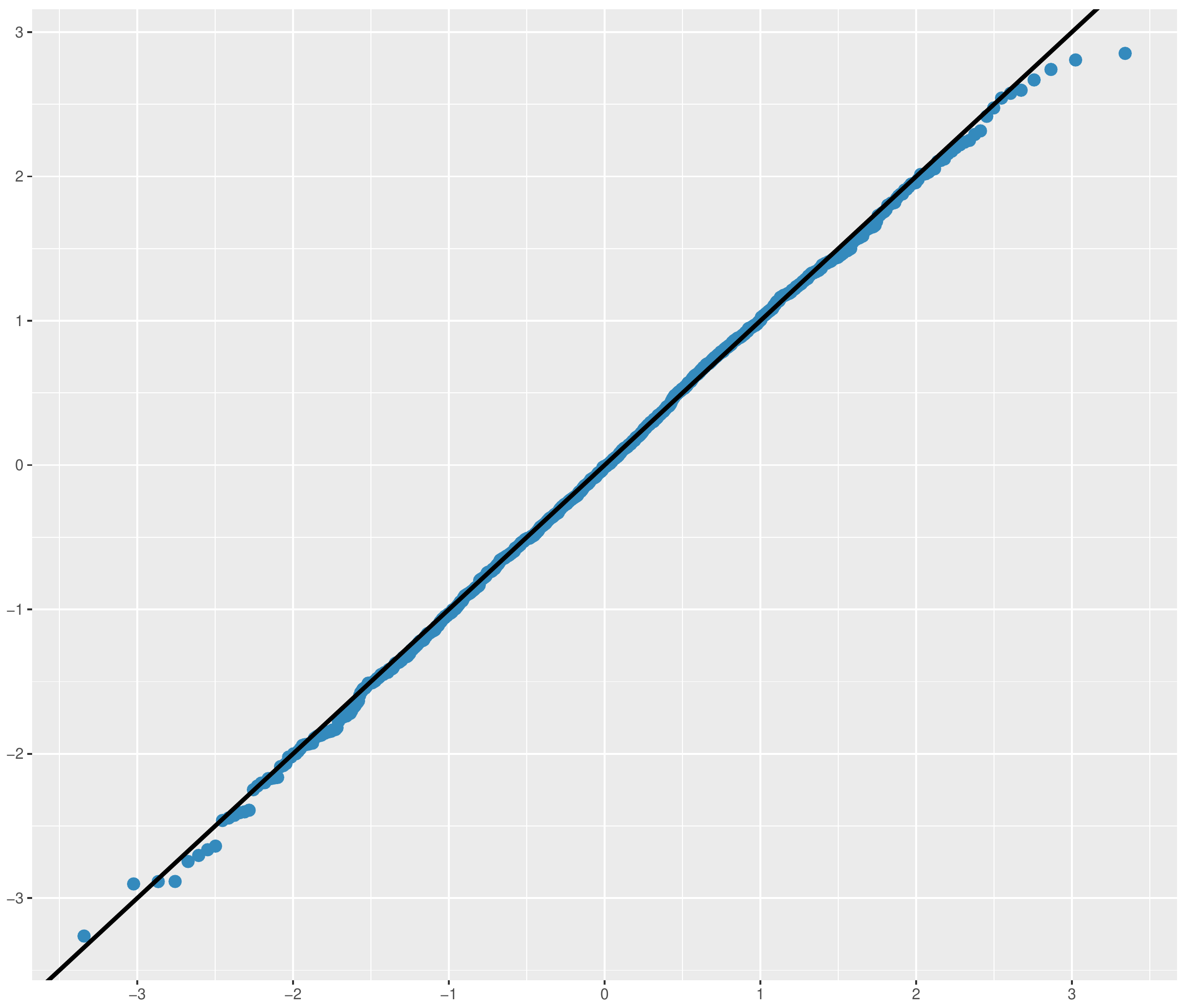}
  \put(-170,60){\rotatebox{90}{\scriptsize{Sample}}}
 \put(-105,-7){\rotatebox{0}{\scriptsize{Theoretical}}}
  \caption{QQ plot of Noise Terms $W_n$}
  \label{fig:all-coord:sub3}
\end{subfigure}%
\begin{subfigure}{.5\textwidth}
  \centering
  \includegraphics[scale =0.23]{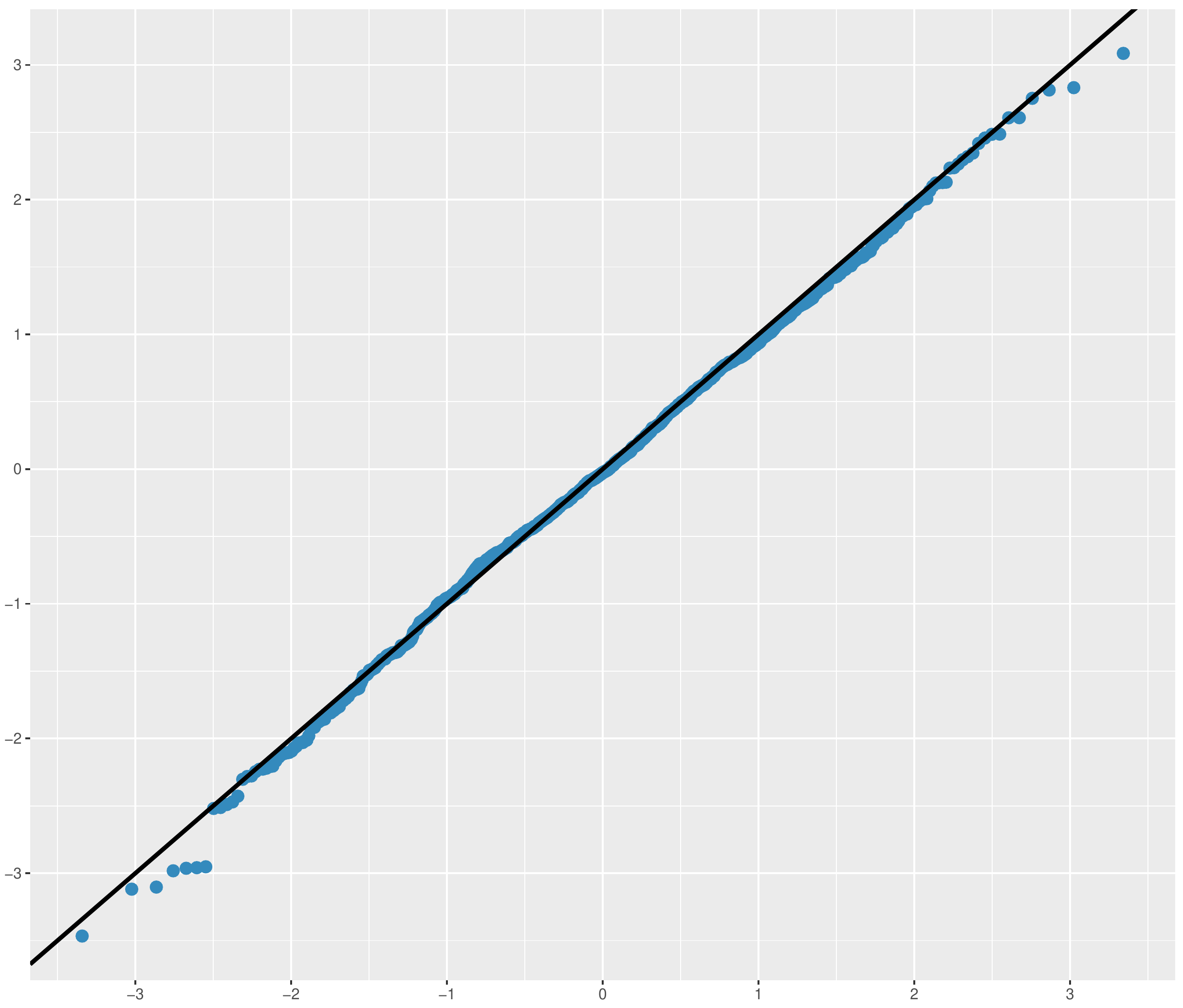}
 \put(-170,60){\rotatebox{90}{\scriptsize{Sample}}}
\put(-105,-7){\rotatebox{0}{\scriptsize{Theoretical}}}
  \caption{QQ plot of Residuals $T_n$}
  \label{fig:all-coord:sub4}
\end{subfigure}
\begin{subfigure}{.5\textwidth}
  \centering
  \includegraphics[scale =0.23]{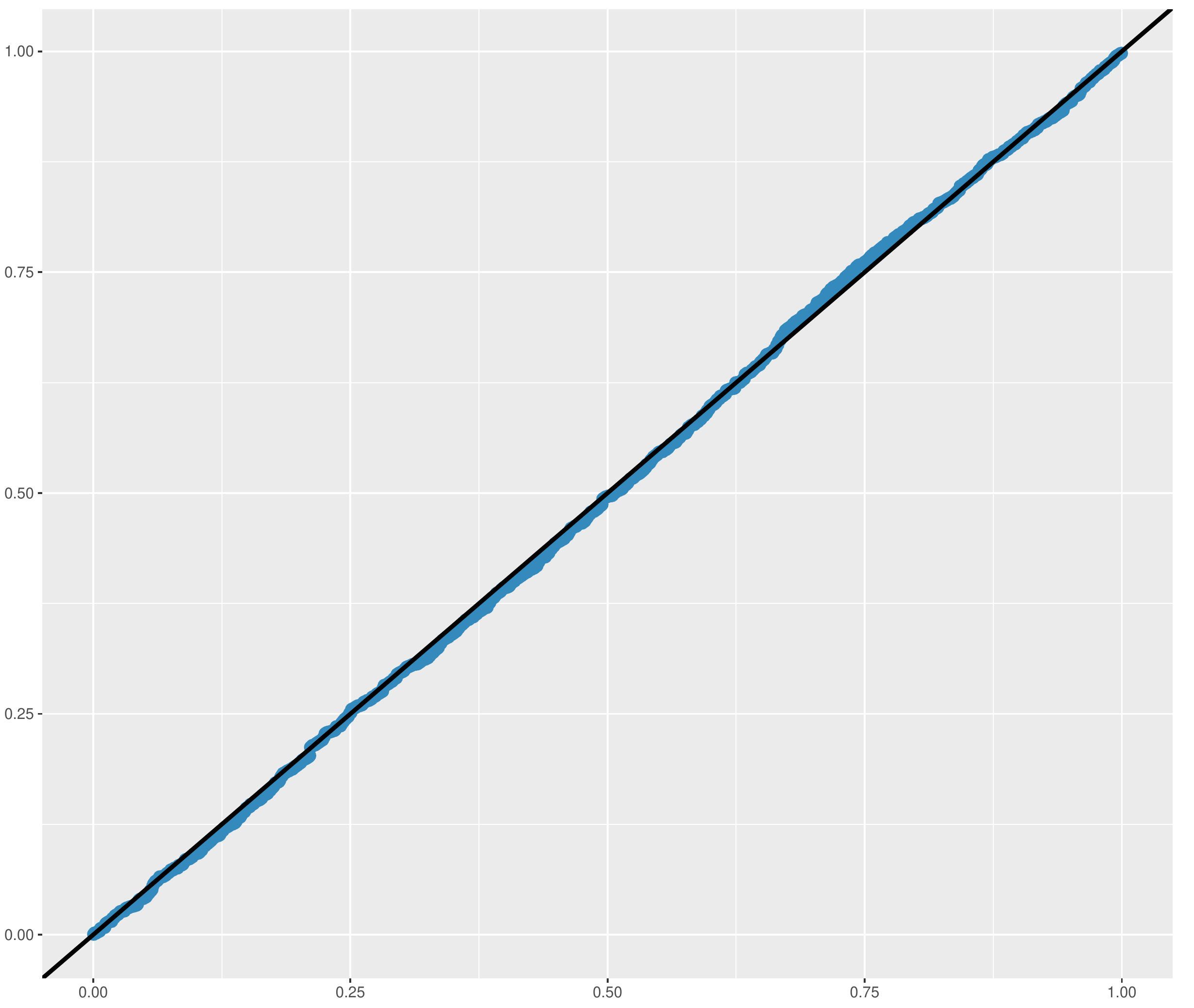}
 \put(-170,60){\rotatebox{90}{\scriptsize{Sample}}}
\put(-105,-7){\rotatebox{0}{\scriptsize{Theoretical}}}  
\caption{PP plot of Noise Terms $W_n$}
  \label{fig:all-coord:sub5}
\end{subfigure}%
\begin{subfigure}{.5\textwidth}
  \centering
  \includegraphics[scale =0.23]{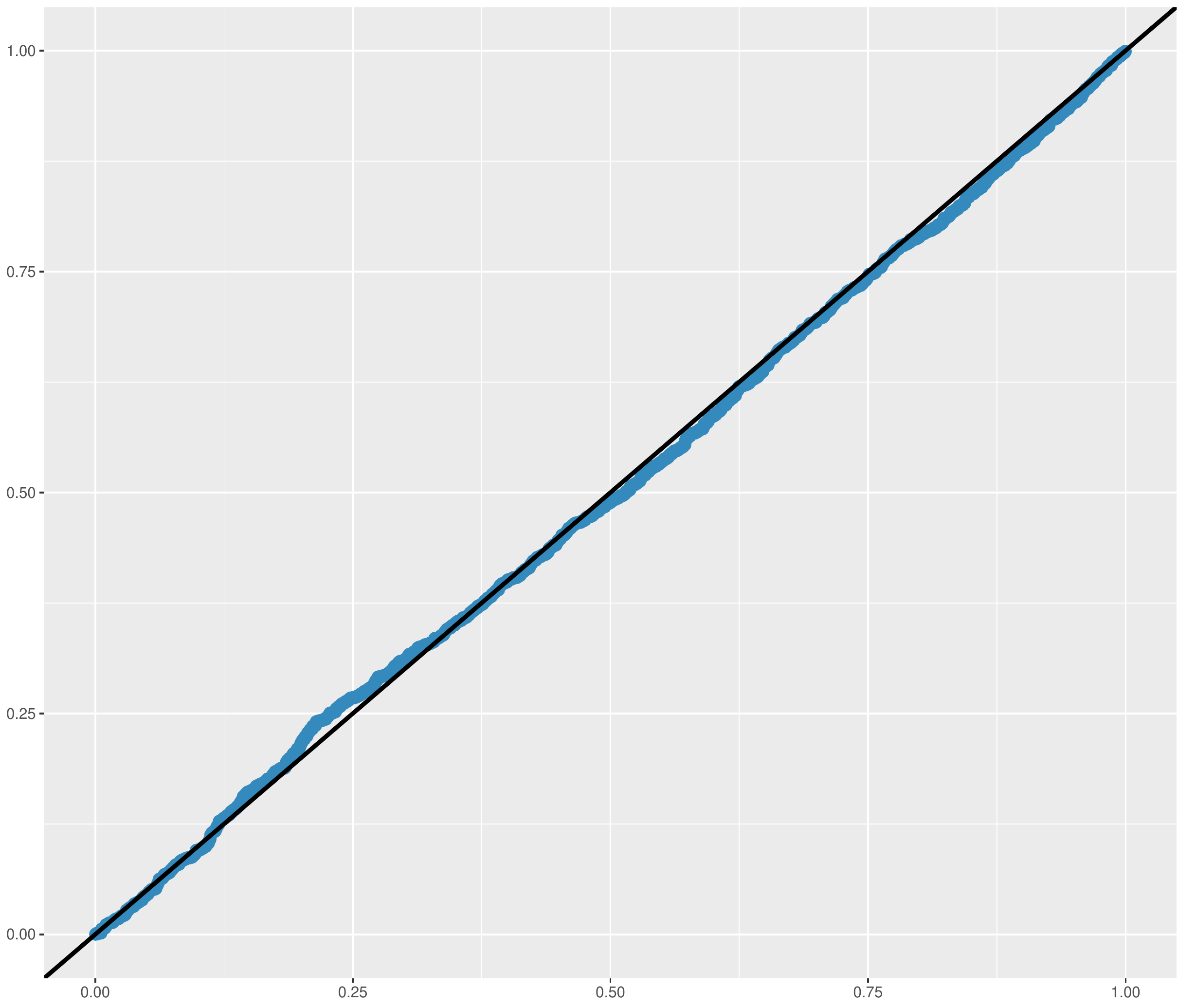}
 \put(-170,60){\rotatebox{90}{\scriptsize{Sample}}}
\put(-105,-7){\rotatebox{0}{\scriptsize{Theoretical}}}
  \caption{PP plot of Residuals $T_n$}
  \label{fig:all-coord:sub6}
\end{subfigure}
\caption{{\small A simple example of an online debiased $\var(3)$ process with dimension $p=20$ and $T=50$ sample data points. Plots \ref{fig:all-coord:sub1}, \ref{fig:all-coord:sub3}, \ref{fig:all-coord:sub5} demonstrate respectively the histogram, QQ-plot, and PP plot of noise values of all $dp^2=1200$ entries of $A_i$ matrices in linear time series model \eqref{eq:varddef}. Plots \ref{fig:all-coord:sub2}, \ref{fig:all-coord:sub4}, \ref{fig:all-coord:sub6} are histogram, QQ-plot, and PP-plot of rescaled residuals of all coordinates as well. Alignment of data points in these plots with their corresponding standard normal $(0,1)$ line corroborates our theoretical results on the asymptotic normal behavior of noise terms and rescaled residuals discussed in corollary \ref{cor:noise-TS} and proposition \ref{pro:SS}, respectively.} }
\label{fig:all-coord-plots}
\end{figure}
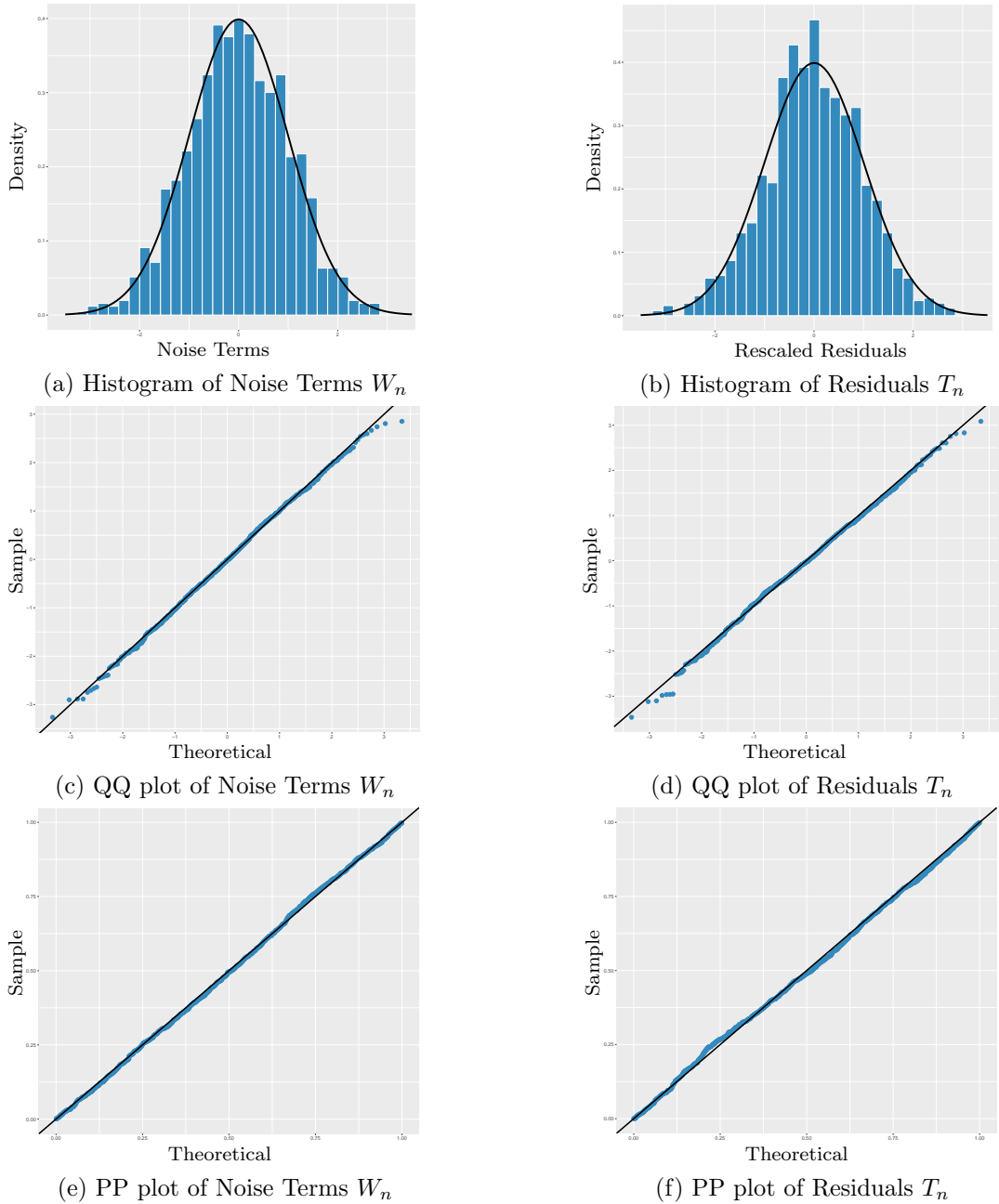

\medskip

{\bf True and False Positive Rates.} Consider the linear time-series model \eqref{eq:varddef} with $A^{(i)}$ matrices having entries drawn independently from the distribution $b\cdot \text{Bern}(q)\cdot\text{Unif}(\{+1,-1\})$ and noise terms be gaussian with covariance matrix $\Sigma_{\zeta}$. In this example, we evaluate the performance of our proposed online debiasing method for constructing confidence intervals and hypothesis testing as discussed in Section~\ref{sec:inference}. 
We consider four metrics: True Positive Rate (TPR), False Positive Rate (FPR), Average length of confidence intervals (Avg CI length), and coverage rate of confidence intervals. Tables \ref{table-FPR-TPR-1} and \ref{table-FPR-TPR-2}  summarize the results for various configurations of the  $\var(d)$ processes and significance level $\alpha=0.05$.  Table \ref{table-FPR-TPR-1} corresponds to the cases where noise covariance has the structure $\Sigma_\zeta(i,j)= 0.1^{|i-j|}$ and Table \ref{table-FPR-TPR-2} corresponds to the case of $\Sigma_\zeta(i,j)= 0.1^{\ind(i\neq j)}$. The reported measures for each configuration (each row of the table) are average over 20 different realizations of the $\VAR(d)$ model.

\begin{table}[H]
\caption{Evaluation of the online debiasing approach for statistical inference on the  coefficients of a $\VAR(d)$ model under different configurations.  Here the noise terms $\zeta_i$ are gaussian with covariance matrix $\Sigma_{\zeta}(i,j)= 0.1^{|i-j|}$. The results are reported in terms of four metrics:  FPR (False Positive Rate), TPR (True Positive Rate), Coverage rate and Average length of confidence intervals (Avg CI length) at significance level $\alpha=0.05$ \label{table-FPR-TPR-1}}

{\small
\begin{tabular}{|c|c|c|c|c||c|c|c|c|}\hline
\backslashbox{$d$}{Parameters}
&\makebox[1em]{$p$}&\makebox[1em]{T}&\makebox[1em]{$q$}
&\makebox[3em]{$b$}&\makebox[3em]{FPR}&\makebox[3em]{TPR}
&\makebox[6em]{Avg CI length}&\makebox[6em]{Coverage rate}\\\hline
\multirow{3}{*}{$d=1$} &40&30&0.01&2&0.0276&1&3.56&0.9725\\
&35&30&0.01&2&0.0354&0.9166&3.7090&0.9648\\
&60&55&0.01&0.9&0.0314&0.7058&2.5933&0.9686\\
\hline
\multirow{3}{*}{$d=2$} &55&100&0.01&0.8&0.0424&0.8000&1.9822&0.9572\\
&40&75&0.01&0.9&0.0343&0.9166&2.5166&0.9656\\
&50&95&0.01&0.7&0.0368&0.6182&2.4694&0.963\\
\hline
\multirow{3}{*}{$d=3$} &45&130&0.005&0.9&0.0370&0.6858&2.070&0.9632\\
&40&110&0.01&0.7&0.0374&0.6512&2.1481&0.9623\\
&50&145&0.005&0.85&0.0369&0.6327&2.2028&0.9631\\
\hline
\end{tabular}
}
\end{table}
\begin{table}[H]
\caption{\label{table-FPR-TPR-2} Evaluation of the online debiasing approach for statistical inference on the  coefficients of a $\VAR(d)$ model under different configurations.  Here the noise terms $\zeta_i$ are gaussian with covariance matrix $\Sigma_{\zeta}(i,j)= 0.1^{\ind(i \neq j)}$. The results are reported in terms of four metrics:  FPR (False Positive Rate), TPR (True Positive Rate), Coverage rate and Average length of confidence intervals (Avg CI length) at significance level $\alpha=0.05$}
{\small
\begin{tabular}{|c|c|c|c|c||c|c|c|c|}\hline
\backslashbox{$d$}{Parameters}
&\makebox[1em]{$p$}&\makebox[1em]{T}&\makebox[1em]{$q$}
&\makebox[3em]{$b$}&\makebox[3em]{FPR}&\makebox[3em]{TPR}
&\makebox[6em]{Avg CI length}&\makebox[6em]{Coverage rate}\\\hline
\multirow{3}{*}{$d=1$} &40&30&0.01&2&0.0402&1&3.5835&0.96\\
&40&35&0.02&1.2&0.0414&0.8125&2.6081&0.9575\\
&50&40&0.015&0.9&0.0365&0.7435&2.0404&0.9632\\
\hline
\multirow{3}{*}{$d=2$} &35&65&0.01&0.9&0.0420&0.8077&2.4386&0.9580\\
&45&85&0.01&0.9&0.0336&0.7298&2.5358&0.9655\\
&50&70&0.01&0.95&0.0220&0.8333&2.4504&0.9775\\
\hline
\multirow{3}{*}{$d=3$} &40&115&0.01&0.9&0.0395&0.7906&1.6978&0.9598\\
&45&130&0.005&0.95&0.0359&0.7714&2.1548&0.9641\\
&50&145&0.005&0.85&0.0371&0.5918&2.1303&0.9624\\
\hline
\end{tabular}
}
\end{table}

\subsection{Real data experiments: a marketing application}
Retailers often offer sales of various categories of products and for an effective management of the business, they need to understand the cross-category effect of products on each other, e.g.,  how the price, promotion or sale of category 
A will effect the sales of category B after some time.  

We used data of sales, prices and promotions of Chicago-area grocery store chain Dominick's that is publicly available at {\small \url{https://research.chicagobooth.edu/kilts/marketing-databases/dominicks}}. 
The same data set has been used in \cite{gelper2016identifying} where a sparse VAR model is fit to data and also in~\cite{wilms2017interpretable} where a VARX model is employed to estimate the demand effects (VARX models incorporate the effect of unmodeled exogenous variables (X) into the VAR). In this experiment, we use the proposed online debiasing approach to provide $p$-values for the category effects. 

We consider $11$ categories of products\footnote{Bottled Juices, Cereals, Cheeses, Cookies, Crackers, Canned Soup, Front-end-Candies, Frozen Juices, Soft Drinks, Snack Crackers and Canned Tuna} over 71 weeks, so for each week $t$, we have information $x_t\in\reals^{33}$ for sales, prices and promotions of the 11 categories. For thorough explanation on calculating sales, prices and promotions, we refer to  \cite{srinivasan2004promotions} and \cite{gelper2016identifying}. We posit $\VAR(2)$ model as the generating process for covariates $x_i$ and then apply our proposed online debiasing method to calculate two-sided $p$-values for the null hypothesis of form $H_0: \theta_{0,a} = 0$ with $\theta_{0,a}$ an entry in the $\VAR$ model, as discussed earlier in Section \ref{sec:inference} (See Eq. \eqref{eq:p-value}). We refer to Appendix~\ref{app:simulation} for the reports of the $p$-values. By running  the Benjamini–Yekutieli procedure  \cite{benjamini2001control} (with log factor correction to account for dependence among $p$-values), we obtain the following statistically significant cross category associations at level $0.05$: 
sales of canned tuna on sales of front-end-candies after one week with $p$-val= 5.8e-05,
and price of crackers on sales of canned tuna after one week with $p$-val= 5.5e-04.
In \cite{gelper2016identifying}, sparse VAR models are used to construct networks of interlinked 
product categories, but they are not accompanied by statistical measures such as $p$-values. Our online debiasing method here provides $p$-values for individual possible cross-category associations.


\section{Implementation and extensions}\label{sec:discussion}
\subsection{Iterative schemes to implement online debiasing}
The online debiased estimator~\eqref{eq:debias} involves the decorrelating matrices $\Mell$, whose rows $(\mli)_{a\in[dp]}$ are constructed by the optimization~\eqref{eq:opt}. 
%
For the sake of computational efficiently, it is useful to work with a Lagrangian equivalent version of this optimization. Consider the following optimization
\begin{align}\label{eq:opt-Lag}
\text{minimize}_{\|m\|_1\le L} \quad \frac{1}{2} m^\sT \hSigma^{(\ell)}m -\<m,e_a\> + \mu_\ell \|m\|_1\,,
\end{align}
with $\mu_\ell$ and $L$ taking the same values as in Optimization~\eqref{eq:opt}. 

The next result, from~\cite[Chapter 5]{javanmard2014inference} is on the connection between the solutions of the unconstrained problem \eqref{eq:opt-Lag} and~\eqref{eq:opt}. For the reader's convenience, the proof is also given in Appendix~\ref{proof-lem:Duality}.
\begin{lemma}\label{lem:Duality}
A solution of optimization~\eqref{eq:opt-Lag} is also a solution of the optimization problem~\eqref{eq:opt}. Also, if  problem~\eqref{eq:opt} is feasible then 
problem~\eqref{eq:opt-Lag} has bounded solution.
\end{lemma}
Using the above lemma, we can instead work with the Lagrangian version~\eqref{eq:opt-Lag} for constructing the decorrelating vector $\mli$.

Here, we propose to solve optimization problem~\eqref{eq:opt-Lag} using iterative method. 
Note the objective function evolves slightly at each episode and hence we expect the solutions $\mli$ and $m^{\ell+1}_a$ to be close to each other.
An appealing property of iterative methods is that we can leverage this observation by setting $\mli$ as the initialization for the iterations that compute $m^{\ell+1}_a$, yielding shorter convergence time. In the sequel we discuss two of such iterative schemes.
\subsubsection{Coordinate descent algorithms} 
In this method, at each iteration we update one of the coordinates of $m$, say $m_j$, while fixing the other coordinates. 
We write the objective function of \eqref{eq:opt-Lag} by separating $m_j$ from the other coordinates: 
\begin{align}\label{opt:CD}
\frac{1}{2} \hSigma^{(\ell)}_{j,j} m_j^2 + \sum_{r,s\neq j} \hSigma^{(\ell)}_{r,s}\; m_r m_s
- m_a + \mu_\ell \|m_{\sim j}\|_1 + \mu_\ell |m_j|\,,
\end{align}
where $\hSigma^{(\ell)}_{j,\sim j}$ denotes the $j^{\rm th}$ row (column) of $\hSigma^{(\ell)}$
with $\hSigma^{(\ell)}_{j,j}$ removed. Likewise, $m_{\sim j}$ represents  the restriction of $m$ to coordinates other than $j$.
Minimizing \eqref{opt:CD} with respect to $m_j$ gives
\[
m_j + \frac{1}{\hSigma^{(\ell)}_{j,j}} \left(\hSigma^{(\ell)}_{j,\sim j} m_{\sim j} - \ind(a=j) + \mu_\ell\, \sign(m_j)\right) = 0\,.
\]
It is easy to verify that the solution of the above is given by
\begin{align}
m_j = \frac{1}{\hSigma^{(\ell)}_{j,j}} \eta\Big(-\hSigma^{(\ell)}_{j,\sim j} m_{\sim j} + \ind(a = j);\mu_\ell\Big)\,,
\end{align}
with $\eta(\cdot;\,\cdot): \reals\times \reals_+ \to \reals$ denoting the soft-thresholding function defined as 
\begin{align}\label{eq:ST-definition}
\eta(z,\mu) = \begin{cases}
z- \mu \quad &\text{if } z>\mu\,,\\
0 &\text{if } -\mu\le z\le \mu\,,\\
z+\mu&\text{otherwise}\,.
\end{cases}
\end{align}
For a vector $u$, $\eta(u;\mu)$ is perceived entry-wise. 

This brings us to the following update rule to compute $\mli\in\reals^{dp}$ (solution of~\eqref{eq:opt-Lag}). Th notation $\Pi_L$, in line 5 below, denotes the Euclidean projection onto the $\ell_1$ ball of radius $L$ and can be computed in $O(dp)$ times using the procedure of~\cite{duchi2008efficient}.

\begin{algorithm}[H]
\begin{algorithmic}[1]
\STATE (initialization): $m(0) \leftarrow {m}^{(\ell-1)}_a$
\FOR{iteration $h = 1,\dotsc, H$}
\FOR{$j = 1,2, \dotsc, dp$}
\STATE  $m_j(h) \leftarrow \frac{1}{\hSigma^{(\ell)}_{j,j}} \eta\Big(-\hSigma^{(\ell)}_{j,\sim j} m_{\sim j}(h-1) + \ind(a = j);\mu_\ell\Big)$
\ENDFOR
\STATE $m(h)\leftarrow \Pi_L(m(h))$
\ENDFOR
\RETURN $\mli \leftarrow m(H)$
\end{algorithmic}
\end{algorithm}

In our experiments we implemented the same coordinate descent iterations explained above to solve for the decorrelating vectors $\mli$.

\subsubsection{Gradient descent algorithms} Letting $\cL(m) = (1/2)m^\sT \hSigma^{(\ell)}m - \<m,e_a\>$, we can write the objective of~\eqref{eq:opt-Lag} as $\cL(m) + \mu_\ell\|m\|_1$.
Projected gradient descent, applied to this constrained objective, results in a sequence of iterates $m(h)$, with $h=0,1,2,\dotsc$ the iteration number, as follows:
\begin{align}\label{eq:PGD}
m(h+1) = \arg\min_{\|m\|_1\le L} \Big\{\cL(m(h)) &+ \<\nabla\cL(m(h)), m - m(h)\>\nonumber\\
& + \frac{\eta}{2} \|m - m(h)\|_2^2 + \mu_\ell \|m\|_1 \Big\}\,.
\end{align}
In words, the next iterate $m(h+1)$ is obtained by constrained minimization of a first order approximation to $\cL(m)$, combined with a smoothing term that keeps the next iterate close to the current one.
Since the objective function is convex ($\hSigma^{(\ell)}\succeq 0$), iterates \eqref{eq:PGD} are guaranteed to converge to the global minimum of \eqref{eq:opt-Lag}.

Plugging for $\cL(m)$ and dropping the constant term $\cL(m(h))$, update \eqref{eq:PGD} reads as
\begin{align}
m(h+1) &= \arg\min_{\|m\|_1\le L} \Big\{\<\hSigma^{(\ell)} m(h) - e_a, m - m(h)\> + \frac{\eta}{2} \|m - m(h)\|_2^2 + \mu_\ell \|m\|_1 \Big\}\nonumber\\
&= \arg\min_{\|m\|_1\le L} \Big\{\frac{\eta}{2} \Big(m - m(h) + \frac{1}{\eta} (\hSigma^{(\ell)} m(h) - e_a)\Big)^2 + \mu_\ell \|m\|_1 \Big\}\,.\label{eq:GD-2}
\end{align}
To compute the update~\eqref{eq:GD-2}, we first solve the unconstrained problem which has a closed form solution given by $\eta\Big(m(h)-\frac{1}{\eta} (\hSigma^{(\ell)} m(h) - e_a);\frac{\mu_\ell}{\eta}\Big)$, with $\eta$ the soft thresholding function given by \eqref{eq:ST-definition}. The solution is then projected onto the ball of radius $L$.

In the following box, we summarize the projected gradient descent update rule for constructing the decorrelating vectors $\mli$.

\begin{algorithm}[H]
\begin{algorithmic}[1]
\STATE (initialization): $m(0) \leftarrow {m}^{(\ell-1)}_a$
\FOR{iteration $h = 1,\dotsc, H$}
\STATE  $m(h) \leftarrow \eta\Big(m(h)-\frac{1}{\eta} (\hSigma^{(\ell)} m(h) - e_a);\frac{\mu_\ell}{\eta}\Big)$
\STATE $m(h)\leftarrow \Pi_L(m(h))$
\ENDFOR
\RETURN $\mli \leftarrow m(H)$
\end{algorithmic}
\end{algorithm}

\subsection{Sparse inverse covariance}\label{sec:sparsePrecision}
In Section~\ref{sec:construct-online-TS} (Figure \ref{fig:fixed-coord}) we provided a numerical example wherein the offline debiasing does not admit an asymptotically normal distribution. 
As we see from the heat map in Figure~\ref{fig:heatmap-omega}, the precision matrix $\Omega$ has $\sim 20\%$ non-negligible entries per row. The goal of this section is to show that when $\Omega$ is sufficiently sparse, the offline debiased estimator has an asymptotically normal distribution and can be used for valid inference on model parameters. 
 
 The idea is to show that the decorrelating matrix $M$ is sufficiently close to the precision matrix $\Omega$. Since $\Omega$ is deterministic, this helps with controlling the statistical dependence between $M$ and $\eps$. Formally, starting from the decomposition~\eqref{eq:classicaldecomp} we write 
 \begin{align}
 \offth &= \theta_0 + (I-M\hSigma) (\hth^\sL - \theta_0) + \frac{1}{n}MX^\sT \eps\nonumber\\
 &= \theta_0 + (I-M\hSigma) (\hth^\sL - \theta_0) + \frac{1}{n} (M-\Omega)X^\sT \eps + \frac{1}{n}\Omega X^\sT \eps\,,
 \end{align}
 where we recall that $\hSigma$ is the empirical covariance of all the covariate vectors (episodes $E_0,\dotsc, E_{K-1}$).
 Therefore, we can write
 \begin{align}
 \begin{split}\label{eq:new-decomposition}
 \sqrt{n}(\offth - \theta_0) &= \Delta_1 + \Delta_2 + \frac{1}{\sqrt{n}} \Omega X^\sT \eps\,, \\
 \Delta_1 &= \sqrt{n}  (I-M\hSigma) (\hth^\sL - \theta_0)\,,\\
 \Delta_2 &= \frac{1}{\sqrt{n}} (M-\Omega)X^\sT \eps\,.
 \end{split}
 \end{align}
 The term $\Omega X^\sT \eps/\sqrt{n}$ is gaussian with $O(1)$ variance at each coordinate. For bias term $\Delta_1$, we show that $\Delta_1 = O(s_0(\log p)/\sqrt{n})$ by controlling $|I-M\hSigma|$. To bound the bias term $\Delta_2$ we write
 \begin{align}
 \|\Delta_2\|_\infty\le \frac{1}{\sqrt{n}}  \|M- \Omega\|_1 \|X^\sT \eps\|_\infty\,,
 \end{align}  
 where $\|M- \Omega\|_1$ denotes the $\ell_1-\ell_1$ norm of $M-\Omega$ (the maximum $\ell_1$ norm of its columns).
 By using \cite[Proposition 3.2]{basu2015regularized}, we have $\|X^\sT \eps\|_\infty/\sqrt{n} = O_P(\sqrt{\log(dp)})$. Therefore, to bound $\Delta_2$ we need to control $\|M-\Omega\|_1$. We provide such bound in our next lemma, under the sparsity assumption on the rows of $\Omega$. 
 
 Define
 $$s_\Omega \equiv \max_{i\in [dp]}\; \Big|j\in[dp]:\;\; \Omega_{i,j} \neq 0\Big|\,,$$
 the maximum sparsity of rows of $\Omega$. In addition, let the (offline) decorrelating vectors $m_a$ be defined as follows, for $a\in [dp]$:
 \begin{align}\label{eq:M-offline}
m_a\in\arg\min_{m\in \reals^{dp}} \quad \frac{1}{2} m^\sT \hSigma m -\<m,e_a\> + \mu \|m\|_1\,.
 \end{align}
\begin{lemma}\label{lem:omeg-estimation}
Consider the decorrelating vectors $m_a$, $a\in [dp]$, given by optimization~\eqref{eq:M-offline} with $\mu = 2\mya\sqrt{\frac{\log (dp)}{n}}$. Then, for some proper constant $c>0$ and the sample size condition $n\ge 32 \alpha (\omega^2 \vee 1) s_\Omega \log(dp)$, the following happens with probability at least 
$1-\exp(-c\log (dp^2))-\exp(-cn (1\wedge\omega^{-2}))$:
\begin{align*}
\max_{i\in [dp]} \|m_a - \Omega e_a\|_1 \le \frac{192\mya}{\alpha} s_\Omega \sqrt{\frac{\log (dp)}{n}},
\end{align*}
where $\alpha$and $\omega$ are defined in Proposition \ref{pro:RE}.
\end{lemma} 
The proof of Lemma~\ref{lem:omeg-estimation} is deferred to Section~\ref{proof-lem:omeg-estimation}.

 By employing this lemma, if $\Omega$ is sufficiently sparse, that is $s_\Omega = o(\sqrt{n}/\log(dp))$, then the bias term $\|\Delta_2\|_\infty$ also vanishes asymptotically and the (offline) debiased estimator $\offth$ admits an unbiased normal distribution. We formalize such distributional characterization in the next theorem.
 \begin{thm}\label{thm:offline-TS}
   Consider the $\VAR(d)$ model \eqref{eq:varddef} for time series and let $\offth$ be the (offline) debiased estimator \eqref{eq:classicaldebias}, with the decorrelating matrix $M = (m_1, \dotsc, m_{dp})^\sT \in \reals^{dp\times dp}$ constructed as in~\eqref{eq:M-offline}, with $\mu = 2\mya\sqrt{\log (dp)/n}$. Also, let $\lambda = \lambda_0 \sqrt{\log(dp)/n}$ be the regularization parameter in the Lasso estimator $\Lsth$, with $\mya,\lambda_0$ large enough constants. 
  
 Suppose that $s_0 = o(\sqrt{n}/\log(dp))$ and $s_\Omega = o(\sqrt{n}/\log(dp))$, then the following holds true for any fixed sequence of integers $a(n)\in [dp]$: For all $x\in \reals$, we have
 \begin{align}\label{eq:DC-TS}
 \lim_{n\to\infty} \sup_{\|\theta_0\|_0\le s_0}\bigg| \prob\left\{\frac{\sqrt{n} (\offth_a - \theta_{0,a})}{\sqrt{V_{n,a}}} \le x\right\}  - \Phi(x)\bigg| = 0\,,
 \end{align}
 where $V_{n,a}\equiv \sigma^2 (M\hSigma M^\sT)_{a,a}$.
 \end{thm}
 We refer to Section~\ref{proof-thm:offline-TS} for the proof of Theorem~\ref{thm:offline-TS}.

 \medskip
 {\bf Numerical example.} Consider a $\VAR(d)$ model with parameters $p=25, d=3, T=70,$ and Gaussian noise terms with covariance matrix $\Sigma_{\zeta}$ satisfying $\Sigma_{\zeta}(i,j)=\rho^{|i-j|}$ for $\rho=0.1$. Let $A_i$ matrices have entries generated independently from 
 $b\cdot\text{Bern}(q)\cdot\text{Unif}(\{+1,-1\})$ formula with parameters $b=0.15$, $q=0.05$. Figure \ref{fig:off-all-coord:sub1} shows the magnitudes of the entries of the precision matrix $\Omega=\E(x_ix_i^T)^{-1}$; as we see $\Omega$ is sparse. Figures \ref{fig:off:all-coord:sub2}, \ref{fig:off:all-coord:sub3}, and \ref{fig:off:all-coord:sub4} demonstrate normality of the rescaled residuals of the offline debiased estimator built by decorrelating matrix $M$ with rows coming from optimization described in \eqref{eq:M-offline}. 
 
 After this paper was posted, we learned 
of  simultaneous work (an updated version of \cite{basu2017system})
 that also studies the performance of the (offline) 
 debiased estimator for time series with \emph{sparse} precision matrix. 
 We would like to highlight some of the differences between 
our discussion in Section~\ref{sec:sparsePrecision} and that paper: 1)~\cite{basu2017system} considers 
 decorrelating matrix $M$ constructed by an optimization 
 of form~\eqref{eq:opt}, using the entire sample 
 covariance $\hSigma^{(K)}$, while we work with the 
 Lagrangian equivalent~\eqref{eq:M-offline}. 2)~\cite{basu2017system} considers $\VAR(1)$ model, 
 while we work with $\VAR(d)$ models. 
 3) \cite{basu2017system} assumes a stronger notion 
 of sparsity, viz. the sparsity of the entire precision 
 matrix as well as the transition matrix to scale as 
 $o(\sqrt{n}/\log p)$. Our results only require 
 the \emph{row-wise sparsity} of the precision matrix to 
 scale as $o(\sqrt{n}/\log p)$, cf. 
 Theorem~\ref{thm:offline-TS}. 
 
 \begin{figure}
 	\centering
	\begin{subfigure}{.45\textwidth}
	\centering
	\includegraphics[scale=0.27]{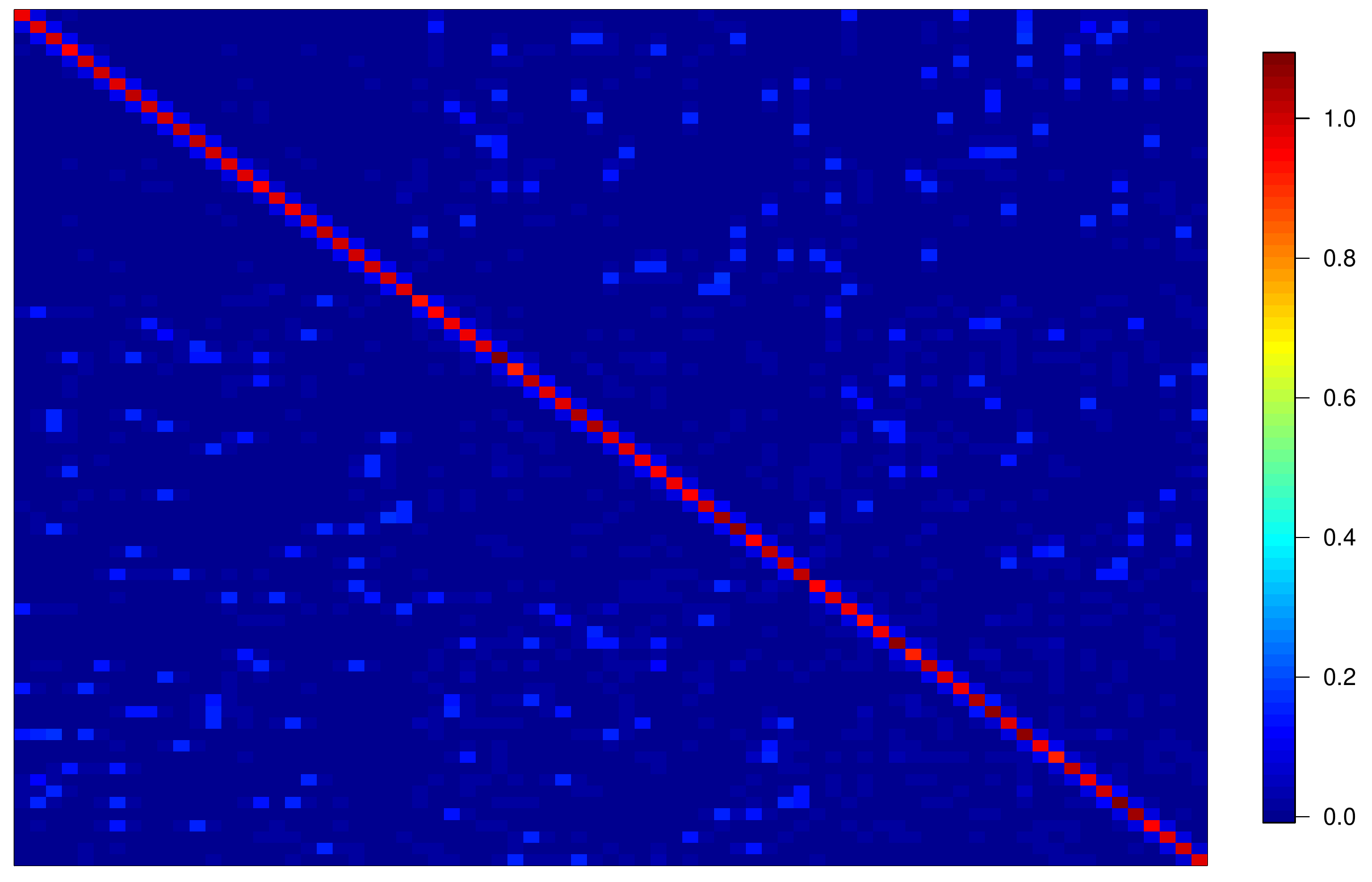}
	\caption{Heat map of magnitudes of entries of $\Omega=\E(x_ix_i^T)^{-1}$}
	\label{fig:off-all-coord:sub1}
\end{subfigure}%
\hspace{1cm}
 \begin{subfigure}{0.45\textwidth}
 	\centering
 	\includegraphics[scale =0.21]{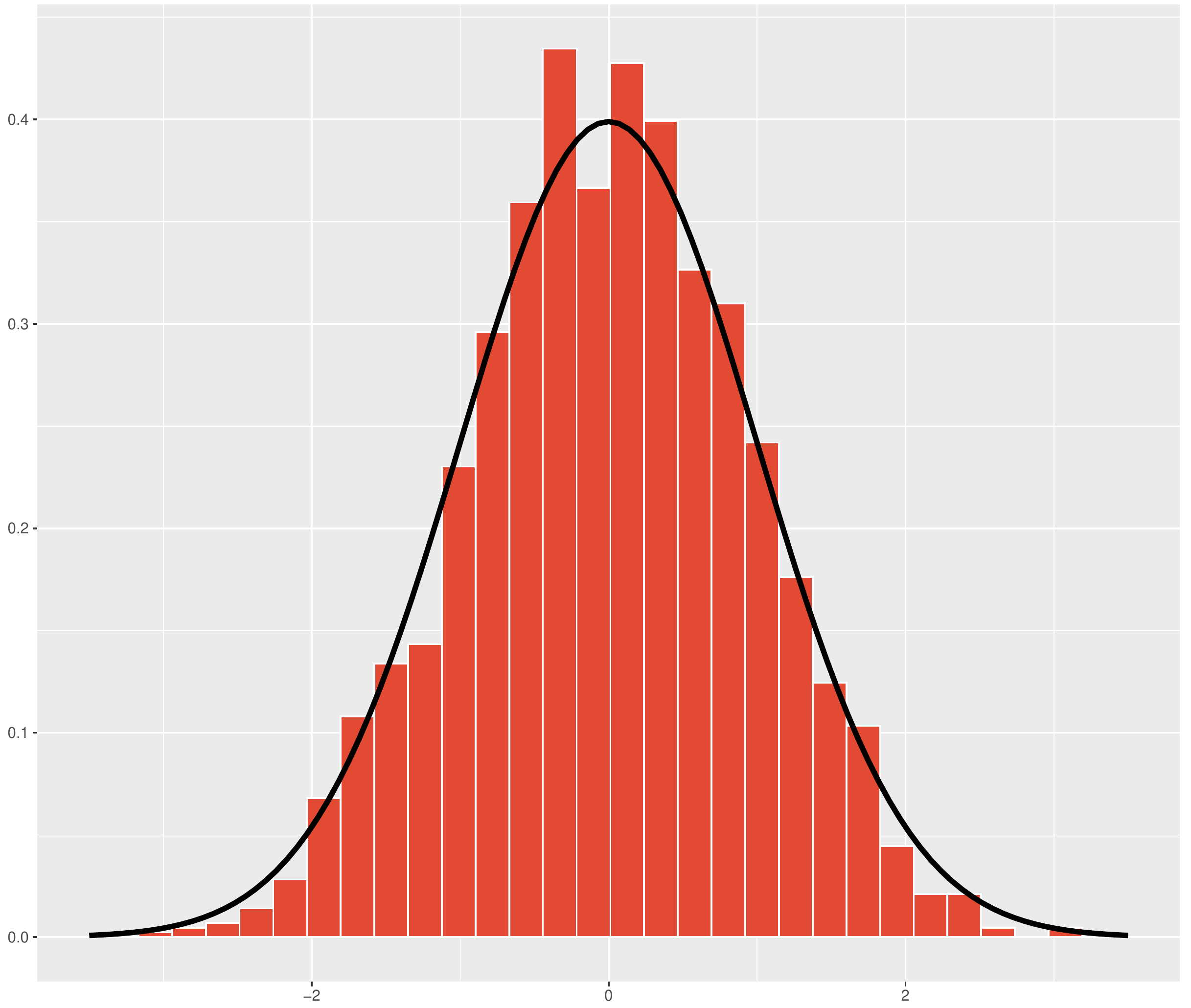}
 	\put(-160,60){\rotatebox{90}{\scriptsize{Density}}}
 	\put(-110,-7){\rotatebox{0}{\scriptsize{Rescaled Residuals}}}
 	\caption{Histogram of Rescaled Residuals}
 	\label{fig:off:all-coord:sub2}
 \end{subfigure}
 \vspace{0.8cm}
 
 \begin{subfigure}{.45\textwidth}
 		\centering
 		\includegraphics[scale =0.21]{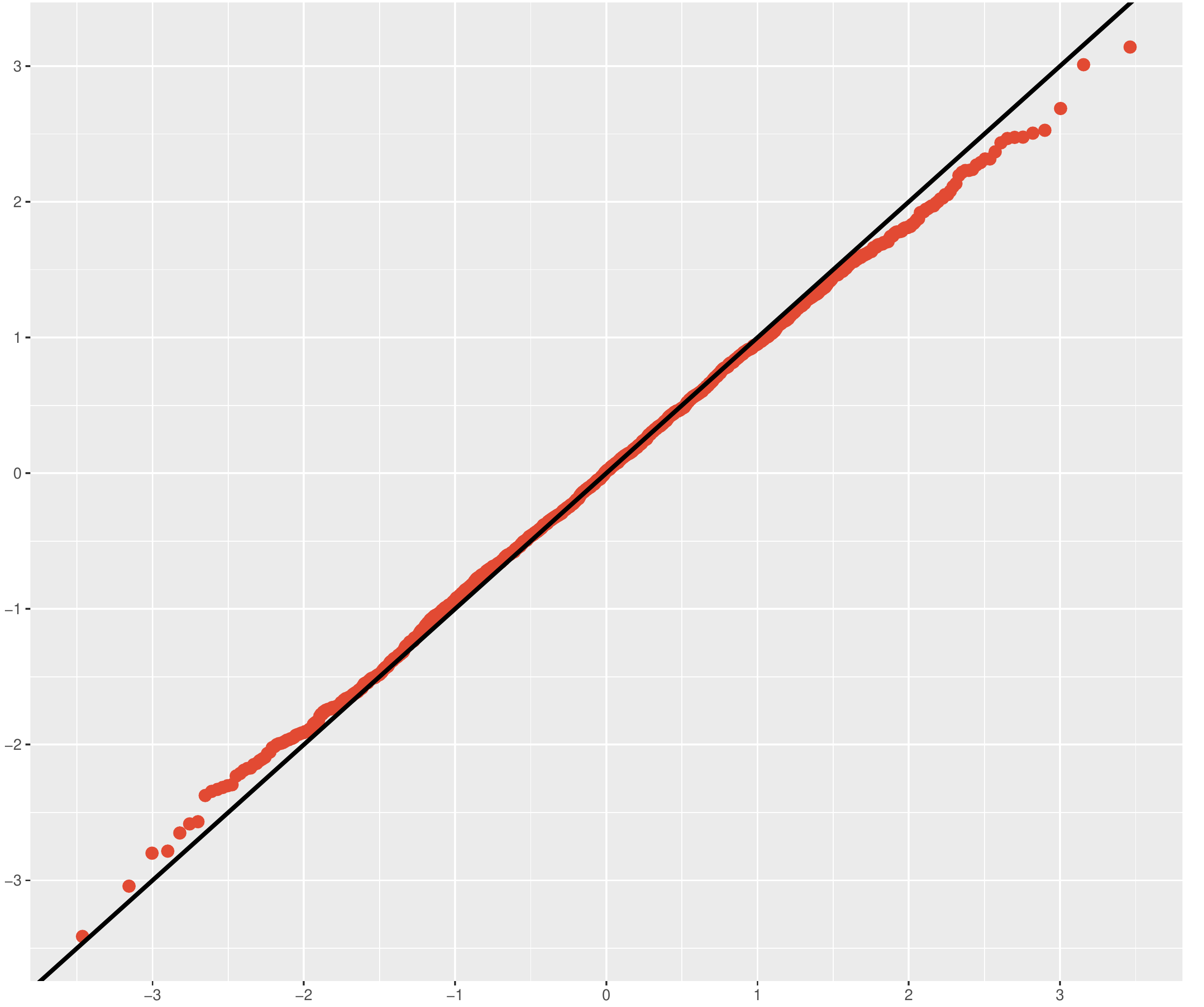}
 		\put(-160,60){\rotatebox{90}{\scriptsize{Sample}}}
 		\put(-95,-7){\rotatebox{0}{\scriptsize{Theoretical}}}
 		\caption{QQ plot of Rescaled Residuals  }
 		\label{fig:off:all-coord:sub3}
 	\end{subfigure}%
	\hspace{1cm}
 	\begin{subfigure}{.45\textwidth}
 		\centering
 		\includegraphics[scale =0.21]{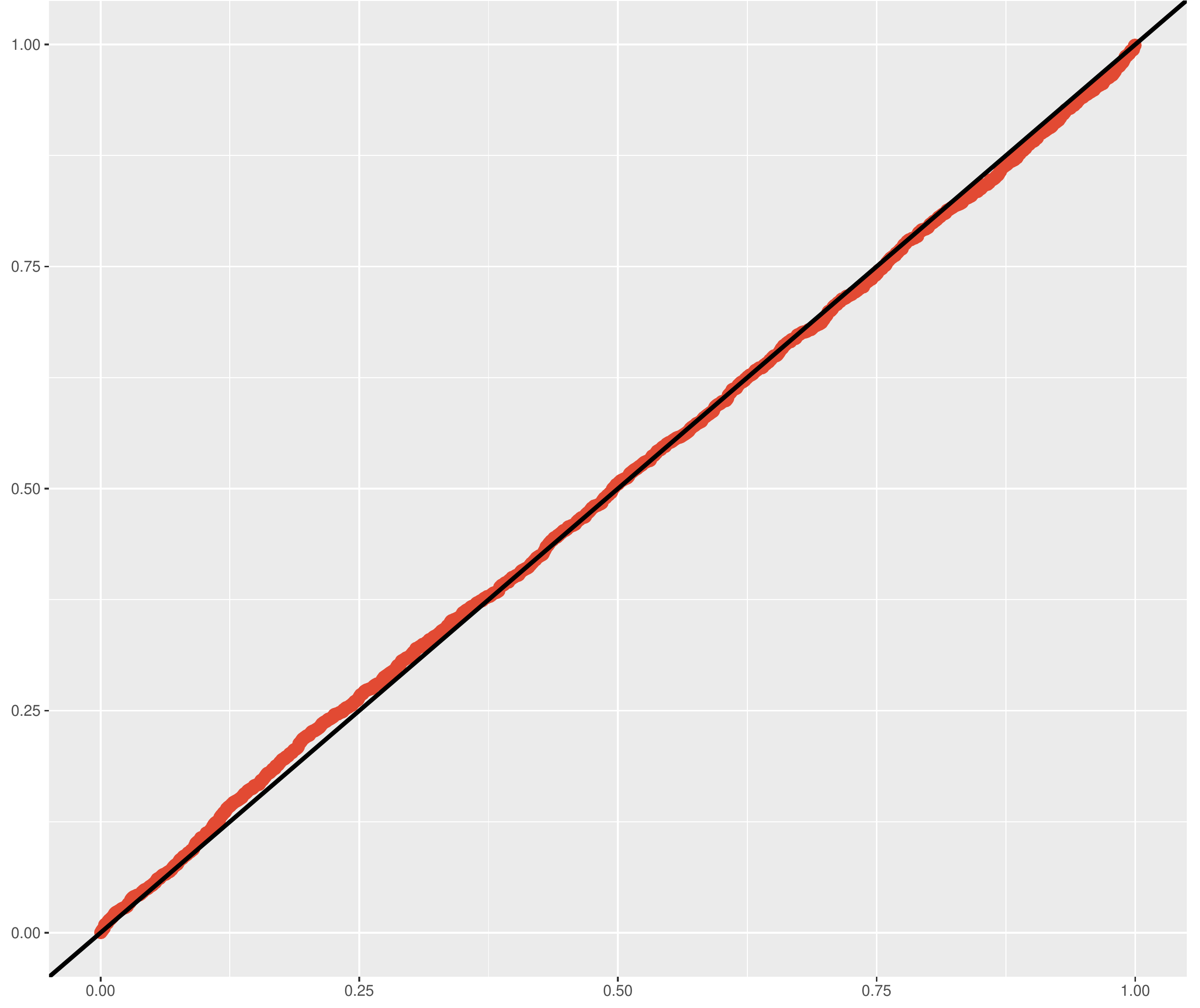}
 		\put(-160,60){\rotatebox{90}{\scriptsize{Sample}}}
 		\put(-94,-7){\rotatebox{0}{\scriptsize{Theoretical}}}
 		\caption{PP plot of Rescaled Residuals}
 		\label{fig:off:all-coord:sub4}
 	\end{subfigure}

 	\caption{{\small A Simple example of a $\VAR(d)$ process with parameters $p=25, d=3, T=70$, and noise term covariance matrix $\Sigma_{\zeta}$ s.t $\Sigma_{\zeta}(i,j)=\rho^{|i-j|}$ with $\rho=0.1$. $A_i$ matrices have independent elements coming from $b\cdot\text{Bern}(q).\text{Unif}(\{+1,-1\})$ formula  with $b=0.15, q=0.05$. Normality of rescaled residuals (figures \ref{fig:off:all-coord:sub2}, \ref{fig:off:all-coord:sub3}, and \ref{fig:off:all-coord:sub4}) validates the successful  performance of offline debiasing estimator under sparsity of precision matrix $\Omega$ ( figure \ref{fig:off-all-coord:sub1}) as we discussed in theorem \ref{thm:offline-TS}.}}
 	\label{fig:offline-sparse-omega-plots}
 \end{figure}
 
 \subsection{Concluding remarks}
 In this work we devised the `online debiasing' approach for the high-dimensional regression and showed that it asymptotically admits an unbiased Gaussian distribution, even when the samples are collected adaptively. Also through numerical examples we demonstrated that the (offline) debiased estimator suffers from the bias induced by the correlation in the samples and cannot be used for valid statistical inference in these settings (unless the precision matrix is sufficiently sparse). 
 
 Since its proposal, the (offline) debiasing approach has been used as a tool to address a variety of problems such as  estimating average treatment effect and casual inference in high-dimension~\cite{athey2016approximate}, precision matrix estimation~\cite{jankova2017honest}, distributed multitask learning, and studying neuronal functional network dynamics~\cite{sheikhattar2018extracting}, hierarchical testing~\cite{guo2019group}, to name a few. It has also been used for different statistical aims such as controlling FDR in high-dimensions~\cite{javanmard2019false}, estimation of the prediction risk~\cite{javanmard2018debiasing}, inference on predictions~\cite{cai2017confidence,javanmard2017flexible} and explained variance~\cite{cai2018semi,javanmard2017flexible}, and testing more general hypotheses regarding the model
parameters, like testing membership in a convex cone,
testing the parameter strength, and testing arbitrary functions of the parameters~\cite{javanmard2017flexible}. We anticipate that the online debiasing approach and analysis can be used to tackle similar problems under adaptive data collection. We leave this for future work.

\subsection*{Acknowledgements}
A. Javanmard was partially supported by an Outlier Research in Business (iORB) grant from the USC Marshall School of Business, a Google Faculty Research Award and the NSF CAREER Award DMS-1844481.

\bibliographystyle{amsalpha}
\bibliography{../../all-bibliography.bib}

\newcommand{\etalchar}[1]{$^{#1}$}
\providecommand{\bysame}{\leavevmode\hbox to3em{\hrulefill}\thinspace}
\providecommand{\MR}{\relax\ifhmode\unskip\space\fi MR }
\providecommand{\MRhref}[2]{%
  \href{http://www.ams.org/mathscinet-getitem?mr=#1}{#2}
}
\providecommand{\href}[2]{#2}
\begin{thebibliography}{VdGBR{\etalchar{+}}14}

\bibitem[AIW16]{athey2016approximate}
Susan Athey, Guido~W Imbens, and Stefan Wager, \emph{Approximate residual
  balancing: De-biased inference of average treatment effects in high
  dimensions}, arXiv preprint arXiv:1604.07125 (2016).

\bibitem[BB15]{bastani2015online}
Hamsa Bastani and Mohsen Bayati, \emph{Online decision-making with
  high-dimensional covariates}, Available at SSRN 2661896 (2015).

\bibitem[BCB{\etalchar{+}}12]{bubeck2012regret}
S{\'e}bastien Bubeck, Nicolo Cesa-Bianchi, et~al., \emph{Regret analysis of
  stochastic and nonstochastic multi-armed bandit problems}, Foundations and
  Trends{\textregistered} in Machine Learning \textbf{5} (2012), no.~1, 1--122.

\bibitem[BCW11]{belloni2011square}
Alexandre Belloni, Victor Chernozhukov, and Lie Wang, \emph{Square-root lasso:
  pivotal recovery of sparse signals via conic programming}, Biometrika
  \textbf{98} (2011), no.~4, 791--806.

\bibitem[BDMP17]{basu2017system}
Sumanta Basu, Sreyoshi Das, George Michailidis, and Amiyatosh~K Purnanandam,
  \emph{A system-wide approach to measure connectivity in the financial
  sector}, Available at SSRN 2816137 (2017).

\bibitem[BM12]{BayatiMontanariLASSO}
M.~Bayati and A.~Montanari, \emph{{The LASSO risk for gaussian matrices}}, IEEE
  Trans. on Inform. Theory \textbf{58} (2012), 1997--2017.

\bibitem[BM15]{basu2015regularized}
Sumanta Basu and George Michailidis, \emph{Regularized estimation in sparse
  high-dimensional time series models}, The Annals of Statistics \textbf{43}
  (2015), no.~4, 1535--1567.

\bibitem[BVDG11]{buhlmann2011statistics}
Peter B{\"u}hlmann and Sara Van De~Geer, \emph{Statistics for high-dimensional
  data: methods, theory and applications}, Springer Science \& Business Media,
  2011.

\bibitem[BY01]{benjamini2001control}
Yoav Benjamini and Daniel Yekutieli, \emph{The control of the false discovery
  rate in multiple testing under dependency}, Annals of statistics (2001),
  1165--1188.

\bibitem[CG17]{cai2017confidence}
T~Tony Cai and Zijian Guo, \emph{Confidence intervals for high-dimensional
  linear regression: Minimax rates and adaptivity}, The Annals of statistics
  \textbf{45} (2017), no.~2, 615--646.

\bibitem[CG18]{cai2018semi}
\bysame, \emph{Semi-supervised inference for explained variance in
  high-dimensional linear regression and its applications}, arXiv preprint
  arXiv:1806.06179 (2018).

\bibitem[DM12]{deshpande2012linear}
Yash Deshpande and Andrea Montanari, \emph{Linear bandits in high dimension and
  recommendation systems}, Communication, Control, and Computing (Allerton),
  2012 50th Annual Allerton Conference on, IEEE, 2012, pp.~1750--1754.

\bibitem[DMST18]{deshpande2018accurate}
Yash Deshpande, Lester Mackey, Vasilis Syrgkanis, and Matt Taddy,
  \emph{Accurate inference for adaptive linear models}, International
  Conference on Machine Learning, 2018, pp.~1202--1211.

\bibitem[DSSSC08]{duchi2008efficient}
John Duchi, Shai Shalev-Shwartz, Yoram Singer, and Tushar Chandra,
  \emph{Efficient projections onto the l 1-ball for learning in high
  dimensions}, Proceedings of the 25th international conference on Machine
  learning, ACM, 2008, pp.~272--279.

\bibitem[FSGM{\etalchar{+}}07]{fujita2007modeling}
Andr{\'e} Fujita, Joao~R Sato, Humberto~M Garay-Malpartida, Rui Yamaguchi,
  Satoru Miyano, Mari~C Sogayar, and Carlos~E Ferreira, \emph{Modeling gene
  expression regulatory networks with the sparse vector autoregressive model},
  BMC systems biology \textbf{1} (2007), no.~1, 39.

\bibitem[GRBC19]{guo2019group}
Zijian Guo, Claude Renaux, Peter B{\"u}hlmann, and T~Tony Cai, \emph{Group
  inference in high dimensions with applications to hierarchical testing},
  arXiv preprint arXiv:1909.01503 (2019).

\bibitem[GWC16]{gelper2016identifying}
Sarah Gelper, Ines Wilms, and Christophe Croux, \emph{Identifying demand
  effects in a large network of product categories}, Journal of Retailing
  \textbf{92} (2016), no.~1, 25--39.

\bibitem[HENR88]{holtz1988estimating}
Douglas Holtz-Eakin, Whitney Newey, and Harvey~S Rosen, \emph{Estimating vector
  autoregressions with panel data}, Econometrica: Journal of the Econometric
  Society (1988), 1371--1395.

\bibitem[HH14]{hall2014martingale}
Peter Hall and Christopher~C Heyde, \emph{Martingale limit theory and its
  application}, Academic press, 2014.

\bibitem[HTW15]{hastie2015statistical}
Trevor Hastie, Robert Tibshirani, and Martin Wainwright, \emph{Statistical
  learning with sparsity: the lasso and generalizations}, Chapman and Hall/CRC,
  2015.

\bibitem[Jav14]{javanmard2014inference}
Adel Javanmard, \emph{Inference and estimation in high-dimensional data
  analysis}, Ph.D. thesis, PhD Thesis, Stanford University, 2014.

\bibitem[JJ{\etalchar{+}}19]{javanmard2019false}
Adel Javanmard, Hamid Javadi, et~al., \emph{False discovery rate control via
  debiased lasso}, Electronic Journal of Statistics \textbf{13} (2019), no.~1,
  1212--1253.

\bibitem[JL17]{javanmard2017flexible}
Adel Javanmard and Jason~D Lee, \emph{A flexible framework for hypothesis
  testing in high-dimensions}, arXiv preprint arXiv:1704.07971 (2017).

\bibitem[JM14a]{javanmard2014confidence}
Adel Javanmard and Andrea Montanari, \emph{Confidence intervals and hypothesis
  testing for high-dimensional regression.}, Journal of Machine Learning
  Research \textbf{15} (2014), no.~1, 2869--2909.

\bibitem[JM14b]{javanmard2014hypothesis}
\bysame, \emph{Hypothesis testing in high-dimensional regression under the
  gaussian random design model: Asymptotic theory}, IEEE Transactions on
  Information Theory \textbf{60} (2014), no.~10, 6522--6554.

\bibitem[JM18]{javanmard2018debiasing}
\bysame, \emph{Debiasing the lasso: Optimal sample size for gaussian designs},
  The Annals of Statistics \textbf{46} (2018), no.~6A, 2593--2622.

\bibitem[JvdG17]{jankova2017honest}
Jana Jankov{\'a} and Sara van~de Geer, \emph{Honest confidence regions and
  optimality in high-dimensional precision matrix estimation}, Test \textbf{26}
  (2017), no.~1, 143--162.

\bibitem[KHW{\etalchar{+}}11]{kim2011battle}
Edward~S Kim, Roy~S Herbst, Ignacio~I Wistuba, J~Jack Lee, George~R
  Blumenschein, Anne Tsao, David~J Stewart, Marshall~E Hicks, Jeremy Erasmus,
  Sanjay Gupta, et~al., \emph{The battle trial: personalizing therapy for lung
  cancer}, Cancer discovery \textbf{1} (2011), no.~1, 44--53.

\bibitem[LR85]{lai1985asymptotically}
Tze~Leung Lai and Herbert Robbins, \emph{Asymptotically efficient adaptive
  allocation rules}, Advances in applied mathematics \textbf{6} (1985), no.~1,
  4--22.

\bibitem[LW82]{lai1982least}
Tze~Leung Lai and Ching~Zong Wei, \emph{Least squares estimates in stochastic
  regression models with applications to identification and control of dynamic
  systems}, The Annals of Statistics (1982), 154--166.

\bibitem[NXTZ17]{nie2017why}
Xinkun Nie, Tian Xiaoying, Jonathan Taylor, and James Zou, \emph{Why adaptively
  collected data have negative bias and how to correct for it}.

\bibitem[PRC{\etalchar{+}}16]{perchet2016batched}
Vianney Perchet, Philippe Rigollet, Sylvain Chassang, Erik Snowberg, et~al.,
  \emph{Batched bandit problems}, The Annals of Statistics \textbf{44} (2016),
  no.~2, 660--681.

\bibitem[RT10]{rusmevichientong2010linearly}
Paat Rusmevichientong and John~N Tsitsiklis, \emph{Linearly parameterized
  bandits}, Mathematics of Operations Research \textbf{35} (2010), no.~2,
  395--411.

\bibitem[SBB15]{seth2015granger}
Anil~K Seth, Adam~B Barrett, and Lionel Barnett, \emph{Granger causality
  analysis in neuroscience and neuroimaging}, Journal of Neuroscience
  \textbf{35} (2015), no.~8, 3293--3297.

\bibitem[SML{\etalchar{+}}18]{sheikhattar2018extracting}
Alireza Sheikhattar, Sina Miran, Ji~Liu, Jonathan~B Fritz, Shihab~A Shamma,
  Patrick~O Kanold, and Behtash Babadi, \emph{Extracting neuronal functional
  network dynamics via adaptive granger causality analysis}, Proceedings of the
  National Academy of Sciences \textbf{115} (2018), no.~17, E3869--E3878.

\bibitem[SPHD04]{srinivasan2004promotions}
Shuba Srinivasan, Koen Pauwels, Dominique~M Hanssens, and Marnik~G Dekimpe,
  \emph{Do promotions benefit manufacturers, retailers, or both?}, Management
  Science \textbf{50} (2004), no.~5, 617--629.

\bibitem[SRR19]{shin2019bias}
Jaehyeok Shin, Aaditya Ramdas, and Alessandro Rinaldo, \emph{On the bias, risk
  and consistency of sample means in multi-armed bandits}, arXiv preprint
  arXiv:1902.00746 (2019).

\bibitem[SS06]{shumway2006time}
Robert~H Shumway and David~S Stoffer, \emph{Time series analysis and its
  applications: with r examples}, Springer Science \& Business Media, 2006.

\bibitem[SW01]{stock2001vector}
James~H Stock and Mark~W Watson, \emph{Vector autoregressions}, Journal of
  Economic perspectives \textbf{15} (2001), no.~4, 101--115.

\bibitem[SZ12]{sun2012scaled}
Tingni Sun and Cun-Hui Zhang, \emph{Scaled sparse linear regression},
  Biometrika \textbf{99} (2012), no.~4, 879--898.

\bibitem[Tib96]{Tibs96}
R.~Tibshirani, \emph{{Regression shrinkage and selection with the Lasso}}, J.
  Royal. Statist. Soc B \textbf{58} (1996), 267--288.

\bibitem[VBW15]{villar2015multi}
Sofia Villar, Jack Bowden, and James Wason, \emph{Multi-armed bandit models for
  the optimal design of clinical trials: benefits and challenges}, Statistical
  science: a review journal of the Institute of Mathematical Statistics
  \textbf{30} (2015), no.~2, 199.

\bibitem[VdGBR{\etalchar{+}}14]{van2014asymptotically}
Sara Van~de Geer, Peter B{\"u}hlmann, Ya'acov Ritov, Ruben Dezeure, et~al.,
  \emph{On asymptotically optimal confidence regions and tests for
  high-dimensional models}, The Annals of Statistics \textbf{42} (2014), no.~3,
  1166--1202.

\bibitem[Ver12]{vershynin2012introduction}
R.~Vershynin, \emph{Introduction to the non-asymptotic analysis of random
  matrices}, Compressed Sensing: Theory and Applications (Y.C. Eldar and
  G.~Kutyniok, eds.), Cambridge University Press, 2012, pp.~210--268.

\bibitem[WBBM17]{wilms2017interpretable}
Ines Wilms, Sumanta Basu, Jacob Bien, and David~S Matteson, \emph{Interpretable
  vector autoregressions with exogenous time series}, arXiv preprint
  arXiv:1711.03623 (2017).

\bibitem[XQL13]{xu2013estimation}
Min Xu, Tao Qin, and Tie-Yan Liu, \emph{Estimation bias in multi-armed bandit
  algorithms for search advertising}, Advances in Neural Information Processing
  Systems, 2013, pp.~2400--2408.

\bibitem[ZLK{\etalchar{+}}08]{zhou2008bayesian}
Xian Zhou, Suyu Liu, Edward~S Kim, Roy~S Herbst, and J~Jack Lee, \emph{Bayesian
  adaptive design for targeted therapy development in lung cancer---a step
  toward personalized medicine}, Clinical Trials \textbf{5} (2008), no.~3,
  181--193.

\bibitem[ZZ11]{ZhangZhangSignificance}
C.-H. Zhang and S.S. Zhang, \emph{{Confidence Intervals for Low-Dimensional
  Parameters in High-Dimensional Linear Models}}, {\sf arXiv:1110.2563}, 2011.

\end{thebibliography}
\clearpage
\appendix

\section{Proofs of Section~3}\label{app:timeseries}

\subsection{Technical preliminaries}

Recall the definition of the regression design from
Eqs.\eqref{eq:Y-X} in the time series case:
\begin{align*}
\theta_0 &= (A^{(1)}_i, A^{(2)}_i, \dots, A^{(d)}_i)^\sT, \nonumber \\
X &= \begin{bmatrix}
z_d^\sT  &z_{d-1}^\sT &\dots &z_1^\sT  \\
z_{d+1}^\sT & z_{d}^\sT &\dots &z_2^\sT\\
\vdots & \vdots &\ddots &\vdots  \\
z_{T-1}^\sT & z_{T-2}^\sT &\dots &z_{T-d}^\sT
\end{bmatrix}, \nonumber \\
y &= (z_{d+1, i}, z_{d+2, i}, \dots, z_{T, i}), \nonumber \\
\eps &= (\zeta_{d+1, i}, \zeta_{d+2, i}, \dots, \zeta_{T, i}). 
\end{align*}

We first establish some preliminary results for
stable time series. 
For the stationary process $x_t = (z_{t+d-1}^\sT, \dotsc, z_t^\sT)^\sT$ (rows of $X$), let $\Gamma_{{x}}(s) = {{\rm Cov}}({x_t},{x_{t+s}})$, for $t,s\in \integers$ and define the spectral density 
 $f_{{x}}(r) \equiv {1}/({2\pi}) \sum_{\ell=-\infty}^\infty \Gamma_{{X}}(\ell) e^{-j\ell r}$, for $r\in[-\pi,\pi]$ . The measure of stability of the process is defined as the maximum eigenvalue of the density

 \begin{align}
 M(f_{{x}}) \equiv \underset{r\in[-\pi,\pi]}{\sup} \sigma_{\max}(f_{{x}}(r))\,.
\end{align}
 Likewise, the minimum eigenvalue of the spectrum is defined as $m(f_{{x}}) \equiv \underset{r\in[-\pi,\pi]}{\inf} \sigma_{\min}(f_{{x}}(r))$,
 which captures the dependence among the covariates.  (Note that for the case of i.i.d. samples, ${M}(f_{{x}})$ and $m(f_{{x}})$ reduce to the maximum and minimum eigenvalue of the population covariance.)

 The $p$-dimensional $\VAR(d)$ model~\eqref{eq:varddef} can be represented as a $dp$-dimensional $\VAR(1)$ model. Recall our notation $x_t = (z_{t+d-1}^\sT, \dotsc, z_t^\sT)^\sT$ (rows of $X$ in \eqref{eq:Y-X}). Then \eqref{eq:varddef} can be written as
\begin{align}\label{tA}
x_t = \tA x_{t-1} + \tilde{\zeta}_t\,,
\end{align}
with 
\begin{align}
\tA = 
\left(\begin{array}{@{}cccc|c@{}}
A_1 & A_2 &\dotsc &A_{d-1}&A_d\\\hline
&&I_{(d-1)p}&&0
\end{array}\right)\,,\quad \quad
\tilde{\zeta}_t = \begin{pmatrix}
\zeta_{t+d-1}\\
0
\end{pmatrix}\,.
\end{align}
The reverse characteristic polynomial for the $\VAR(1)$ model reads as $\tilde{\cA} = I - \tilde{A} z$. 

The following lemma controls $M(f_x), m(f_x)$ in terms of the spectral 
properties of the noise $\Sigmazeta$ and the characteristic polynomials
$\cA, \tilde{\cA}$.

\begin{lemma}[\cite{basu2015regularized}] \label{lem:basuspectral}
We have:
\begin{align}
\frac{1}{2\pi} \lambdamax(\Sigma) &\le M(f_x) \le \frac{\lambdamax(\Sigmazeta)}{\mumin(\tilde{\cA})}, \nonumber\\
\lambdamin(\Sigma) &\ge \frac{\lambdamin(\Sigmazeta)}{\mumax(\cA)}.\label{their-min}
\end{align}
 \end{lemma} 

 We also use the following bound on $M(f_x)$ in terms of 
 characteristic polynomial $\cA$ of the time series $z_t$. 

 \begin{lemma}\label{lem:basuspectral2}
 The following holds:
 \begin{align*}
\frac{1}{2\pi} \lambdamax(\Sigma) &\le  M(f_x) \le d M(f_z) \le \frac{d\lambdamax(\Sigmazeta)}{\mumin(\cA)}.
 \end{align*}
 \end{lemma}
 \begin{proof}
 
Let $\Gamma_x(\ell)=\mathbb{E}[ x_t x_{t+\ell}^\sT]$ to refer the autocovariance of the $dp$-dimensional process $x_t$.  Therefore $\Sigma = \Gamma_x(0)$. Likewise, the autocovariance $\Gamma_z(\ell)$ is defined for the $p$-dimensional process $z_t$.
We represent $\Gamma_x(\ell)$
in terms of 
$d^2$ blocks, each of which is a $p\times p$ matrix. The block 
in position
$(r,s)$ is $\Gamma_{z}(\ell+r-s)$.
Now, for a vector $v\in \reals^{dp}$ with unit $\ell_2$ norm, decompose it 
as $d$ blocks of $p$ dimensional vectors $v= (v_1^\sT, v_2^\sT, \dots, v_d^\sT)^\sT$, by which we have
\begin{equation}\label{eq:GammaZrayleighquotientdecomp}
    v^\sT \Gamma_z(\ell)v=\sum\limits_{1\leq r,s \leq d}^{} v_r^\sT \Gamma_{x}(\ell+r-s)v_s\,.
\end{equation}
Since the spectral density $f_z(\theta)$ is the Fourier transform
of the autocorrelation function, we have by Equation \eqref{eq:GammaZrayleighquotientdecomp}, 
\begin{align*}
    \<v,  f_z(\theta)v\> &=\frac{1}{2\pi}\sum\limits_{\ell=-\infty}^{\infty}\<v, \Gamma_z(\ell)e^{-j\ell\theta} v \>\\
    &=\frac{1}{2\pi} \sum\limits_{\ell=-\infty}^{\infty}\sum\limits_{1\leq r,s\leq d}^{} \<v_r, \Gamma_{z}(\ell+r-s)e^{-j\ell\theta} v_s\> \\
    &= \sum\limits_{1\leq r,s\leq d}^{} \<v_r, \Big(\frac{1}{2\pi}\sum\limits_{\ell=-\infty}^{\infty}\Gamma_{x}(\ell+r-s)e^{-j(\ell+r-s)\theta}\Big) v_se^{j(r-s)\theta} \>\\
    &=\sum\limits_{1\leq r,s\leq d}^{} \<v_r, f_{x}(\theta)e^{j(r-s)\theta}v_s\>\\
    &=V(\theta)^*f_x(\theta)V(\theta),
    \end{align*}
with $V(\theta)=\sum\limits_{r=1}^{d}e^{-jr\theta}v_r$. Now, 
we have:
\begin{align*}
\norm{V(\theta)}_2& \leq \sum\limits_{r=1}^{d}\norm{v_r}_2 
\le \Big({d \sum_{r=1}^d \norm{v_r}_2^2 }\Big)^{1/2} \le  \sqrt{d}. 
\end{align*} 
Combining this with the Rayleigh quotient calculation above,
yields $M(f_x)\leq d M(f_{z})$. Now, by using \cite[Equation (4.1)]{basu2015regularized} for the process $z_t$, with reverse characteristic polynomial $\cA$, we obtain
\begin{align}\label{our-M}
\lambdamax(\Sigma) \le 2\pi M(f_x) \le 2\pi d M(f_{z}) \le \frac{d \lambdamax(\Sigmazeta)}{\mumin(\cA)}\,.
\end{align}

 \end{proof}

The following proposition is a straightforward consequence of
the spectral bounds above and \cite[Proposition 2.4]{basu2015regularized}. 
\begin{proposition}\label{prop:basuconcentration}
There exists a constant $c>0$,  such that for any vectors $u,v\in \reals^{dp}$ with $\|u\|\le 1$, $\|v\|\le 1$, and any $\eta\ge 0$,

\begin{align}\label{HW-1}
\prob\left(|u^\sT (\hSigma^{(\ell)} - \Sigma) v| >  \frac{d \lambdamax(\Sigmazeta)}{\mu_{\min}(\cA)} \eta \right)
\le 6 \exp\left(-cn_\ell \min\{\eta^2,\eta\}\right)\,.
\end{align}
\end{proposition}

\subsection{Remarks on proof of Theorem \ref{propo:estimation}}

The key part of establishing Theorem \ref{propo:estimation}
is to establish an appropriate `restricted eigenvalue'
condition as follows:

\begin{propo}\label{pro:RE}
Let $\{z_1, \dotsc, z_{T}\}$ be generated according to the (stable) $\VAR(d)$ process~\eqref{eq:varddef} and let $n = T-d$. Then there exist constants $c\in (0, 1)$ and
$C> 1$ such that
for all $n \ge C \omega^2\log(dp)$, with probability at least $1-\exp(-cn/\omega^2)$, satisfies
\begin{align*}
\<v, (X^\sT X/n) v\> &\ge \alpha \norm{v}^2 -\alpha\tau \norm{v}_1^2. 
\end{align*}
Here, $\alpha$, $\omega$ and $\tau$ are given by:
\begin{align}\label{eq:REparamsdef}
\begin{split}
\omega &=\frac{d \lambdamax(\Sigmazeta)\mumax(\cA) }
{\lambdamin(\Sigmazeta)\mumin(\cA)}\,, \\
\alpha &=\frac{\lambdamin(\Sigmazeta)}{2\mumax(\cA)}\,,  \\
\tau &= \omega^2 \sqrt\frac{\log(dp)}{n} \,.
\end{split}
\end{align}
\end{propo}

Given Proposition \ref{pro:RE}, the estimation result
of Theorem \ref{propo:estimation} is standard (see \cite{buhlmann2011statistics}). 
Proposition~\ref{pro:RE} can be proved analogous to~\cite[Proposition 4.2]{basu2015regularized}, with the following considerations and
minor modifications: 

\begin{enumerate}
\item \cite{basu2015regularized} writes the $\VAR(d)$ model as a 
$\VAR(1)$ model and then vectorize the obtained equation to get a linear regression form 
(cf. Section 4.1 of \cite{basu2015regularized}). This way, 
they prove $I \otimes (X^\sT X/n)$ satisfies a restricted eigenvalue property. Towards this,  the first step in their proof is to show that $X^\sT X/n$ 
satisfies a restricted eigenvalue property, i.e. Proposition \ref{pro:RE}.  
\item \cite[Proposition 4.2]{basu2015regularized} assumes $n \ge C k\max\{\omega^2,1\}\log(dp)$, with $k = \sum_{\ell=1}^{d} \|{\rm vec}({A}^{(\ell)})\|_0$, the total number of nonzero entries of matrices $A_\ell$ and then it is later used to get $\tau\le 1/(Ck)$. However, as the restricted eigenvalue condition is independent of the sparsity of matrices $A^{(\ell)}$,  we can use their result with $k=1$. 
\item The proof involves upper bounding $M(f_x)$, for
which we use Lemma \ref{lem:basuspectral2} in lieu of Lemma \ref{lem:basuspectral}.
\end{enumerate}

\subsection{Proof of Lemma~\ref{lem:biasmatrixdeviation}}\label{proof:lem:biasmatrixdeviation}


The idea is to use Proposition \ref{prop:basuconcentration} along
with the union bound. Fix $i,j\in [dp]$ and let $u = \tfrac{\Omega e_i}{\|\Omega e_i\|}$ and $v = e_j$. Then:
\begin{align*}
\abs{(\Omega\Sell - I)_{ij}} &= \abs{\<\Omega e_i, (\Sell - \Sigma)e_j\> }\\
&= \norm{\Omega e_i} \abs{\<u, (\Sell - \Sigma) v\> }\\
&\le \lambdamax(\Omega) \abs{\<u, (\Sell - \Sigma) v\> } \\
&\le \frac{\mumax(\cA)}{\lambdamin(\Sigmazeta)} \abs{\<u, (\Sell - \Sigma) v\> },
\end{align*}
where the last line uses Lemma \ref{lem:basuspectral} to bound
$\lambdamin(\Sigma)$ from below. Combining this
with Proposition \ref{prop:basuconcentration}, for $\eta \le 1$:
\begin{align*}
\P\Big\{ \abs{ (\Omega\Sell-I)_{ij}    }  \ge d\lambdamax(\Sigmazeta)\eta/\mumin(\cA) \Big\} &\le \P\Big\{ \abs{ \<u, (\Sell - \Sigma)v\>    }  \ge 
\omega \eta \Big\} \\
&\le 6 \exp(-c n_\ell \eta^2).
\end{align*}
Setting $\eta = C \sqrt{\log(dp)/n_\ell} $ for a large enough constant $C$,
the probability bound above is smaller than $ (dp)^{-8}$. With a union
bound over $i, j \in [dp]$:
\begin{align*}
 \P\bigg\{\norm{\Omega\Sell - I}_\infty \ge C\omega \sqrt{\frac{\log(dp)}{n_\ell}  }   \bigg\} &\le (dp)^2 \sup_{i, j} \P\bigg\{\abs{(\Omega\Sell - I)_{ij}} \ge C\omega \sqrt{\frac{\log(dp)}{n_\ell}  }   \bigg\} \\
 &\le (dp)^{-6}.
 \end{align*} 
 This completes the proof.


\subsection{Proof of Theorem~\ref{thm:TSbiasbound}}\label{proof:thm:TSbiasbound} 
Starting from the decomposition~\eqref{eq:TS-decomposition}, we have
\[
\sqrt{n}(\onth - \theta_0) = \Delta_n + W_n\,,
\]
with $\Delta_n = B_n(\Lsth - \theta_0)$.  As explained below~\eqref{eq:TS-decomposition}, $W_n$ is a martingale with respect to filtration $\cF_j = \{\varepsilon_1, \dotsc, \varepsilon_j\}$, $j\in \naturals$ and hence $\E(W_n) = 0$.

We also note that $\|\Delta_n\|_\infty\le \|B_n\|_\infty \|\Lsth - \theta_0\|_1$. Our next lemma bounds $\|B_n\|_\infty$.

\begin{lemma}\label{lem:onlinebias}
Suppose that the Optimization problem~\eqref{eq:opt} is feasible for all $i\in[dp]$. Let $\omega$ and $\gamma$ be:
\begin{align*}
\omega &=  \frac{d \mumax(\cA) \lambdamax(\Sigmazeta)} {\mumin(\cA)\lambdamin(\Sigmazeta)}, \\
\gamma &= \frac{d\lambdamax(\Sigmazeta)}{\mumin(\cA)}.
\end{align*}
Then, with probability at least $1 - (dp)^{-8}$
\begin{align}
\|B_n\|_\infty &\le 
\frac{r_0}{\sqrt{n}} + C(\omega + L\gamma) \sqrt{\frac{\log(dp)}{n}} \sum_{\ell=1} ^{K-1} \Big( \frac{r_{\ell}}{\sqrt{n_\ell}} + \sqrt{r_{\ell}}\Big).\label{eq:onlinebias}
\end{align}
\end{lemma}
The bound provided in Lemma \ref{lem:onlinebias} holds for general batch sizes $r_0, \dotsc, r_{K-1}$. We choose the batch lengths as $r_\ell = \beta^\ell$ for some $\beta >1$ and $\ell = 1, \dotsc, K-1$. We also let $r_0 = \sqrt{n}$ and choose $r_{K-1}$ so that the total lengths of batches add up to $n$ (that is $r_0+r_1+\dotsc+r_{K-1} = n$). Therefore, $K = O(\log _\beta(n))$. Following this choice, bound~\eqref{eq:onlinebias} simplifies to:
\begin{align}
\|B_n\|_\infty \le C_\beta (\omega + \gamma L)  \sqrt{{\log (dp)}}  \,,
\end{align}
for some constant $C_\beta>0$ that depends on the constant $\beta$.

Next by combining Theorem~\ref{propo:estimation} and Lemma~\ref{lem:onlinebias} we obtain that, with probability at least $1-2(dp)^{-6}$
\begin{align}\label{bias-tail}
\|\Delta_n\|_\infty&\le  C_\beta (\omega+L \gamma)   \sqrt{\log (dp)}\cdot \Big(
\frac{s_0\lambda_n}{\alpha}\Big) \nonumber\\
 &\le C_\beta \frac{\lambda_0(\omega + L\gamma)}{\alpha}\frac{s_0\log(dp)}{\sqrt{n}}.
\end{align}
This implies the claim by selecting a $\beta$ bounded away from 1, say $\beta = 1.3$.

It remains to prove the claim on the bias $\E\{\onth - \theta_0 \}$. For this,
define $G$ to be the event where $\Delta_n$ satisfies the upper bound
in Eq.\eqref{bias-tail}. Therefore:
\begin{align*}
\norm{\E\{\onth - \theta_0\}}_\infty &= \frac{\norm{\E\{\Delta_n\}}_\infty}{\sqrt{n}}\\
&\le \frac{ \norm{ \E\{ \Delta_n \ind(G)    \} }_\infty }{\sqrt{n}}
+ \E\{ \norm{\hth^\sL - \theta_0}_1 \ind(G^c)  \}.
\end{align*}
For the first term we use the bound Eq.\eqref{bias-tail}. For the 
second, we use Lemma \ref{lem:sizeofLASSO}:
\begin{align*}
\norm{\E\{\onth - \theta_0\}}_\infty &\le 
\frac{C\lambda_0(\omega + L\gamma)}{\alpha} \frac{s_0 \log p}{n}
+ \frac{\E\{\norm{\eps}^2 \ind(G^c) \}}{n\lambda_n} + 2\norm{\theta_0}_1 \P(G^c).
\end{align*}
It suffices, therefore, to show that the final two terms are at most
$ C\norm{\theta_0}_1/(dp)^6 $. By Holder inequality and $\P(G^c) \le 2(dp)^{-6}$:
\begin{align*}
\frac{\E\{\norm{\eps}^2 \ind(G^c) \}}{n\lambda_n} + 2\norm{\theta_0}_1 \P(G^c)
&\le \frac{\E\{\norm{\eps}^4\}^{1/2} \P(G^c)^{1/2}}{n\lambda_n} + 2\norm{\theta_0}_1 \P(G^c)\\
&\le C \frac{\lambdamax(\Sigmazeta)^2 }{(dp)^{3}\lambda_0 \sqrt{n\log (dp)}}
+ C\frac{\norm{\theta_0}_1}{(dp)^6}.  
\end{align*}
In the high-dimensional regime, the first term is negligible in comparison
to $s_0\log (dp)/n$, which yields, after adjusting $C$ appropriately:
\begin{align*}
\norm{\E\{\onth - \theta_0\}}_\infty &\le \frac{C_1\lambda_0(\omega + L\gamma)}{\alpha} \frac{s_0 \log p}{n} + C_2\frac{\norm{\theta_0}_1}{(dp)^6},
\end{align*}
as required.




It remains to prove Lemma \ref{lem:onlinebias}:
\begin{proof}[Proof of Lemma~\ref{lem:onlinebias}]
For each episode $\ell$, let
\[
\Rell: = \frac{1}{r_\ell} \sum_{t\in E_\ell} x_t x_t^\sT
\]
be the sample covariance in episode $\ell$. Fix $a\in [dp]$ and define $B_{n,a}\equiv \sqrt{n}e_a - \frac{1}{\sqrt{n}} \sum_{\ell=1}^{K-1} r_{\ell} \Rell\mli $. We then have
\begin{align}\label{vNi}
B_{n,a} = \sqrt{n}e_a - \frac{1}{\sqrt{n}}\sum_{\ell=1}^{K-1} r_{\ell}\Rell \mli  = \frac{r_0}{\sqrt{n}} e_a+\sum_{\ell=1}^{K-1}\frac{r_{\ell}}{\sqrt{n}} \Big(e_a - \Rell \mli  \Big)\,,
\end{align}
where we used that $\sum_{\ell=0}^{K-1} r_\ell = n$. By triangle inequality, 
followed by Holder inequality:
\begin{align*}
\norm{B_{n, a}}_\infty &\le \frac{r_0}{\sqrt{n}}
+ \frac{1}{\sqrt{n}}\sum_{\ell=1}^{K-1} {r_{\ell}} \|e_a-\Rell \mli\|_\infty \\
&\le \frac{r_0}{\sqrt{n}} + \sum_{\ell=1}^{K-1}
\frac{r_{\ell}}{\sqrt{n}} \big(\norm{e_a - \Sell m^\ell_a}_\infty+
\norm{(\Sell - \Sigma) m^\ell_a}_\infty +
\norm{ (\Sigma - \Rell) m^\ell_a  }_\infty    \big) \\
&\le \frac{r_0}{\sqrt{n}}
+ \sum_{\ell=1}^{K-1}
\frac{r_{\ell}}{\sqrt{n}} \big(\norm{e_a - \Sell m^\ell_a}_\infty+
\norm{\Sell - \Sigma}_\infty   \norm{m^\ell_a}_1 +
\norm{ \Sigma - \Rell}_\infty  \norm{m^\ell_a}_1  \big)
\end{align*}
We now bound each of the three terms appearing in the sum above:
\begin{enumerate}
\item By the construction of decorrelating vectors $\mli$ as in optimization~\eqref{eq:opt}, we have
\begin{align}\label{eq:step1}
\|\Sell \mli - e_a\|_\infty\le \mu_{\ell}\,, \quad \ell=0,\dotsc, K-1\,.
\end{align}
\item Also by construction, $\norm{m^\ell_a}_1\le L$. From an argument
similar to that of Lemma \ref{lem:biasmatrixdeviation}, 
$\norm{\Sell -\Sigma}_\infty \le C\gamma \sqrt{\log (dp)/n_\ell}$ with
probability at least $1-K(dp)^{-9}$, where $\gamma = d\lambdamax(\Sigmazeta)/\mumin(\cA)$. Therefore, with the same probability,
the third term is at most $CL\gamma \sqrt{\log(dp)/n_\ell} $.
\item Again, by construction $\norm{m^\ell_a}_1\le L$. Similar to
Lemma \ref{lem:biasmatrixdeviation}, $\norm{\Rell - \Sigma}_\infty$\
is at most $C\gamma \sqrt{\log (dp)/r_\ell}$ with probability at 
least $1-K(dp)^{-9}$. 
\end{enumerate}

Combining these and the fact that we set $\mu_\ell = C\omega \sqrt{\log(dp)/n}$ we have that, with probability at least
$1-2K(dp)^{-9}$, 
\begin{align*}
\norm{B_{n, a}}_\infty &\le \frac{r_0}{\sqrt{n}} +  \frac{C}{\sqrt{n}} \sum_{\ell=0}^{K-2}
r_{\ell}  \Bigg(\omega \sqrt\frac{\log (dp)}{n_\ell}+  L \gamma  \sqrt\frac{\log(dp)}{n_\ell} +
L \gamma \sqrt\frac{\log(dp)}{r_{\ell}}   \Bigg) \\
&\le \frac{r_0}{\sqrt{n}} + C(\omega + L\gamma) \sqrt{\frac{\log(dp)}{n}} \sum_{\ell=0} ^{K-2} \Big( \frac{r_{\ell}}{\sqrt{n_\ell}} + \sqrt{r_{\ell}}\Big).
\end{align*}

This bound holds uniformly over $a\in [dp]$, and since
$\norm{B_n}_\infty = \sup_a \norm{B_{n, a}}_\infty$, the same bound
holds for $\norm{B_n}_\infty$. This completes the proof. 
\end{proof}
\subsection{Proof of Lemma~\ref{lem:stab-W-TS}}\label{proof:lem:stab-W-TS}
We start by proving Claim~\eqref{eq:conditionalvar}. Let $m_a = \Omega e_a$ be the first column
of the inverse (stationary) covariance. Using
the fact that
$\E\{x_t x_t^\sT\} = \Sigma$ we have
$\<m_a,\E\{ x_t x_t^\sT\} m_a\> =  \Omega_{a,a}$, 
which is to be the dominant term in the conditional variance $V_{n,a}$. 
Using the shorthand $\sigma^2 = \Sigmazeta_{i, i}$
Therefore, we decompose the difference as follows: 
\begin{align}
V_{n,a} -\Omega_{a,a} &= \frac{\sigma^2}{n} \Bsum \Big[\<m^\ell_a, x_t\>^2 - \Omega_{a,a} \Big] - \frac{r_0\sigma^2}{n}\Omega_{a,a} \nonumber\\
& = \frac{\sigma^2}{n} \Bsum \Big[ \<m^\ell_a , x_t\>^2 - 
\<m_a, \E\{x_t x_t^\sT\} m_a\>\Big] - \frac{r_0 \sigma^2}{n}\Omega_{a,a} \nonumber\\
& = \frac{\sigma^2}{n}\Bsum [\<m^\ell_a, x_t\>^2 - \<m_a, x_t\>^2]\nonumber\\
&\quad + \frac{1}{n}\sum_{t =0}^{n-1} \<m_a, (x_t x_t^\sT - \E\{x_t x_t^\sT\} )m_a\> - \frac{r_0\sigma^2}{n}\Omega_{a,a}\,.\label{eq:Err0}
\end{align}
We treat each of these three terms separately. 
Write
\begin{align}
\bigg|\frac{1}{n}\Bsum  [\<m^\ell_a, x_t\>^2 - \<m_a, x_t\>^2 ]\bigg| & =
\frac{1}{n}\bigg|\Bsum  [\<m^\ell_a - m_a, x_t\> 
 \<m^\ell_a + m_a, x_t \> ]\bigg| \nonumber \\
 &\le \frac{1}{n} \bigg\|\Bsum \<m^\ell_a - m_a, x_t\>  x_t \bigg\|_\infty
 \|m^\ell_a + m_a\|_1\nonumber\\
 &\le \frac{2L}{n} \bigg\|\Bsum \<m^\ell_a - m_a, x_t\>  x_t \bigg\|_\infty\,.
  \label{eq:mstartrackingerror}
\end{align}  
To bound the last quantity, note that
\begin{align}\label{eq:mstartrackingerror2}
\frac{1}{n} \bigg\|\Bsum \<m^\ell_a - m_a, x_t\>  x_t \bigg\|_\infty &\le  \bigg\|e_a - \frac{1}{n} \Bsum \<m^\ell_a, x_t\>  x_t \bigg\|_\infty \nonumber\\
&+ \bigg\|e_a - \frac{1}{n} \Bsum \<m_a, x_t\>  x_t \bigg\|_\infty \nonumber \\
& = \bigg\|e_a - \frac{1}{n} \sum_{\ell=1}^{K-1} r_{\ell} \Rell m^\ell_a\bigg\|_\infty
+ \bigg\|e_a -   \hSigma^{(K)} m_a \bigg\|_\infty\nonumber \\
& = \frac{1}{\sqrt{n}}\|B_{n,a}\|_\infty
+ \bigg\|e_a -   \hSigma^{(K)} m_a \bigg\|_\infty\nonumber \\
&\le CL\gamma\sqrt{\frac{\log (dp)}{n}} + C\omega \sqrt{\frac{\log (dp)}{n}} \le C(L\gamma + \omega) \sqrt{\frac{\log (dp)}{n}}\,,
\end{align}
for some constant $C$. The last inequality follows from the positive events of Lemma \ref{lem:onlinebias} and Lemma \ref{lem:biasmatrixdeviation}. Combining Equations~\eqref{eq:mstartrackingerror} and \eqref{eq:mstartrackingerror2}, we obtain
\begin{align}\label{eq:Err1}
\bigg|\frac{1}{n}\Bsum  [\<m^\ell_a, x_t\>^2 - \<m_a, x_t\>^2 ]\bigg|  &\le  CL(\omega + L\gamma)\sqrt{\frac{\log (dp)}{n}}.
\end{align}


For the second term in \eqref{eq:Err0}, we can use Proposition \ref{prop:basuconcentration}
with $v = u = m_a/\lVert{m_a}\rVert, \eta =  C\sqrt{\log (dp)/n} $ to obtain
\begin{align}
\Big\lvert \frac{1}{n} \sum_{t=0}^{n-1} \<m_a, (x_t x_t^\sT - \E\{x_t x_t^\sT\} )m_a\> \Big\rvert & = \big\lvert \< m_a, (\hSigma^{(K-1)} - \Sigma) m_a\>\big\rvert \nonumber\\
&\le \frac{Cd \lambdamax(\Sigmazeta)}{\mumin(\cA)} \|m_a\|^2 \sqrt{\frac{\log (dp)}{n}}\nonumber\\
 &\le \frac{Cd \lambdamax(\Sigmazeta)}{\mumin(\cA) \lambdamin(\Sigma)^2} \sqrt{\frac{\log (dp)}{n}}\\
 &\le \frac{C\omega}{\alpha}\sqrt{\frac{\log (dp)}{n}}\,,  \label{eq:Err2}
\end{align}
where we used that $\|m_a\| = \|\Omega e_a\| \le \lambdamax(\Omega) = \lambdamin(\Sigma)^{-1} \le 1/\alpha$.
For the third term, we have $r_0 = \sqrt{n}$. Also, $\Omega_{a,a} \le \lambdamax(\Omega) \le 1/\alpha$. Therefore, this term is $O(1/\alpha\sqrt{n})$.  Combining this bound with~\eqref{eq:Err1} and \eqref{eq:Err2} in Equation \eqref{eq:Err0} we get the Claim~\eqref{eq:conditionalvar}. 

We next prove Claim~\eqref{eq:summand}. 
Note that $|\eps_t| = |\zeta_{t+d, i}|$ is bounded with $\sigma\sqrt{2\log(n)}$, with high probability for $t\in [n]$, by tail bound for Gaussian variables.
In addition, $\max_\ell \lvert\<m^\ell_a, x_t\> \rvert
\le \|m^\ell_a\|_1 \|x_t\|_\infty \le L \|x_t\|_\infty \le L \norm{X}_\infty $. Note that variance of each entry $x_{t,i}$ is bounded by $\Sigma_{ii}\le\lambda_{\max}(\Sigma)$. Hence, by tail bound for Gaussian variables and union bounding we have
\begin{align}
\prob\left(\norm{X}_\infty < \sqrt{2 \lambda_{\max}(\Sigma) \log(dpn)}\right) \ge 1 - (pdn)^{-2}\,,
\end{align}
Putting these bounds together we get 
\begin{align*}
&\max \Big\{\frac{1}{\sqrt{n}}\lvert \<m^\ell_a, x_t\> \eps_t \rvert:\, \ell \in [K-2],\, t\in [n]\Big\} \\
&\le \frac{1}{\sqrt{n}} L \sqrt{2 \lambda_{\max}(\Sigma) \log(dpn)} \sigma\sqrt{2\log(n)}\\
&\le  2L\sigma \sqrt{\lambda_{\max}(\Sigma)} \;\frac{\log(dpn)}{\sqrt{n}} \\
&\le 2 L_0\sigma  \|\Omega\|_1 \left( \frac{2\pi d\lambdamax(\Sigmazeta)}{\mumin(\cA)} \right)^{1/2} \frac{\log(dpn)}{\sqrt{n}} 
= o(1)\,,
\end{align*} 
where in the last inequality we used Lemma~\ref{lem:basuspectral2} to upper bound $\lambda_{\max}(\Sigmazeta)$. The conclusion that the final expression is $o(1)$ follows from Assumption~\ref{assmp:TS}.
\subsection{Proof of Proposition \ref{pro:SS}}\label{proof:pro:SS}
We prove that for all $x\in \reals$, 
\begin{align}\label{eq:UB-T}
\lim_{n\to \infty}\sup_{\|\theta_0\|_0\le s_0}  \prob\Big\{\frac{\sqrt{n}(\onth_a - \theta_{0,a})}{\sqrt{V_{n,a}}} \le x\Big\}  \le \Phi(x)\,.
\end{align}
We can obtain a matching lower bound by a similar argument which implies the result.

Invoking the decomposition \eqref{mydecompose} we have
\[
\frac{\sqrt{n}(\onth_a - \theta_{0,a})}{\sqrt{V_{n,a}}} = \frac{W_n}{\sqrt{V_{n,a}}} + \frac{\Delta_n}{\sqrt{V_{n,a}}}\,.
\]
By Corollary~\ref{cor:noise-TS}, we have that $\widetilde{W}_n \equiv W_n/\sqrt{V_{n,a}} \to \normal(0,1)$ in distribution.  Fix an arbitrary $\eps>0$ and write
\begin{align*}
 \prob\Big\{\frac{\sqrt{n}(\onth_a - \theta_{0,a})}{\sqrt{V_{n,a}}} \le x\Big\}  
 &=  \prob\Big\{\widetilde{W}_n +\frac{\Delta_n}{\sqrt{V_{n,a}}} \le x \Big\}\\
 &\le \prob\{\widetilde{W}_n\le x+\eps\} + \prob\Big\{\frac{|\Delta_a|}{\sqrt{V_{n,a}}} \ge \eps\Big\}    
\end{align*}
By taking the limit and using Equation~\eqref{mydecompose}, we get
\begin{align}\label{eq:dec-B1}
\lim_{n\to \infty}\sup_{\|\theta_0\|_0\le s_0}  \prob\Big\{\frac{\sqrt{n}(\onth_a - \theta_{0,a})}{\sqrt{V_{n,a}}} \le x\Big\}\le
 \Phi(x+\eps) + \lim_{n\to \infty}\sup_{\|\theta_0\|_0\le s_0} \prob\Big\{\frac{|\Delta_a|}{\sqrt{V_{n,a}}} \ge \eps\Big\} 
\end{align}
We show that the limit on the right hand side vanishes for any $\eps>0$. By virtue of Lemma~\ref{lem:stab-W-TS} (Equation \eqref{eq:conditionalvar}), we have
\begin{align}
 \lim_{n\to \infty} \prob\Big\{\frac{|\Delta_a|}{\sqrt{V_{n,a}}} \ge \eps\Big\} 
 &\le \lim_{n\to \infty} \prob\Big\{\frac{|\Delta_a|}{\sigma\sqrt{\Omega_{a,a}}} \ge \eps\Big\}\nonumber\\
 &\le \lim_{n\to \infty} \prob\Big\{{|\Delta_a|} \ge \eps \sigma\sqrt{\Omega_{a,a}} \Big\}\nonumber\\
 &\le \lim_{n\to \infty} (dp)^{-4} = 0\,.\label{eq:Delta-B}
\end{align}
Here, in the last inequality we used that $s_0 (L\gamma +\omega)= o(\sqrt{n}/\log(dp))$ and therefore, for large enough $n$, $\eps \sigma\sqrt{\Omega_{a,a}}$ exceeds the bound~\eqref{mydecompose-B} of Theorem \ref{thm:TSbiasbound}. 

Using \eqref{eq:Delta-B} in bound~\eqref{eq:dec-B1} and then taking the limit $\eps\to 0$, we obtain \eqref{eq:UB-T}.
  
\section{Proofs of Section~\ref{sec:discussion}}
\subsection{Proof of Lemma \ref{lem:Duality}}\label{proof-lem:Duality}
Rewrite the optimization problem~\eqref{eq:opt} as follows:
\begin{align}\label{myprimal}
\begin{split}
\text{minimize}\quad &m^\sT \Sell m\\
\text{subject to}\quad &\<z,\Sell m - e_a\> \le \mu_\ell,\quad \|m\|_1\le L,\quad \|z\|_1=1\,, 
\end{split}
\end{align}
The Lagrangian is given by
\begin{align}
\cL(m,z,\lambda)=m^\sT \Sell m+\lambda(\<z,\Sell m-e_a\>-\mu_\ell), \quad \|z\|_1=1,\quad \|m\|_1\le L\,,
\end{align}
If  $\lambda \le 2L$, minimizing Lagrangian over $m$ is equivalent to $\frac{\partial \cL}{\partial m}=0$ and we get $m_*=-\lambda z_*/2$. The dual problem is then given by
\begin{align}
\begin{split}
\text{maximize}\quad &- \frac{\lambda^2}{4} z^\sT \Sell z- \lambda\<z,e_a\>-\lambda\mu_\ell\\
\text{subject to}\quad &\frac{\lambda}{2}\leq L, \quad  \|z\|_1=1\,, 
\end{split}
\end{align}
As $\|z\|_1=1$, by introducing $\beta=-\frac{\lambda}{2}z$,  we get $\|\beta\|_1=\frac{\lambda}{2}$. Rewrite the dual optimization problem in terms of $\beta$ to get
\begin{align}\label{mydual}
\begin{split}
\text{minimize}\quad &\frac{1}{2} \beta^\sT \Sell \beta- \<\beta,e_a\>+\mu_\ell\|\beta\|_1\\
\text{subject to}\quad &\|\beta\|_1 \le L\,,
\end{split}
\end{align}
Given $\beta_*$ as the minimizer of the above optimization problem, from the relation of $\beta$ and $z$ we realize that $m_*=\beta_*$.

Also note that since optimization \eqref{mydual} is the dual to problem~\eqref{myprimal}, we have that if \eqref{myprimal} is feasible then the problem \eqref{mydual} is bounded.
  \subsection{Proof of Lemma~\ref{lem:omeg-estimation}}\label{proof-lem:omeg-estimation}
  By virtue of Proposition~\ref{pro:RE}, the sample covariance $\hSigma$ satisfies RE condition, $\hSigma \sim \RE(\alpha,\tau)$, where
  \begin{align}
  \alpha = \frac{\lambdamin(\Sigmazeta)}{2\mumax(\cA)}\,, \quad \quad \tau = C\omega^2\sqrt\frac{\log(dp)}{n}\,,
  \end{align}
  and by the sample size condition we have $s_\Omega <1/32\tau$.
  
 Hereafter, we use the shorthand $m^*_a = \Omega e_a$ and let $\cL(m)$ be the objective function in the optimization~\eqref{eq:M-offline}. By optimality of $m_a$, we have $\cL(m_a^*)\le \cL(m_a)$. Defining the error vector $\nu \equiv m_a-m_a^*$ and after some simple algebraic calculation we obtain the equivalent inequality
\begin{align}\label{eq:ineq1}
\frac{1}{2} \nu^\sT \hSigma \nu \le \<\nu, e_a - \hSigma m^*_a\> + \mu_n (\|m^*_a\|_1 - \|m^*_a + \nu\|_1 )\,.
\end{align}
In the following we first upper bound the right hand side.  By Lemma~\ref{lem:biasmatrixdeviation} (for $\ell=K$ and $n_{K} = n$), we have that with high probability 
\[
\<\nu, e_a - \hSigma m^*_a\> \le \|\nu\|_1 a\sqrt{\frac{\log (dp)}{n}} = (\|\nu_S\|_1+ \|\nu_{S^c}\|_1) \frac{\mu_n}{2}\,,
\]
where $S = \supp(\Omega e_a)$ and hence $|S| \le s_\Omega$. On the other hand,
\[
\|m_a + \nu\|_1 - \|m^*_a\|_1 \ge (\|m^*_{a,S}\|_1 - \|\nu_S\|_1) + \|\nu_{S^c}\|_1 - \|m^*_a\|_1 = \|\nu_{S^c}\|_1- \|\nu_{S}\|_1\,.
\]
Combining these pieces we get that the right-hand side of~\eqref{eq:ineq1} is upper bounded by 
\begin{align}\label{eq:ineq2}
 (\|\nu_S\|_1+ \|\nu_{S^c}\|_1) \frac{\mu_n}{2} + \mu_n \left(\|\nu_S\|_1 - \|\nu_{S^c}\|_1 \right) = \frac{3}{2}\mu_n \|\nu_S\|_1 - \frac{1}{2}\mu_n \|\nu_{S^c}\|_1  \,,
\end{align}
 Given that $\hSigma \succeq 0$, the left hand side of~\eqref{eq:ineq1} is non-negative, which implies that $\|\nu_{S^c}\|_1 \le 3 \|\nu_S\|_1$ and hence 
\begin{align}\label{eq:L1nu}
\|\nu\|_1\le 4\|\nu_S\|_1\le 4\sqrt{\som} \|\nu_S\|_2 \le 4\sqrt{\som} \|\nu\|_2\,.
\end{align} 

Next by using the restricted eigenvalue condition for $\hSigma$ we write
\begin{align}\label{eq:ineq3}
 \nu^\sT \hSigma \nu \ge \alpha \|\nu\|_2^2 - \alpha\tau \|\nu\|_1^2 \ge \alpha(1  - 16\som \tau)\|\nu\|_2^2 \ge \frac{\alpha}{2} \|\nu\|_2^2\,,
\end{align}
where we used $\tau \le1/ (32 \som)$ in the final step. 

Putting~\eqref{eq:ineq1}, \eqref{eq:ineq2} and \eqref{eq:ineq3} together, we obtain
\[
\frac{\alpha}{4}\|\nu\|_2^2 \le \frac{3}{2} \mu_n \|\nu_S\|_1 \le 6\sqrt{\som} \mu_n \|\nu\|_2\,.
\]
Simplifying the bound and using equation \ref{eq:L1nu}, we get
\begin{align*}
\|\nu\|_2 &\le \frac{24}{\alpha} \sqrt{s_\Omega} \mu_n\,,\\
\|\nu\|_1
&\le \frac{96}{\alpha} s_\Omega \mu_n\,,\end{align*}
which completes the proof.
  \subsection{Proof of Theorem~\ref{thm:offline-TS}}\label{proof-thm:offline-TS}
  Continuing from the decomposition~\eqref{eq:new-decomposition} we have
  \begin{align}\label{decom2}
  \sqrt{n}(\offth - \theta_0) = \Delta_1 + \Delta_2 + Z\,,
  \end{align}
  with $Z = \Omega X^\sT \eps/\sqrt{n}$. By using Lemma~\ref{lem:biasmatrixdeviation} (for $\ell=K$) and recalling the choice of $\mu = \mya\sqrt{\log (dp)/n}$
  we have that the following optimization is feasible, with high probability:
  \begin{align*}
  &\text{minimize} \;\; \; m^\sT\hSigma m\\
  &\text{subject to} \;\; \|\hSigma m - e_a\|_\infty\le \mu\,.
  \end{align*}
  Therefore, optimization \eqref{eq:M-offline} (which is shown to be its dual in Lemma~\eqref{lem:Duality}) has bounded solution. Hence, its solution should satisfy the KKT condition which reads as
  \begin{align}\label{eq:KKT-TS}
  \hSigma m_a - e_a + \mu \sign(m_a) = 0\,,
  \end{align}
 which implies $\|\hSigma m_a - e_a\|_\infty \le \mu$. Invoking the estimation error bound of Lasso for time series (Proposition \ref{propo:estimation}),
 we bound $\Delta_1$ as
 \begin{align}\label{Delta1}
 \|\Delta_1\|_\infty\le C \sqrt{n} \mu s_0 \sqrt{\frac{\log p}{n}} = O_P\Big(s_0 \frac{\log(dp)}{\sqrt{n}}\Big)\,. 
 \end{align}
 
 We next bound the bias term $\Delta_2$. By virtue of~\cite[Proposition 3.2]{basu2015regularized} we have the deviation bound $\|X^\sT \eps\|_\infty/\sqrt{n} = O_P(\sqrt{\log(dp)})$, which
 in combination with Lemma~\ref{lem:omeg-estimation} gives us the following bound 
 \begin{align}
 \|\Delta_2\|_\infty \le \left(\max_{i\in [dp]} \| (M-\Omega)e_i)\|\right) \left(\frac{1}{\sqrt{n}}\|X^\sT\eps\|_\infty\right) = O_P\Big(s_\Omega \frac{\log(dp)}{\sqrt{n}}\Big)\,.
 \end{align}
Therefore, letting $\Delta = \Delta_1+\Delta_2$, we have $\|\Delta\|_\infty = o_P(1)$, by recalling our assumption $s_0=o(\sqrt{n}/\log (dp))$ and
$s_\Omega = o(\sqrt{n}/\log(dp))$.

Our next lemma is analogous to Lemma~\ref{lem:stab-W-TS} for the covariance of the noise component in the offline debiased estimator, and its proof is deferred to Section~\ref{lem:var-con}.
\begin{lemma}\label{lem:var-con}
Assume that $s_\Omega = o(\sqrt{n}/\log(dp))$ and $\Lambda_{\min}(\Sigma_\epsilon)/{\mu_{\max}({\cA})} > c_{\min}>0$ for some constant $c_{min}>0$.
For $\mu = \mya\sqrt{\log(dp)/n}$ and the decorrelating vectors $m_i$ constructed by~\eqref{eq:M-offline}, the following holds.  For any fixed sequence of integers
$a(n)\in [dp]$, we have
\begin{align}
m_a^\sT\hSigma m_a = \Omega_{a,a}+o_P(1/\sqrt{\log(dp)})\,.
\end{align}
\end{lemma}
 We are now ready to prove the theorem statement. We show that 
 \begin{align}
 \lim_{n\to\infty} \sup_{\|\theta_0\|_0\le s_0} \prob\left\{\frac{\sqrt{n} (\offth_a - \theta_{0,a})}{\sqrt{V_{n,a}}} \le u \right\}  \le \Phi(u)\,.
 \end{align}
A similar lower bound can be proved analogously. By the decomposition~\eqref{decom2} we have
\[
\frac{\sqrt{n} (\offth_a - \theta_{0,a})}{\sqrt{V_{n,a}}} = \frac{\Delta_a}{\sqrt{V_{n,a}}} + \frac{Z_a}{\sqrt{V_{n,a}}}\,.
\] 
Define 
\[
\widetilde{Z}_a \equiv \frac{Z_a}{\sigma\sqrt{\Omega_{a,a}}} = \frac{1}{\sigma\sqrt{n\Omega_{a,a}}} (\Omega X^\sT \eps)_a =  \frac{1}{\sigma\sqrt{n\Omega_{a,a}}}  \sum_{i=1}^n e_a^\sT \Omega  x_i \eps_i\,.  
\]
Since $\eps_i$ is independent of $x_i$, the summand $\sum_{i=1}^n e_a^\sT\Omega x_i \eps_i$ is a martingale. Furthermore, $\E[(e_a^\sT\Omega x_i \eps_i)^2] = \sigma^2 \Omega_{a,a}$. Hence, by a martingale central limit theorem \cite[Corollary 3.2]{hall2014martingale}, we have that $\widetilde{Z}_a \to \normal(0,1)$ in distribution. In other words, 
\begin{align}\label{tilZ}
\lim_{n\to \infty} \prob\{\widetilde{Z}_a u \} = \Phi(u)\,.
\end{align}
Next, fix $\delta \in (0,1)$ and write
\begin{align*}
\prob\left\{\frac{\sqrt{n} (\offth_a - \theta_{0,a})}{\sqrt{V_{n,a}}}\le u \right\} &= \prob\left\{\frac{\sqrt{\Omega_{a,a}}}{\sqrt{V_{n,a}}} \widetilde{Z}_a+ \frac{\Delta_a}{\sqrt{V_{n,a}}} \le u \right\}\\
& \le\prob\left\{\frac{\sqrt{\Omega_{a,a}}}{\sqrt{V_{n,a}}} \widetilde{Z}_a \le u+\delta\right\} + \prob\left\{\frac{\Delta_a}{\sqrt{V_{n,a}}} \ge \delta\right\}\\
 & \le \prob\left\{\widetilde{Z}_a\le u+2\delta+ \delta|u| \right\} +\prob\left\{\Big|\frac{\sqrt{\Omega_{a,a}}}{\sqrt{V_{n,a}}} -1 \Big| \ge \delta \right\}  \\
 &\quad+ \prob\left\{ \frac{\Delta_a}{\sqrt{V_{n,a}}} \ge \delta\right\}\,.
\end{align*}
Now by taking the limit of both sides and using~\eqref{tilZ} and Lemma \ref{lem:var-con}, we obtain
\begin{align}
\lim\sup_{n\to\infty} \sup_{\|\theta_0\|_0\le s_0} &\prob\left\{\frac{\sqrt{n} (\offth_a - \theta_{0,a})}{\sqrt{V_{n,a}}}\le u \right\} \le\nonumber\\
&\Phi(u+2\delta + \delta\abs{u}) + \lim\sup_{n\to\infty} \sup_{\|\theta_0\|_0\le s_0}\prob\left\{ \frac{\Delta_a}{\sqrt{V_{n,a}}} \ge \delta\right\}\,.
\end{align}
Since $\delta\in (0,1)$ was chosen arbitrarily, it suffices to show that the limit on the right hand side vanishes.  To do that, we use Lemma~\ref{lem:var-con} again to write
\begin{align*}
 \lim_{n\to \infty}\sup_{\|\theta_0\|_0\le s_0} \prob\Big\{\frac{|\Delta_a|}{\sqrt{V_{n,a}}} \ge \delta\Big\} 
 &\le \lim_{n\to \infty}\sup_{\|\theta_0\|_0\le s_0} \prob\Big\{\frac{|\Delta_a|}{\sigma\sqrt{(\Omega_{a,a}}} \ge\delta \Big\}\nonumber\\
 &\le \lim_{n\to \infty}\sup_{\|\theta_0\|_0\le s_0} \prob\Big\{{|\Delta_a|} \ge \delta \sigma\sqrt{\Omega_{a,a}} \Big\} = 0\,,
 \end{align*}
 where the last step follows since we showed $\|\Delta\|_\infty = o_P(1)$.  The proof is complete.
 
 \subsubsection{Proof of Lemma~\ref{lem:var-con}}\label{proof-lem:var-con}

 By invoking bound~\eqref{their-min} on minimum eigenvalue of the population covariance, we have
 \begin{align}
\lambdamin(\Sigma) \ge \frac{\lambdamin(\Sigmazeta)}{\mu_{\max}({\cA})},
\end{align}
bounded away from 0 by our assumption.  Therefore, $\lambdamax(\Omega) = \lambdamin(\Sigma)^{-1}$ is bounded away from $\infty$. Since $\Omega\mge 0$, we have 
$|\Omega_{a,b}|\le \sqrt{\Omega_{a,a}\Omega_{b,b}}$ for any two indices $a,b\in[dp]$. Hence, $|\Omega|_\infty \le 1/\lambdamin(\Sigma)$.
This implies that $\|\Omega e_a\|_1\le s_\Omega/\lambdamin(\Sigma)$. Using this observation along with the bound established in Lemma~\ref{lem:omeg-estimation}, we obtain
\begin{align}\label{aux1}
\|m_a\|_1 \le \|\Omega e_a\| + \|m_a - \Omega e_a\|_1 \le \frac{s_\Omega}{\lambdamin(\Sigma)} + \frac{192\mya}{\alpha} s_\Omega \sqrt{\frac{\log(dp)}{n}} = O(s_\Omega)\,.
\end{align}
We also have
\begin{align}\label{aux2}
 \|m_a - \Omega e_a\|_\infty \le  \|m_a - \Omega e_a\|_1 = O\Big(s_\Omega \sqrt{\frac{\log(dp)}{n}}\Big) \,.
\end{align}
In addition, by the KKT condition~\eqref{eq:KKT-TS} we have
\begin{align}\label{aux3}
\|\hSigma m_a - e_a\|_\infty\le \mu\,.
\end{align}
Combining bounds~\eqref{aux1}, \eqref{aux2} and \eqref{aux3}, we have
\begin{align*}
|m_a^\sT\hSigma m_a - \Omega_{a,a}| &\le |(m_a^\sT\hSigma - e_a^\sT)m_a | + |e_a^\sT m_a - \Omega_{a,a}|\\
&\le \|m_a^\sT\hSigma - e_a^\sT\|_\infty \|m_a\|_1 + \|m_a - \Omega e_a\|_\infty\\
&= O\Big(s_\Omega \sqrt{\frac{\log(dp)}{n}}\Big) = o(1/\sqrt{\log(dp)})\,,
\end{align*}
which completes the proof.
\section{Proofs of Section \ref{sec:batch}}\label{sec:batchproofs}

\subsection{Consistency results for LASSO under adaptively collected samples}


Theorem \ref{thm:batchlassoerr} shows
that, under an appropriate compatibility condition,
the LASSO estimate admits $\ell_1$ error at a  rate of $s_0 \sqrt{\log p/n}$.
Importantly, despite the adaptivity introduced by
the sampling of data, the  error of LASSO estimate 
has the same asymptotic rate as
expected without adaptivity. With slightly stronger
restricted-eigenvalue conditions on the covariances
$\E\{xx^\sT\}$ and $\E\{xx^\sT\vert \<x, \htheta^1\>\ge \varsigma\}$, 
 it is also possible to 
extend Theorem \ref{thm:batchlassoerr} to show
$\ell_2$ error of order $s_0 \log p/n$, analogous
to the non-adaptive setting. However, since
the $\ell_2$ error rate will not be used for our 
analysis of online debiasing, we do not pursue
this direction here.

\subsubsection{Proof of Theorem~\ref{thm:batchlassoerr}}\label{proof:batchlassoerr}

The important technical step is to prove that, under the conditions
specified in Theorem \ref{thm:batchlassoerr}, the \emph{sample 
covariance} $\hSigma = (1/n)\sum_i x_i x_i^\sT$ is
$(\phi_0/4, \supp(\theta_0))$ compatible. 

\begin{proposition}\label{prop:batchsamplecompatibility}
With probability exceeding $1-p^{-4}$ the sample
covariance $\hSigma$ is $(\phi_0/4, \supp(\theta_0))$
compatible when $n_1\vee n_2 \ge C (\kappa^4/\phi_0^2) s_0^2 \log p $, 
for an absolute constant $C>0$. 
\end{proposition}

Let $\hSigma^{(1)} $
and $\hSigma^{(2)} $ denote the sample covariances of
each batch, i.e. $\hSigma^{(1)}  = (1/n_1) \sum_{i\le n_1} x_ix_i^\sT$
and similarly $\hSigma^{(2)} = (1/n_2)\sum_{i> n_1} x_i x_i^\sT$. We
also let $\Sigma^{(2)}$ be the conditional covariance
$\Sigma^{(2)} = \Sigma^{(2)}(\htheta^1)= \E\{x x^\sT | \<x, \htheta^1\> \ge \varsigma \}$. 
We first prove that at
least one of the sample
covariances $\hSigma^{(1)}$ and $\hSigma^{(2)} $ 
closely approximate
their population counterparts,
and that this implies they
are $(\phi_0/2, \supp(\theta_0))$-compatible.
\begin{lemma}\label{lem:batchentrywiseSampCovBnd}  
With probability at least $1- p^{-4}$ 
\begin{align*}
\norm{\hSigma^{(1)} - \Sigma}_\infty \wedge 
\norm{\hSigma^{(2)} - \Sigma^{(2)} }_\infty &\le 
12\kappa^2 \sqrt{\frac{\log p}{n}},
\end{align*}
\end{lemma}
\begin{proof}
Since $n=n_1+n_2 \le 2\max(n_1, n_2)$, 
at least one of $n_1$ and $n_2$ exceeds
$n/2$. We assume that $n_2 \ge n/2$, and
prove that $\norm{\hSigma^{(2)} - \Sigma^{(2)} }_\infty $
satisfies the bound in the claim. The case
$n_1 \ge n/2$ is similar.  
Since we are proving the case $n_2 \ge n/2$, for notational convenience, we assume probabilities and 
expectations in the rest of the proof are conditional
on the first batch $(y_1, x_1), \dots (y_{n_1}, x_{n_1})$,
and omit this in the notation.

For a fixed pair $(a, b) \in [p]\times[p]$:
\begin{align*}
\hSigma^{(2)} _{a, b} - \Sigma^{(2)}_{a, b} 
= \frac{1}{n_2} \sum_{i> n_1} x_{i, a} x_{i, b} - \E\{x_{i, a} x_{i, b}\}
\end{align*}
Using Lemma \ref{lem:subexpprod} we have that
$\norm{x_{i, a} x_{i, b}}_{\psi_1} \le 2 \norm{x_i}_{\psi_2}^2 \le 2 \kappa^2$ almost surely. Then using
the tail inequality Lemma \ref{lem:subexptail} we have
for any $\eps \le 2e\kappa^2$
\begin{align*}
\P\Big\{\abs{\hSigma^{(2)} _{a, b} - \Sigma^{(2)} _{a, b}} \ge \eps  \Big\}
&\le 2 \exp\Big\{ -\frac{n_2 \eps^2 }{6 e \kappa^4} \Big\}
\end{align*}
With $\eps = \eps(p, n_2, \kappa) = 12 \kappa^2 \sqrt{\log p/n_2}  \le 20 \kappa^2 \sqrt{\log p /n }$ we have that
$\P\{ \abs{\hSigma^{(2)} _{a, b} - \Sigma^{(2)}_{a, b}} \ge \eps({p, n_2, \kappa})\} \le  p^{-8}$, whence the claim follows by union bound over pairs $(a, b)$.
\end{proof}

\begin{lemma}[{\cite[Corollary~6.8]{buhlmann2011statistics}}]\label{lem:compatibComparison}
Suppose that $\Sigma$ is $(\phi_0, S)$-compatible. Then
any matrix $\Sigma'$ such that 
$\norm{\Sigma' - \Sigma}_\infty \le \phi_0/(32 |S|)$ is 
 $(\phi_0/2, S)$-compatible.
\end{lemma}
We can now prove Proposition \ref{prop:batchsamplecompatibility}.
\begin{proof}[Proof of Proposition \ref{prop:batchsamplecompatibility}]
Combining Lemmas \ref{lem:batchentrywiseSampCovBnd} and \ref{lem:compatibComparison} yields that, with probability
$1-p^{-4}$, at least one of $\hSigma^{(1)} $ and $\hSigma^{(2)} $ are $(\phi_0/2, \supp(\theta_0))$-compatible provided 
\begin{align*}
12\kappa^2 \sqrt{\frac{\log p}{n}} &\le \frac{\phi_0}{32 s_0},\\
\text{ which is implied by } n  &\ge  \Big(  \frac{400 \kappa^2 s_0}{ \phi_0} \sqrt{\log p} \Big)^2. 
\end{align*}
Since $\hSigma = (n_1/n) \hSigma^{(1)} + (n_2/n) \hSigma^{(2)} $ and 
at least one of $n_1/n$ and $n_2/n$ exceed $1/2$,
this implies that $\hSigma$ is $(\phi_0/4, \supp(\theta_0))$-compatible
with probability exceeding $1-p^{-4}$.
\end{proof}
The following lemma shows that $X^\sT\eps$ is small entrywise.
\begin{lemma}\label{lem:batchlinfbnd}
For any $\lambda_n \ge 40\kappa \sigma \sqrt{(\log p)/n} $, with probability at least $1 - p^{-4}$, 
$\norm{X^\sT\eps}_\infty \le n \lambda_n /2 $.
\end{lemma}
\begin{proof}
The $a^\th$ coordinate of the vector $X^\sT \eps$ is
$\sum_i x_{ia}\eps_i$. As the rows of $X$ are uniformly $\kappa$-subgaussian and $\norm{\eps_i}_{\psi_2} = \sigma$, 
Lemma \ref{lem:subexpprod} implies that the sequence $(x_{ia}\eps_i)_{1\le i \le n}$ is uniformly $2\kappa\sigma$-subexponential. Applying
the Bernstein-type martingale tail bound Lemma \ref{lem:martingalesubexptail}, for $\eps \le 12e \kappa \sigma$:
\begin{align*}
\P\Big\{\Big\lvert\sum_{i} x_{ia}\eps_i \Big\rvert \ge \eps n\Big\}
&\le  2 \exp\Big\{ -\frac{n \eps^2}{24e\kappa ^2\sigma^2} \Big\}
\end{align*}
Set $\eps = \eps(p, n, \kappa, \sigma) = 20\kappa \sigma \sqrt{(\log p )/n}$, 
the exponent on the right hand side above is at least $5\log p$, 
which implies after union bound over $a$ that
\begin{align*}
\P\{\norm{X^\sT \eps}_\infty \ge \eps n \} &=
\P\Big\{\max_a \Big\lvert\sum_{i} x_{ia}\eps_i \Big\rvert \ge \eps n \Big \}\\
&\le \sum_{a} \P\Big\{ \Big\lvert\sum_{i} x_{ia}\eps_i \Big\rvert \ge \eps n    \Big\} \\
&\le 2p^{-6}. 
\end{align*}
This implies the claim for $p$ large enough. 
\end{proof}

The rest of the proof is standard, cf. \cite{hastie2015statistical} and is given below for the reader's convenience. 
\begin{proof}[Proof of Theorem \ref{thm:batchlassoerr}]
Throughout we condition on the intersection of good events in Proposition 
\ref{prop:batchsamplecompatibility} and 
Lemma \ref{lem:batchlinfbnd}, which happens with 
probability at least $1-2p^{-4}$. On this good event,
the sample covariance $\hSigma$ is $(\phi_0/4, \supp(\theta_0))$-compatible
and $\norm{X^\sT \eps}_\infty \le 20 \kappa \sigma \sqrt{n\log p} \le n\lambda_n/2$.

By optimality of $\htheta^\sL$:
\begin{align*}
\frac{1}{2} \norm{y - X\htheta^\sL}^2 + \lambda_n \norm{\htheta^\sL}_1
&\le \frac{1}{2}\norm{y - X\theta_0}^2 + \lambda_n \norm{\theta_0}_1.
\end{align*}
Using $y = X\theta_0 + \eps$, the shorthand
$\nu =\htheta^\sL - \theta_0$ and expanding the squares leads to
\begin{align}
\frac{1}{2} \<\nu, \hSigma\nu\>
&\le \frac{1}{n}\<X^\sT \eps, \nu\>
+ \lambda_n( \norm{\theta_0}_1 - \norm{\htheta^\sL}_1)   \nonumber\\
&\le \frac{1}{n} \norm{\nu}_1 \norm{X^\sT \eps}_\infty +
\lambda_n( \norm{\theta_0}_1 - \norm{\htheta^\sL}_1) \nonumber \\
&\le \lambda_n \Big\{  \frac{1}{2} \norm{\nu}_1 +  \norm{\theta_0}_1
- \norm{\htheta^\sL}_1   \Big\}. \label{eq:batchlassobasicineq}
\end{align}
First we show that the error vector $\nu$ satisfies
$\norm{\nu_{S_0^c}}_1 \le 3 \norm{\nu_{S_0}}_1 $, where $S_0 \equiv\supp(\theta_0)$. Note that $\norm{\htheta^\sL}_1 = \norm{\theta_0 + \nu}_1 = \norm{\theta_0 + \nu_{S_0}}_1 + \norm{\nu_{S_0^c}}_1$. By triangle inequality, therefore:
\begin{align*}
\norm{\theta_0}_1 - \norm{\htheta^\sL }_1
&= \norm{\theta_0}_1 - \norm{\theta_0 + \nu_{S_0}}_1 - \norm{\nu_{S_0^c}}_1 \\
&\le \norm{\nu_{S_0}}_1 - \norm{\nu_{S_0^c}}_1. 
\end{align*}
Combining this with the basic lasso inequality Eq.\eqref{eq:batchlassobasicineq} we obtain
\begin{align*}
\frac{1}{2} \<\nu, \hSigma\nu\> &\le \lambda_n \Big\{ \frac{1}{2}\norm{\nu}_1 + \norm{\nu_{S_0}}_1 - \norm{\nu_{S_0^c}}_1  \Big\} \\
& = \frac{\lambda_n}{2} \Big\{ 3 \norm{\nu_{S_0}} _1 - \norm{\nu_{S_0^c}}.\Big\}
\end{align*}
As $\hSigma$ is positive-semidefinite, the LHS above is non-negative, which implies $\norm{\nu_{S_0^c}}_1 \le 3\norm{\nu_{S_0}}_1$. Now,
we can use the fact that $\hSigma$ is $(\phi_0/4, S_0)$-compatible to
lower bound the LHS by $\norm{\nu}_1^2 \phi_0/2s_0$. This leads
to
\begin{align*}
\frac{ \phi_0 \norm{\nu}_1^2}{2s_0} &\le \frac{3\lambda_n \norm{\nu_{S_0}}_1}{2} \le \frac{3\lambda_n \norm{\nu}_1}{2}. 
\end{align*}
Simplifying this results in $\norm{\nu}_1 = \norm{\htheta^\sL - \theta_0}_1 \le 3 s_0\lambda_n /\phi_0$ as required.

\end{proof}

\subsection{Bias control: Proof of Theorem \ref{thm:batchbiasbound}}\label{proof:thm:batchbiasbound}

 Recall 
the decomposition
\eqref{eq:batchdebiasdecomp} from which we obtain:
\begin{align*}
\Delta_n &=  B_n (\htheta^\sL -\theta_0),\\ 
B_n &= \sqrt{n} \Big( I_p - \frac{n_1}{n} M^{(1)} \hSigma^{(1)}  - \frac{n_2}{n} M^{(2)} \hSigma^{(2)} \Big) , \\
W_n &= \frac{1}{\sqrt n} \sum_{i\le n_1} M^{(1)} x_i \eps_i
+ \frac{1}{\sqrt n} \sum_{n_1 < i \le n} M^{(2)} x_i \eps_i.
\end{align*}
By construction $M^{(1)}$ is a function of $X_1$ and hence is independent of $\eps_1, \dots, \eps_{n_1}$. In addition, $M^{(2)}$
is independent of $\eps_{n_1+1}, \dots, \eps_n$. Therefore
$\E\{W_n\} = 0$ as required. The key is to show
the bound on $\norm{\Delta_n}_\infty$. 
We start by using H\"older inequality
\begin{align*}
 \norm{\Delta_n}_\infty &\le \norm{B_n}_\infty \norm{\htheta^\sL - \theta_0}_1. 
 \end{align*} 
Since the $\ell_1$ error of $\htheta^\sL$ is bounded
in Theorem \ref{thm:batchlassoerr}, we need only to show
the bound on $B_n$. For this, we use triangle inequality and that
$M^{(1)} $ and $M^{(2)} $ are feasible for the online debiasing program:
\begin{align*}
\norm{B_n}_\infty &= \sqrt{n} \Big\lVert \frac{n_1}{n} (I_p - M^{(1)} \hSigma^{(1)}) + \frac{n_2}{n} (I_p-M^{(2)} \hSigma^{(2)}  )  \Big\rVert_\infty  \\
&\le \sqrt{n}\Big(\frac{n_1}{n} \norm{I_p - M^{(1)} \hSigma^{(1)} }_\infty
+ \frac{n_2}{n} \norm{I_p - M^{(2)} \hSigma{(2)}  }_\infty\Big)  \\
&\le \sqrt{n} \Big(\frac{n_1\mu_1}{n} + \frac{n_2\mu_2}{n}  \Big). 
\end{align*}
The following lemma shows that, with high probability, 
we can take $\mu_1$, $\mu_2$ so that the resulting bound
on $B_n$ is of
order $\sqrt{\log p}$. 
\begin{lemma}\label{lem:batchBnbound}
Denote by $\Omega = (\E\{xx^\sT\})^{-1}$ and
 $\Omega^{(2)} (\hth) = (\E\{x x^\sT \vert \<x, \hth\> \ge \varsigma\})^{-1}$
 be the population precision matrices for the first and second 
 batches. Suppose that $n_1\wedge n_2 \ge 2 \lambdalbd/\kappa^2 \log p$. Then,
with probability at least $1-p^{-4}$ 
\begin{align*}
\norm{I_p - \Omega \hSigma^{(1)} }_\infty &\le  15 \kappa \lambdalbd^{-1/2} \sqrt{\frac{\log p}{n_1}}, \\
\norm{I_p - \Omega^{(2)} \hSigma^{(2)} }_\infty &\le  15 \kappa \lambdalbd^{-1/2}  \sqrt{\frac{\log p}{n_2}}.
\end{align*}
In particular, with the same probability, the online
debiasing program \eqref{eq:batchonlinedebias} is feasible
with $\mu_\ell = 15\kappa^2\lambdalbd^{-1}\sqrt{(\log p)/n_\ell} < 1/2$.
\end{lemma}
It follows from the lemma, Theorem \ref{thm:batchlassoerr} and the previous display that, with probability at least
$1 - 2p^{-3}$
\begin{align}
\norm{\Delta_n}_\infty &\le \norm{B_n}_\infty
\norm{\htheta^\sL - \theta_0}_1 \nonumber \\
&\le 15\kappa \Lambda_0^{-1/2}\sqrt{n} \Big(\frac{n_1}{n} \sqrt{\frac{\log p}{n_1}} + 
\frac{n_2}{n} \sqrt{\frac{\log p}{n_2}}   \Big)\cdot
120\kappa \sigma \phi_0^{-1} s_0 \sqrt{\frac{\log p}{n}}, \nonumber \\
&\le 2000 \frac{\kappa^2 \sigma}{\sqrt{\lambdalbd} \phi_0} \frac{s_0 \log p}{n}
(\sqrt{n_1} + \sqrt{n_2}) \nonumber \\
&\le 4000 \frac{\kappa^2\sigma}{\sqrt{\lambdalbd} \phi_0} \frac{s_0 \log p}{\sqrt{n}}. \label{eq:batchDeltanprobbnd}
\end{align}
This implies the first claim that, with probability rapidly converging to one, 
$\Delta_n/\sqrt{n}$ is of order $s_0 \log p/n$.

We should also expect $\norm{\E\{\onth - \theta_0\}}_\infty$ 
to be of the same order. To prove this, however, we need
some control (if only rough) on $\onth$ in the exceptional
case when the LASSO error is large or the online debiasing program is
infeasible.  
Let $G_1$ denote the good event of Lemma \ref{lem:batchlinfbnd}
and $G_2$ denote the good event of Theorem \ref{thm:batchlassoerr}
as below:
\begin{align*}
G_1 &= \bigg\{\text{ For } \ell=1, 2:   \norm{I_p - \Omega^{(\ell)} \hSigma^{(\ell)}}_\infty \le  15 \kappa \lambdalbd^{-1/2}\sqrt{\frac{\log p}{n_\ell}} \bigg\},\\
G_2 &= \Big\{ 
\norm{\htheta^\sL - \theta_0}_1 \le  \frac{3s_0\lambda_n}{\phi_0} 
=  \frac{120 \kappa \sigma }{\phi_0}  s_0 \sqrt{\frac{\log p}{n}}.
  \Big\}. 
\end{align*}
On the intersection $G = G_1 \cap G_2$,  $\Delta_n$ satisfies the bound \eqref{eq:batchDeltanprobbnd}. For the complement:
we will use the following rough bound on the LASSO error:

Now, since $W_n$ is unbiased:
\begin{align*}
\norm{\E\{\onth - \theta_0\}}_\infty
&=\Big\lVert \frac{\E\{\Delta_n \}}{\sqrt{n}} \Big\rVert_\infty \\
& =  \Big\lVert \frac{\E\{\Delta_n \ind(G) \}}{\sqrt{n}} \Big\rVert_\infty 
+ 
 \Big\lVert \frac{\E\{\Delta_n \ind(G^c) \}}{\sqrt{n}} \Big\rVert_\infty \\
 &\le 4000 \frac{\kappa^2\sigma}{\sqrt{\lambdalbd}\phi_0} \frac{s_0 \log p}{n}
 + \E\{\norm{\htheta^\sL - \theta_0}_1 \ind(G^c) \}. 
\end{align*}
For the second term, we can use Lemma \ref{lem:sizeofLASSO}, Cauchy Schwarz and that $\P\{G^c\} \le 4p^{-3}$ to obtain:
\begin{align*}
\E\{\norm{\htheta^\sL - \theta_0}_1 \ind(G^c) \}
&\le \E\Big\{  \frac{\norm{\eps}^2 \ind(G^c)}{2n\lambda_n } + 2\|\theta_0\|_1  \ind(G^c) \Big\} \\
&\le \frac{\E\{\norm{\eps}^4\}^{1/2} \P(G^c)^{1/2} }{2n\lambda_n} + 2\|\theta_0\|_1 \P\{G^c\} \\
&\le \frac{\sqrt{3}\sigma^2}{\sqrt{n}p^{1.5}\lambda_n }  + 8\|\theta_0\|_1 p^{-3} \le 10c \frac{s_0 \log p}{n},
\end{align*}
for $n, p$ large enough . This implies the claim on the bias.

It remains only to prove the intermediate Lemma 
\ref{lem:batchBnbound}.

\begin{proof}[Proof of Lemma \ref{lem:batchBnbound}]
We prove the claim for the second batch, and in the rest of the proof, 
we assume that all probabilities and expectations are conditional
on the first batch (in particular, the intermediate estimate $\htheta^1$). The $(a, b)$ entry of $I_p - \Omega^{(2)} \hSigma^{(2)} $ reads
\begin{align*}
(I_{p} - \Omega^{(2)} \hSigma^{(2)} )_{a, b} &=  \ind(a = b) - 
\<\Omega^{(2)} e_a, \hSigma^{(2)} e_b\> \\
&= \frac{1}{n_2}\sum_{i> n_1} \ind(a=b) -  \<e_a , \Omega^{(2)} x_i\>x_{ib}. 
\end{align*}
Now, $\E\{\<e_a, \Omega^{(2)} x_i\> x_{i, b}\> \} = \ind(a = b)$ and
$\<e_a, \Omega^{(2)} x_i\>$ is $(\norm{\Omega^{(2)} }_2 \kappa)$-subgaussian.
Since $\Sigma^{(2)} \mge \lambdalbd I_p $, we have that $\norm{\Omega^{(2)} }_2 \le \lambdalbd^{-1}$. This observation, coupled with Lemma \ref{lem:subexpprod}, yields
$\<e_a, \Omega^{(2)} x_i\> x_{i, b}$ is $2\kappa^2/\lambdalbd$-subexponential.
Then we may apply Lemma \ref{lem:subexptail} for
$\eps \le 12 \kappa^2/\lambdalbd$ as below:
\begin{align*}
 \P\{(I_p - \Omega^{(2)} \hSigma^{(2)} )_{a, b} \ge  \eps\} 
 &\le  \exp\Big(-\frac{n_2\eps^2 }{36 \kappa^2 \lambdalbd^{-1}} \Big). 
 \end{align*} 
 Keeping $\eps = \eps(p, n_2, \kappa, \lambdalbd) = 15 \kappa \lambdalbd^{-1/2} \sqrt{(\log p)/n_2}$ we obtain:
 \begin{align*}
  \P\Big\{ (I_p - \Omega^{(2)} \hSigma^{(2)} )_{a, b} \ge 15\kappa \lambdalbd^{-1/2}\sqrt{\frac{\log p}{n_2}} \Big\} 
  &\le p^{-6}. 
  \end{align*} 
  Union bounding over the pairs $(a, b)$ yields the claim. 
  The requirement $n_2 \ge 2 (\lambdalbd/\kappa^2) \log p$
  ensures that the choice $\eps$ above satisfies 
  $\eps \le 12 \kappa^2 /\lambdalbd$.

\end{proof}

\subsection{Central limit asymptotics: proofs of Proposition \ref{prop:batchvarianceclt}
and Theorem \ref{thm:batchdistchar}}\label{proof:prop:batchvarianceclt}

Our approach is to apply a martingale central limit theorem
to show that  $W_{n, a}$ is approximately normal. 
An important first step is to show that the conditional
covariance $V_{n, a}$ is stable, or approximately
constant. Recall that $V_{n, a}$ is defined as
\begin{align*}
V_{n, a}&= \sigma^2 \Big(\frac{n_1}{n} \<m^{(1)}_a, \hSigma^{(1)} m^{(1)} _a\> + 
\frac{n_2}{n} \<m^{(2)} _a, \hSigma^{(2)} m^{(2)} _a\>\Big). 
\end{align*}
We define its deterministic equivalent as
follows. Consider the function $f:\psdcone^n\to\reals$ by:
\begin{align*}
f(\Sigma) &= \{\min \;\<m, \Sigma m\> : \norm{\Sigma m - e_a}_\infty \le \mu\,,\;\; \norm{m}_1 \le L\}. 
\end{align*}

We begin with two lemmas about
the 
stability of the optimization program used
to obtain the online debiasing matrices. 
\begin{lemma}\label{lem:stabilityJMprogram}
On its domain (and uniformly in $\mu, e_a$),  
$f$ is $L^2$-Lipschitz with respect to the $\norm{\cdot}_\infty$
norm.
 \end{lemma}
\begin{proof}
 For two matrices $\Sigma, \Sigma'$ in
the domain, let $m, m'$ be the respective optimizers (which 
exist by compactness of the set $\{m: \norm{\Sigma m - v}_\infty \le \mu, \norm{m}_1 \le L\}.$ We prove that $\abs{f(\Sigma) - f(\Sigma') }\le L^2 \norm{\Sigma - \Sigma'}_\infty$.
\begin{align*}
f(\Sigma) - f(\Sigma') &= \<\Sigma, mm^\sT\> - \<\Sigma', m'(m')^\sT\>\\
&\le \<\Sigma, m'(m')^\sT\> - \<\Sigma', m'(m')^\sT\> \\
&= \<(\Sigma - \Sigma') m', m'\> \\
&\le \norm{(\Sigma - \Sigma')m'}_\infty \norm{m'}_1 \\
&\le \norm{\Sigma - \Sigma'}_\infty \norm{m'}_1^2 \le L^2\norm{\Sigma - \Sigma'}_\infty. 
\end{align*}
Here the first inequality follows from optimality of $m$ and
the last two inequalities are H\"older inequality. 
The reverse inequality $f(\Sigma) - f(\Sigma') \ge -L^2 \norm{\Sigma - \Sigma'}_\infty$ is proved in the same way.
\end{proof}

\begin{lemma}\label{lem:boundedJMprogram}
We have the following lower bound on the optimization
value reached to compute $f(\Sigma)$:
\begin{align*}
\frac{(1-\mu)^2}{\lambdamax(\Sigma)} &\le f(\Sigma) \le \frac{1}{\lambdamin(\Sigma)}.
\end{align*}
\end{lemma}
\begin{proof}
We first prove the lower bound for $f(\Sigma)$.
Suppose $m$ is an optimizer for the program. Then 
\begin{align*}
\norm{\Sigma m}_2 \ge \norm{\Sigma m}_\infty &\ge \norm{e_a}_\infty - \mu = 1-\mu.
\end{align*}
On the other hand, the value is given by 
\begin{align*}
\<m, \Sigma m\> = \<\Sigma m, \Sigma^{-1} (\Sigma m)\> \ge \lambdamin(\Sigma^{-1})\norm{\Sigma m}_2^2 = \norm{\Sigma m}_2^2\; \lambdamax(\Sigma)^{-1}.
\end{align*}
Combining these gives the lower bound.

For the upper bound, it suffices to consider any feasible point;
we choose $m = \Sigma^{-1}e_a$, which is feasible since $\norm{\Sigma^{-1}}_1 \le L$. The value is then $\<e_a, \Sigma^{-1} e_a \>\le \lambdamax(\Sigma^{-1})$ which gives the upper bound. 
\end{proof}

\begin{lemma}(Stability of $W_{n, a}$)\label{lem:batchstability}
Define $\Sigma^{(2)} (\theta) = \E\{x x^\sT | \<x_1, \theta\>\ge \varsigma\} $. 
Then,
under Assumptions \ref{assmp:batchCov} and \ref{assmp:batchasymp}
 \begin{align*}
\lim_{n\to \infty}\Big\lvert V_{n, a} - \sigma^2\Big(\frac{n_1 f(\Sigma)}{n}
+ \frac{n_2 f(\Sigma^2(\theta_0))}{n}\Big) \Big\rvert   &=  0, \quad \text{ in probability.}
\end{align*}
\end{lemma}
\begin{proof}
Using Lemma \ref{lem:stabilityJMprogram}:
\begin{align*}
&\Big\lvert V_{n, a} - \sigma^2\Big(\frac{n_1}{n} f(\Sigma) + \frac{n_2}{n}f(\Sigma(\theta_0)\Big)\Big\rvert\\
&=  \frac{\sigma^2 n_1}{n} (f(\hSigma^{(1)}) - f(\Sigma))
+ \frac{\sigma^2 n_2}{n} (f(\hSigma^{(2)} - f(\Sigma(\theta_0)))) \\
&\le L^2 \frac{\sigma^2 n_1}{n}\norm{\Sigma - \hSigma^{(1)} }_{\infty}
+ L^2\frac{\sigma^2 n_2}{n} \norm{\Sigma^{(2)} (\theta_0) - \hSigma^{(2)} }_\infty \\
&\le L^2\frac{\sigma^2 n_1}{n} \norm{\Sigma - \hSigma^{(1)} }_\infty
+ L^2 \frac{\sigma^2 n_2}{n} \big( \norm{\Sigma^{(2)} (\theta_0) - \Sigma^{(2)} (\htheta^1)}_\infty + \norm{\Sigma^{(2)} (\htheta^1) - \hSigma^{(2)} }_\infty\big) \\
&\le \sigma^2 L^2 \norm{\Sigma - \hSigma^{(1)} }_\infty
+ \sigma^2 L^2 \big(K\norm{\htheta^1 - \theta_0}_1 + \norm{\Sigma^{(2)} (\hth^1) - \hSigma^{(2)} }_\infty\big)\,.
\end{align*}
Using Lemma  \ref{lem:batchentrywiseSampCovBnd} the first and
third term vanish in probability. It is straightforward
to apply Theorem \ref{thm:batchlassoerr} to the intermediate
estimate $\hth^1$; indeed Assumption \ref{assmp:batchasymp} 
guarantees that $n_1 \ge c n$ for a universal $c$. 
Therefore the intermediate estimate has an error
$\norm{\htheta^1-\theta_0}_1 $ of order $\kappa \sigma \phi_0^{-1} \sqrt{(s_0^2\log p)/n}$ with probability converging to one. In particular, 
the second term is, 
with probability converging to one, of order $KL^2\sigma^3 \kappa\phi^{-1}_0 \sqrt{s_0^2(\log p)/n} = o(1)$ by Assumption \ref{assmp:batchasymp}. 
\end{proof}

\begin{lemma}\label{lem:batchlindeberg}
Under Assumptions \ref{assmp:batchCov} and \ref{assmp:batchasymp},
with probability at least $1-p^{-2}$
\begin{align*}
\max_{i} \abs{\<m_a, x_i\>} &\le 10L\kappa\sqrt{\log p},
\end{align*}
In particular $\lim_{n\to\infty}\max_i \, \abs{\<m_a, x_i\>} =0$
in probability.
\end{lemma}
\begin{proof}
By H\"older inequality, $\max_i \<\abs{\<m_a, x_i\>} \le \max_i\, \norm{m_a}_1 \norm{x_i}_\infty \le L \max_i\, \norm{x_i}_\infty$. 
Therefore, it suffices to prove that, with the required probability
$\max_{i, a} \abs{x_{i, a}}\le  10\kappa \sqrt{\log p}$.
Let $u = 10\kappa\sqrt{\log p} $. 
Since $x_i$ are uniformly $\kappa$-subgaussian, we obtain for
 $q > 0$:
\begin{align*}
\P\{  \abs{x_{i, a}} \ge u  \} &\le u^{-q}\E\{\abs{x_{i, a}}^q\} \le (\sqrt{q} \kappa/u)^q \\
&= \exp\Big(-\frac{q}{2}\log \frac{u^2}{\kappa^2 q}\Big)
\le \exp\Big(- \frac{u^2}{2\kappa^2} \Big) \le p^{-5}\,,
\end{align*}
where the last line follows by choosing $q = u^2 /e\kappa^2$.
By union bound over $i\in[n], a\in[p]$, we obtain:
\begin{align*}
\P\{\max_{i, a} \abs{x_{i, a}} \ge u\} &\le \sum_{i,a} \P\{\abs{x_{i, a}}\ge u\} \le p^{-3},
\end{align*}
which implies the claim (note that $p\ge n$ as we are focusing on the high-dimensional regime).
\end{proof}

With these in hand we can prove Proposition \ref{prop:batchvarianceclt} 
and Theorem \ref{thm:batchdistchar}.

\begin{proof}[Proof of Proposition \ref{prop:batchvarianceclt}]
Consider the minimal  filtration $\fF_i$  so that
\begin{enumerate}
\item For $i < n_1$, $y_1, \dots, y_i$, $x_1, \dots x_{n_1}$
and $\eps_1, \dots, \eps_i$ are measurable with respect to $\fF_i$.
\item For $i \ge n_1$ $y_1, \dots, y_i$, $x_1, \dots, x_n$
and $\eps_1, \dots \eps_i$ are measurable with respect to $\fF_i$.
\end{enumerate}
The martingale $W_{n}$ (and therefore, its $a^\th$ coordinate
$W_{n, a}$) is adapted to the filtration $\fF_i$. We can now
apply the martingale central limit theorem \cite[Corollary~3.1]{hall2014martingale} to $W_{n, a}$ to obtain the result. From
Lemmas \ref{lem:boundedJMprogram} and \ref{lem:batchstability} we
know that $V_{n, a}$ is bounded away from $0$, asymptotically. 
The stability and conditional Lindeberg conditions of \cite[Corollary~3.1]{hall2014martingale}
are verified by
Lemmas \ref{lem:batchstability} and \ref{lem:batchlindeberg}.
\end{proof}

\begin{proof}[Proof of Theorem \ref{thm:batchdistchar}]
This is a straightforward corollary of the bias bound of
\ref{thm:batchbiasbound} and Proposition \ref{prop:batchvarianceclt}. 
 We will show that:
\begin{align*}
\lim_{n\to \infty} \P\Big\{ \sqrt\frac{n}{V_{n, a}}(\onth_a - \theta_{0, a}) \le x   \Big\} &\le \Phi(x).
\end{align*}
The reverse inequality follows using the same argument.

Fix a $\delta > 0$. We decompose the difference above as:
\begin{align*}
\sqrt\frac{n}{V_{n, a}} (\onth_a - \theta_{0, a}) &=
\frac{W_{n, a}}{\sqrt{V_{n, a}}} + \frac{\Delta_{n, a}}{\sqrt{V_{n, a}}}\,.
\end{align*}
Therefore, 
\begin{align*}
 \P\Big\{ \sqrt\frac{n}{V_{n, a}}(\onth_a - \theta_{0, a}) \le x   \Big\}
 &\le \P\Big\{ \frac{W_{n, a}}{\sqrt{V_{n, a}}} \le x+\delta\Big\}
+ \P\{ \abs{\Delta_{n, a}} \ge \sqrt{V_{n, a}} \delta\}.
\end{align*}
By Proposition \ref{prop:batchvarianceclt} the first term
converges to $\Phi(x+ \delta)$. To see that the second
term vanishes, observe first that  Lemma \ref{lem:boundedJMprogram} and Lemma \ref{lem:batchstability}, imply that  $V_{n, a}$
is bounded away from $0$ in probability. Using this: 
\begin{align*}
\lim_{n\to\infty} \P\{\abs{\Delta_{n, a}} \ge \sqrt{V_{n, a}}\delta\}
&\le \lim_{n\to \infty} \P\{ \norm{\Delta_n}_\infty \ge \sqrt{V_{n, a}}\delta\} \\
&\le \lim_{n\to\infty} \P \Big\{ \norm{\Delta_n}_\infty \ge 4000
\frac{\kappa^2\sigma}{\sqrt{\lambdalbd}\phi_0} \frac{s_0\log p}{\sqrt{n}}   \Big\}  = 0
\end{align*}
by applying Theorem \ref{thm:batchbiasbound} and that 
for $n$ large enough, $  \sqrt{V_{n, a}}\delta$ exceeds
the bound on $\norm{\Delta_n}_\infty$ used. 
Since $\delta$ is arbitrary, the claim follows. 
\end{proof}

\subsection{Proofs for Gaussian designs}\label{sec:examples}

In this Section we prove 
that Gaussian designs of Example \ref{ex:batchgaussianrequirements} 
satisfy the requirements of Theorem \ref{thm:batchlassoerr}
and Theorem \ref{thm:batchbiasbound}. 

The following distributional identity will be important.
\begin{lemma}\label{lem:xconditionaldistribution}
Consider the parametrization $\varsigma = \bvsigma{\<\hth, \Sigma\hth\>}^{1/2}$.
Then 
\begin{align*}
x\vert_{\<x, \hth\> \ge \varsigma} \stackrel{{\rm d}}{=} \frac{\Sigma\hth}{\<\hth, \Sigma\hth\>^{1/2}}  \xi_1 + \Big(\Sigma - \frac{\Sigma\hth\hth^\sT\Sigma}{\<\hth, \Sigma\hth\>}\Big)^{1/2} \xi_2,
\end{align*}
where $\xi_1, \xi_2$ are independent, 
$\xi_2\sim\normal(0, I_p)$ and $\xi_1$ 
has the density:
\begin{align*}
\frac{\d\P_{\xi_1}}{\d u}(u)&= \frac{1}{\sqrt{2\pi}\Phi(-\bvsigma)}\exp(-u^2/2)\ind(u \ge \bvsigma).
\end{align*}
\end{lemma}
\begin{proof}
This follows from the distribution of $x|\<x, \hth\>$ 
being $\normal(\mu', \Sigma')$ with
\begin{align*}
\mu' = \frac{\Sigma\hth}{\<\hth, \Sigma\hth\>}\<x, \hth\>, \quad
\Sigma' = \Sigma - \frac{\Sigma\hth\hth^\sT\Sigma}{\<\hth, \Sigma\hth\>}.
\end{align*}
\end{proof}

The following lemma shows that they satisfy 
compatibility.
\begin{lemma}\label{lem:gaussiancompatibility}
Let $\P_x= \normal(0, \Sigma)$ 
for a positive definite covariance $\Sigma$. 
Then, for any vector $\htheta$ and subset $S\subseteq[p]$, the second moments
$\E\{x x^\sT\}$ and 
$\E\{x x^\sT | \<x, \htheta\> \ge \varsigma\}$ are $(\phi_0, S)$-compatible
with $\phi_0 = \lambda_{\min}(\Sigma)/16$. 
\end{lemma}
\begin{proof}
Fix an $S\subseteq[p]$. We prove that $\Sigma = \E\{x_1x_1^\sT\}$ is 
$(\phi_0, S)$-compatible with $\phi_0 = \lambda_{\min}(\Sigma)/16$. Note that, for any $v$ satisfying $\norm{v_{S^c}}_1 \le 3 \norm{v_S}$, its
$\ell_1$ norm satisfies
$\norm{v}_1 \le 4\norm{v_S}_1 $. Further  $\Sigma \mge \lambda_{\min}(\Sigma) I_p$ implies: 
\begin{align*}
\frac{\abs{S}\<v, \Sigma v\>}{\norm{v}_1^2} &\ge \lambda_{\min}(\Sigma) \frac{\abs{S}\norm{v}^2}{\norm{v}_1^2} \ge \lambda_{\min}(\Sigma) \frac{\abs{S} \norm{v_S}^2}{16\norm{v_S}_1^2}\ge \frac{ \lambda_{\min} (\Sigma)}{16}.
\end{align*}

For $\E\{x x^\sT \vert \<x, \htheta\> \ge \varsigma\}$, 
we use Lemma \ref{lem:xconditionaldistribution} to obtain
\begin{align*}
\E\{x x^\sT \vert \<x, \hth\> \ge \varsigma\}
&= \Sigma + (\E\{\xi_1^2\} - 1) \frac{\Sigma\hth\hth^\sT \Sigma}{\<\hth,\Sigma\hth\>},
\end{align*}
where $\xi_1$ is as in Lemma \ref{lem:xconditionaldistribution}. 
Since $\E\{\xi_1^2\} = 1 +\bvsigma\varphi(\bvsigma)/\Phi(-\bvsigma) \ge 1 + \bvsigma^2 $ whenever
$\bvsigma \ge 0$:
\begin{align*}
\E\{xx^\sT \vert \<x, \hth\> \ge \varsigma\}
&\ge \Sigma + \bvsigma^2 \frac{\Sigma\hth\hth^\sT\Sigma}{\<\hth, \Sigma\hth\>} \mge \lambdamin(\Sigma) I_p\,.
\end{align*}
 The rest of the proof is as for $\Sigma$.

\end{proof}


\begin{lemma}\label{lem:subG}
Let $\P_x= \normal(0, \Sigma)$ 
for a positive definite covariance $\Sigma$. 
Then, for any vector $\htheta$ and subset $S\subseteq[p]$, 
the random vectors $x$ and $x \vert_{\<x, \hth\> \ge \varsigma}$
are $\kappa$-subgaussian with $\kappa = 3\lambdamax(\Sigma)^{1/2} (\bvsigma \vee \bvsigma^{-1})$, where
$\bvsigma = \varsigma/\<\hth, \Sigma\hth\>^{1/2}$.

\end{lemma}
\begin{proof}
By definition, $\<x, v\> \sim\normal(0, v^\sT\Sigma v)$ is $\sqrt{v^\sT\Sigma v}$-subGaussian. Optimizing over all unit vectors $v$, $x$ is
$\lambda_{\max}^{1/2}(\Sigma)$-subgaussian. 

For $x\vert_{\<x, \hth\> \ge \varsigma}$, we use the decomposition
of Lemma \ref{lem:xconditionaldistribution}:
\begin{align*}
x\vert_{\<x, \hth\> \ge\varsigma} &\stackrel{d}{=}\frac{\Sigma\hth}{\<\hth, \Sigma\hth\>^{1/2}}  \xi_1 + \Big(\Sigma - \frac{\Sigma\hth\hth^\sT\Sigma}{\<\hth, \Sigma\hth\>}\Big)^{1/2} \xi_2.
\end{align*}
Clearly, $\xi_2$ is 1-subgaussian, which means the second
term is $\lambda_{\max}^{1/2}(\Sigma)$-subgaussian. For the
first term, we claim that $\xi_1$ is $1$-subgaussian and 
therefore the first term is $\lambda_{\max}^{1/2}(\Sigma)$-subgaussian.
To show this, we start with the moment generating function
of $\xi_1$. Recall that $\bvsigma = \varsigma/\<\hth, \Sigma\hth\>^{1/2}$:
\begin{align*}
 \E\{e^{\lambda\xi_1}\} &= \int_{\bvsigma}^\infty e^{\lambda u} e^{-u^2/2}\frac{\d u}{\sqrt{2\pi} \Phi(-\bvsigma)} = e^{\lambda^2/2} \frac{\Phi(\lambda - \bvsigma)}{\Phi(-\bvsigma)}. 
 \end{align*} 
 Here $\varphi$ and $\Phi$ are the density and c.d.f. of the
 standard normal distribution.
 It follows that:
 \begin{align*}
 \frac{\d^2}{\d \lambda^2}\log \E\{ e^{\lambda \xi_1}\} &= 
 \frac{1}{2} + \frac{(\lambda - \bvsigma)\varphi(\lambda - \bvsigma)}{\Phi(\lambda - \bvsigma)} - \frac{\varphi(\lambda-\bvsigma)^2}{\Phi(\lambda - \bvsigma)^2} \\
 &\le\frac{1}{2} + \sup_{\lambda \ge \bvsigma} \frac{(\lambda - \bvsigma)\varphi(\lambda - \bvsigma)}{\Phi(\lambda - \bvsigma)} \\
 &\le \frac{1}{2} + \sup_{\lambda \ge 0} \frac{\lambda \varphi(\lambda)}{\Phi(\lambda)} <1\,.
 \end{align*}
 Now, consider the centered version $\xi_1' = \xi_1 -\E\{\xi_1\}$. The above bound also holds for $\d^2/\d\lambda^2 (\log \E\{ e^{\lambda\xi_1'} \} )$. Therefore, 
 by integration, $\d\log\E\{e^{\lambda \xi'_1}\}/\d\lambda \le \lambda + C$, 
 for some constant $C$ independent of $\lambda$. Now
\begin{align*}
\frac{\d\log\E\{e^{\lambda\xi_1'}\}}{\d\lambda} \Big\vert_{\lambda = 0} &= \E\{\xi_1'\} = 0.
\end{align*}
Therefore, we can take the constant $C$ to be 0. 
Repeating this integration argument, we obtain $\log \E\{e^{\lambda\xi_1'}\} \le \lambda^2/2$, which
 implies that $\xi_1' = \xi_1-\E\{\xi_1\}$ is 1-subgaussian.

 It follows, by triangle inequality, that $\xi_1$ is $(1+\E\{\xi_1\})$-subgaussian. It only remains to bound $\E\{\xi_1\}$ as below:
\begin{align*}
\E\{\xi_1\} & = \frac{\varphi(\bvsigma)}{\Phi(-\bvsigma)}\le \frac{1+\bvsigma^2}{\bvsigma} \le 2(\bvsigma \vee \bvsigma^{-1}).
\end{align*}
Therefore, the subgaussian constant of $x\vert_{\<x, \hth\> \ge \varsigma}$
is at most $\lambdamax(\Sigma)^{1/2}(2\bvsigma \vee \bvsigma^{-1} +1) \le 3\lambdamax(\Sigma)^{1/2} (\bvsigma \vee \bvsigma^{-1})$. 

\end{proof}

For Example \ref{ex:batchgaussianrequirements}, it remains
only to show the constraint on the approximate
sparsity of
the inverse covariance. We show this in the following
\begin{lemma}
Let $\P_x = \normal(0, \Sigma)$ and $\htheta$ be 
any vector such that $\norm{\htheta}_1\norm{\htheta}_\infty \le L  \lambda_{\min}(\Sigma)  \norm{\htheta}^2 /2 $
and $\norm{\Sigma^{-1}}_1 \le L/2$. Then, 
with $\Omega = \E\{x x^\sT\}^{-1}$
and $\Omega^{(2)}(\hth) = \E\{x x^\sT | \<x, \htheta\> \ge \varsigma\}^{-1}$:
\begin{align*}
\norm{\Omega}_1 \vee \norm{\Omega^{(2)} }_1 &\le L. 
\end{align*}
\end{lemma}
\begin{proof}

By assumption $\norm{\Omega}_1 \le L/2$, so we only require
to prove the claim for 
$\Omega^{(2)} =  \E\{xx^\sT | \<x, \htheta\> \ge \varsigma\}^{-1}$.
Using Lemma \ref{lem:xconditionaldistribution}, we can 
compute the precision matrix:
\begin{align*}
 \Omega^{(2)} &= \E\{xx^\sT|\<x, \hth\>\ge \varsigma\}^{-1}\\
 &= \Big(\Sigma + (\E\{\xi_1^2\} - 1) \frac{\Sigma\hth\hth^\sT\Sigma}{\<\hth, \Sigma\hth\>}\Big)^{-1} \\
 &= \Omega +  (\E\{\xi_1^2\}^{-1} - 1) \frac{\hth\hth^\sT}{\<\hth, \Sigma\hth\>}\,,
 \end{align*} 
 where the last step follows by an application of Sherman–Morrison formula. 
 Since $\E\{\xi_1^2\} = 1 + \bvsigma\varphi(\bvsigma)/\Phi(-\bvsigma)$, where
 $\bvsigma = \varsigma/\<\hth, \Sigma\hth\>^{1/2}$ this yields:
 \begin{align*}
 \Omega^{(2)} &= \Omega - \frac{\bvsigma\varphi(\bvsigma)}{\Phi(-\bvsigma) + \bvsigma\varphi(\bvsigma)}
 \frac{\hth\hth^\sT}{\<\hth, \Sigma\hth\>}. 
 \end{align*}
 By triangle inequality, for any $\bvsigma\ge 0$:
 \begin{align*}
 \norm{\Omega^{(2)}}_1 &\le \norm{\Omega}_1 +\frac{\norm{\hth\hth^\sT}_1}{\<\hth, \Sigma\hth\>} \\
 &\le \frac{L}{2} + \frac{\norm{\hth}_1 \norm{\hth}_\infty}{\lambdamin(\Sigma)\norm{\hth}^2}  \le L. 
 \end{align*}
\end{proof}

Next we show that the conditional
covariance of $x$ is appropriately Lipschitz.

\begin{lemma}
Suppose $\varsigma = \bvsigma\<\theta, \Sigma\theta\>^{1/2}$ for
a constant $\bvsigma\ge 0$. Then
The conditional
covariance function 
$\Sigma^{(2)}(\theta) = \E\{xx^\sT | \<x, \theta\> \ge \varsigma\}$ satisfies:
\begin{align*}
\norm{\Sigma^{(2)}(\theta') - \Sigma^{(2)}(\theta) }_\infty &\le K\norm{\theta' - \theta},
\end{align*}
where $K = \sqrt{8}(1+\bvsigma^2)\lambdamax(\Sigma)^{3/2}/\lambdamin(\Sigma)^{1/2}$.
\end{lemma}
\begin{proof}
Using Lemma \ref{lem:xconditionaldistribution}, 
\begin{align*}
\Sigma^{(2)}(\theta) &= \Sigma 
+ (\E\{\xi_1^2\} - 1)\frac{\Sigma \theta\theta^\sT \Sigma}{\<\theta, \Sigma\theta\>}.
\end{align*}
Let $v = \Sigma^{1/2}\theta/\|\Sigma^{1/2}\theta\|$ and $v'=\Sigma^{1/2}\theta'/\|\Sigma^{1/2}\theta'\|$. With this,
\begin{align*}
\norm{\Sigma^{(2)}(\theta')-
\Sigma^{(2)}(\theta)}_\infty
&= {(\E\{\xi_1^2\}- 1) } \norm{\Sigma^{1/2}(vv^\sT - v'v'^\sT)\Sigma^{1/2}}_\infty \\
&\le {(\E\{\xi_1^2\}- 1) }\, \lambdamax(\Sigma)
\norm{vv^\sT - v'v'^\sT}_2 \\
&\le {(\E\{\xi_1^2\}- 1) } \lambdamax(\Sigma)
\norm{vv^\sT - v'v'^\sT}_F \\
&\stackrel{(a)}{\le} \sqrt{2} {(\E\{\xi_1^2\}- 1) } \lambdamax(\Sigma)
\norm{v - v'} \\
&\stackrel{(b)}{\le} \frac{\sqrt{8} \lambdamax(\Sigma)^{3/2} } {\lambdamin(\Sigma)^{1/2}} (\E\{\xi_1^2\} - 1) \norm{\theta - \theta'}\\
&\stackrel{(c)}{\le}\frac{\sqrt{8} \lambdamax(\Sigma)^{3/2} } {\lambdamin(\Sigma)^{1/2}} (\bvsigma^2+1)  \norm{\theta - \theta'}\,.
\end{align*}
Here, $(a)$ follows by noting that for two unit vectors $v$, $v'$, we have
\begin{align*}
\|vv^\sT - v'v'^\sT\|_F^2 = 2 - 2(v^\sT v')^2 = 2(1-v^\sT v') (1+v^\sT v')
\le 2\|v - v'\|^2\,.
\end{align*}
Also, $(b)$ holds using the following chain of triangle inequalities  
\begin{align*}
\norm{v - v'} &=  \Big\|\frac{\Sigma^{1/2}\theta}{\norm{\Sigma^{1/2}\theta}} - \frac{\Sigma^{1/2}\theta'}{\norm{\Sigma^{1/2}\theta'}}\Big\|\\
&\le \frac{\norm{\Sigma^{1/2}(\theta - \theta')}}{\norm{\Sigma^{1/2} \theta}}
+ \norm{\Sigma^{1/2}\theta'}  \Big|\frac{1}{\norm{\Sigma^{1/2}\theta }} 
-\frac{1}{\norm{\Sigma^{1/2}\theta' }}  \Big|\\
&\le 2\frac{\norm{\Sigma^{1/2}(\theta - \theta')}}{\norm{\Sigma^{1/2} \theta}}
\le 2\sqrt{\frac{\lambda_{\max}(\Sigma)}{\lambda_{\min}(\Sigma)}}\, \|\theta-\theta'\|
\end{align*}
Finally $(c)$ holds since
\[
\E\{\xi_1^1\} - 1 = \bvsigma \varphi(\bvsigma)/\Phi(-\bvsigma) \le  \bvsigma^2+1\,,
\]
using standard tail bound $\varphi(\bvsigma)\frac{\bvsigma}{\bvsigma^2+1} \le \Phi(-\bvsigma)$.
\end{proof}

\section{Technical preliminaries}\label{sec:proofs}

\begin{definition}(Subgaussian norm)
The subgaussian norm of a random variable $X$, denoted
by $\norm{X}_{\psi_2}$, is defined as
\begin{align*}
\norm{X}_{\psi_2} &\equiv \sup_{q\ge 1} q^{-1/2} \E\{\abs{X}^q\}^{1/q}.
\end{align*}
For a random vector $X$ the subgaussian norm is defined as
\begin{align*}
\norm{X}_{\psi_2} &\equiv \sup_{\norm{v} = 1} \norm{\<X, v\>}_{\psi_2}. 
\end{align*}
\end{definition}

\begin{definition}(Subexponential norm)
The subexponential norm of a random variable $X$ is
defined as
\begin{align*}
\norm{X}_{\psi_1} &\equiv \sup_{q \ge 1} q^{-1}\E\{\abs{X}^q\}^{1/q}. 
 \end{align*}
 For a random vector $X$ the subexponential norm is defined by
 \begin{align*}
\norm{X}_{\psi_1} &\equiv \sup_{\norm{v} = 1} \norm{\<X, v\>}_{\psi_1}. 
 \end{align*}
\end{definition}

\begin{definition}(Uniformly subgaussian/subexponential sequences)
We say a sequence of random variables $\{X_i\}_{i\ge 1}$ adapted to a 
filtration $\{\cF_i\}_{i\ge 0}$ is \emph{uniformly
$K$-subgaussian} if, almost surely:
\begin{align*}
\sup_{i\ge 1} \sup_{q\ge 1} q^{-1/2} \E\{\abs{X_i}^q |\cF_{i-1}\}^{1/q} 
&\le K. 
\end{align*}
A sequence of random vectors $\{X_i\}_{i\ge 1}$
is uniformly $K$-subgaussian if, almost surely,
\begin{align*}
\sup_{i\ge 1} \sup_{\norm{v}=1} \sup_{q\ge 1} \E\{\abs{\<X_i, v\>}^q|\cF_{i-1}\}^{1/q} &\le K. 
\end{align*}
Subexponential sequences are defined analogously, replacing
the factor $q^{-1/2}$ with $q^{-1}$ above. 
\end{definition}

\begin{lemma}\label{lem:subexpprod}
For a pair of random variables $X, Y$, $\norm{XY}_{\psi_1} \le 2\norm{X}_{\psi_2} \norm{Y}_{\psi_2}$.
\end{lemma}
\begin{proof}
By Cauchy Schwarz:
\begin{align*}
\norm{XY}_{\psi_1} &= \sup_{q \ge 1} q^{-1} \E\{|XY|^{q}\}^{1/q} \\
&\le \sup_{q\ge 1} q^{-1} \E\{\abs{X}^{2q}\}^{1/2q} \E\{\abs{Y}^{2q}\}^{1/2q}\\
&\le 2\big(\sup_{q \ge 2} (2q)^{-1/2} \E\{\abs{X}^{2q}\}^{1/2q} \big)
\cdot \big(\sup_{q \ge 2} (2q)^{-1/2} \E\{\abs{Y}^{2q}\}^{1/2q} \big)\\
&\le 2\norm{X}_{\psi_2} \norm{Y}_{\psi_2}. 
\end{align*}
\end{proof}

The following lemma from \cite{vershynin2012introduction} is a Bernstein-type tail
inequality for sub-exponential random variables. 
\begin{lemma}[{\cite[Proposition~5.16]{vershynin2012introduction}}]
\label{lem:subexptail}
Let $X_1, X_2, \dots, X_n$ be a sequence
of independent random variables 
with $\max_{i}\norm{X_i}_{\psi_1}\le K$.
Then for any $\eps \ge 0$:
\begin{align}
\P\Big\{\Big\lvert\frac{1}{n}  \sum_{i=1}^n X_i -\E\{X_i\}  \Big \rvert \ge  \eps \Big\} &\le 2 \exp \Big\{ - \frac{n \eps}{6e K} \min\Big( \frac{\eps}{eK},1 \Big) \Big\}
\end{align}
\end{lemma}

We also use a martingale generalization of 
\cite[Proposition~5.16]{vershynin2012introduction}, whose
proof is we omit.  
\begin{lemma}
\label{lem:martingalesubexptail}
Suppose $(\cF_i)_{i\ge 0}$ is a filtration,  $X_1, X_2, \dots, X_n$
is a uniformly $K$-subexponential sequence of random variables adapted to $(\cF_i)_{i\ge 0}$ such that almost surely $\E\{X_i | \cF_{i-1}\} = 0$.
Then for any $\eps \ge 0$:
\begin{align}
\P\Big\{\Big\lvert\frac{1}{n}  \sum_{i=1}^n X_i  \Big \rvert \ge  \eps \Big\} &\le 2 \exp \Big\{ - \frac{n \eps}{6e K} \min\Big( \frac{\eps}{eK},1 \Big) \Big\}
\end{align}
\end{lemma}

The following is a rough bound on the LASSO error. 

\begin{lemma}[Rough bound on LASSO error]\label{lem:sizeofLASSO}
For LASSO estimate $\htheta^\sL$ with regularization $\lambda_n$ the following bound holds:
\begin{align*}
\norm{\htheta^\sL - \theta_0}_1 &\le \frac{\norm{\eps}^2}{2n\lambda_n } + 2\|\theta_0\|_1\,.
\end{align*}
\end{lemma}
\begin{proof}[Proof of Lemma \ref{lem:sizeofLASSO}]
We first bound the size of $\htheta^\sL$. By optimality of $\htheta^\sL$:
\begin{align*}
 \lambda_n \norm{\htheta^\sL}_1
&\le \frac{1}{2n}\norm{\eps}_2^2 +\lambda_n \norm{\theta_0}_1 - \frac{1}{2n} \norm{y - X\htheta^\sL}_2^2 \\
&\le \frac{1}{2n}\norm{\eps}_2^2 +\lambda_n \norm{\theta_0}_1.
\end{align*}
We now use triangle inequality and the bound above to get the claim:
\begin{align*}
\norm{\htheta^\sL - \theta_0}_1 &\le \norm{\htheta^\sL}_1 + \norm{\theta_0}_1 \\
&\le \frac{1}{2n\lambda_n} \norm{\eps}^2 + 2\norm{\theta_0}_1 \,.
\end{align*}
\end{proof}


\section{Simulation results for the Dominick's data set}\label{app:simulation}
In this section we report the $p$-values obtained by the online debiasing for the cross-category effects.
Figures~\ref{fig:market1}, \ref{fig:market2}, \ref{fig:market3} provide the $p$-values corresponding to the effect of price, sale, and promotions of different categories on the other categories, after one week ($d=1$) and two weeks ($d=2$). The darker cells indicate smaller $p$-values and hence higher statistical significance.
\begin{figure}
\begin{subfigure}{\textwidth}
\centering
  \includegraphics[scale=0.49]{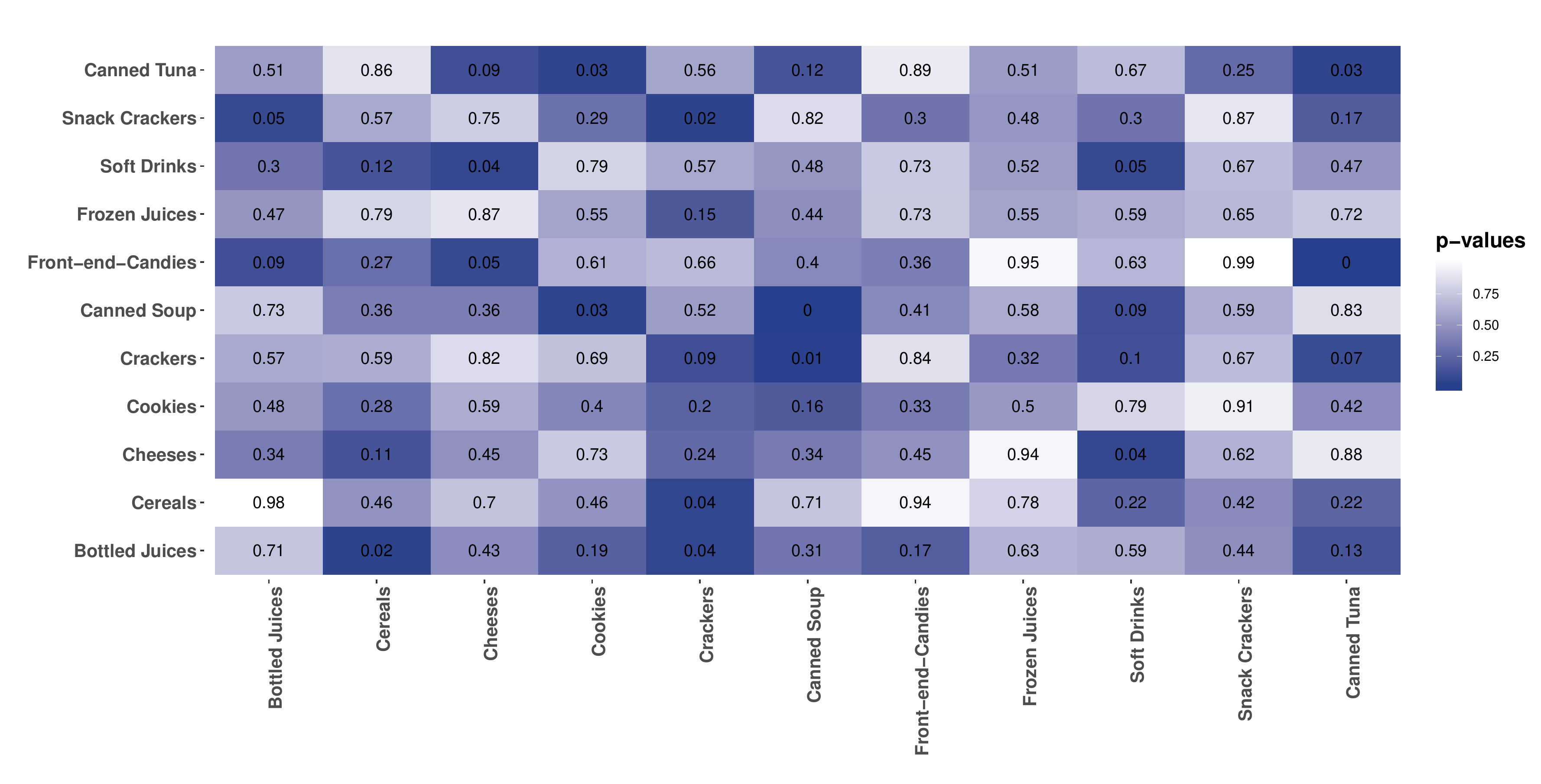}
 \caption{{\bf 1-Week} effect of {\bf sales} of $x-$axis categories on {\bf sales} of $y-$axis categories}
 \label{fig:d_1_sales_on_sales}
\end{subfigure}
\begin{subfigure} {\textwidth} 
\centering
\includegraphics[scale=0.49]{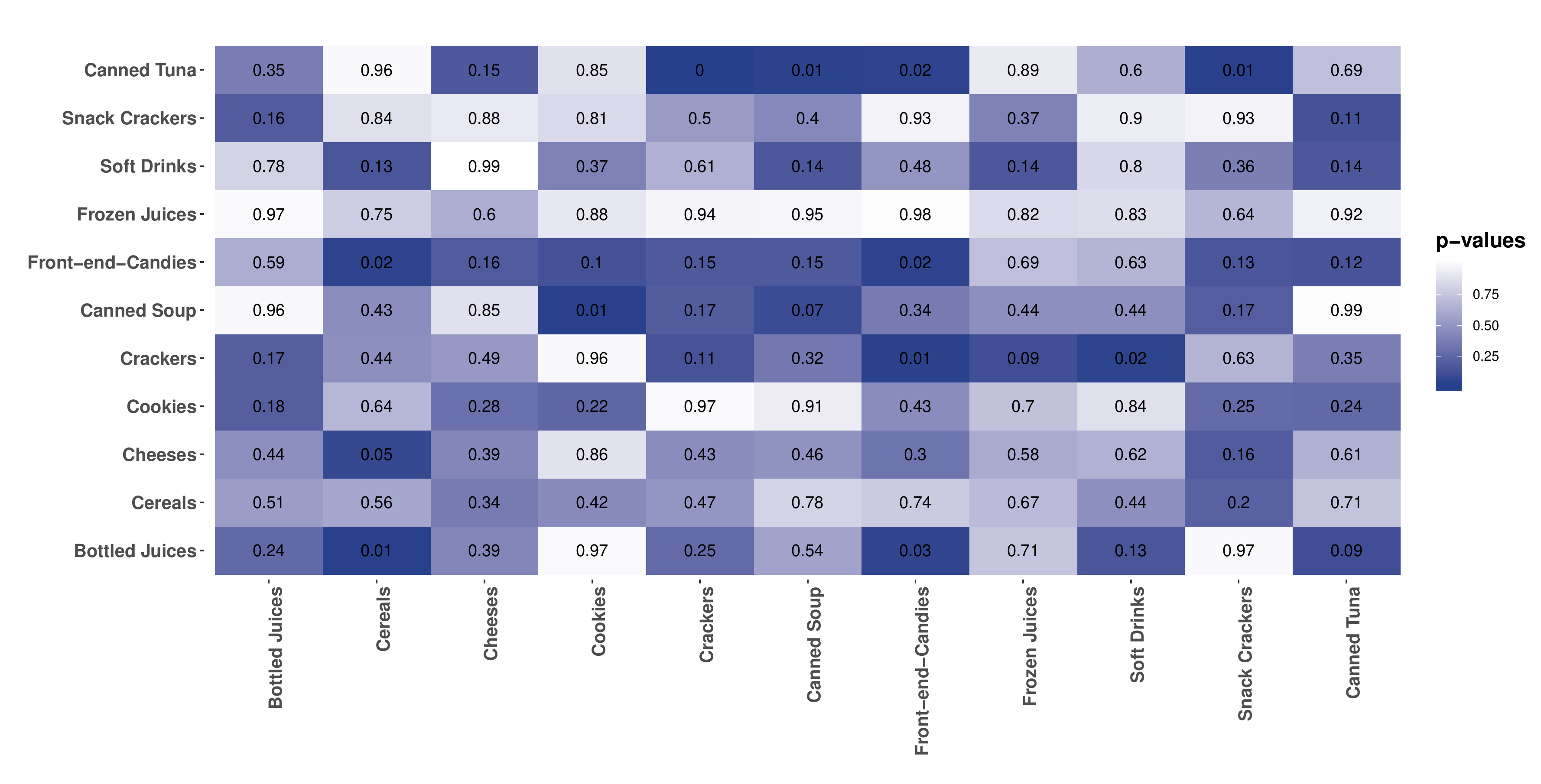}
 \caption{{\bf 1-Week} effect of {\bf prices} of $x-$axis categories on {\bf sales} of $y-$axis categories}
 \label{fig:d_1_price_on_sales}
\end{subfigure}
\caption{{\small Figures \ref{fig:d_1_sales_on_sales}, and \ref{fig:d_1_price_on_sales} respectively show the $p$-values for cross-category effects of sales and prices of $x-$axis categories on sales of $y-$axis categories after one week.}}\label{fig:market1}
\end{figure}

\begin{figure}
\begin{subfigure}{\textwidth}
  \centering
  \includegraphics[scale=0.48]{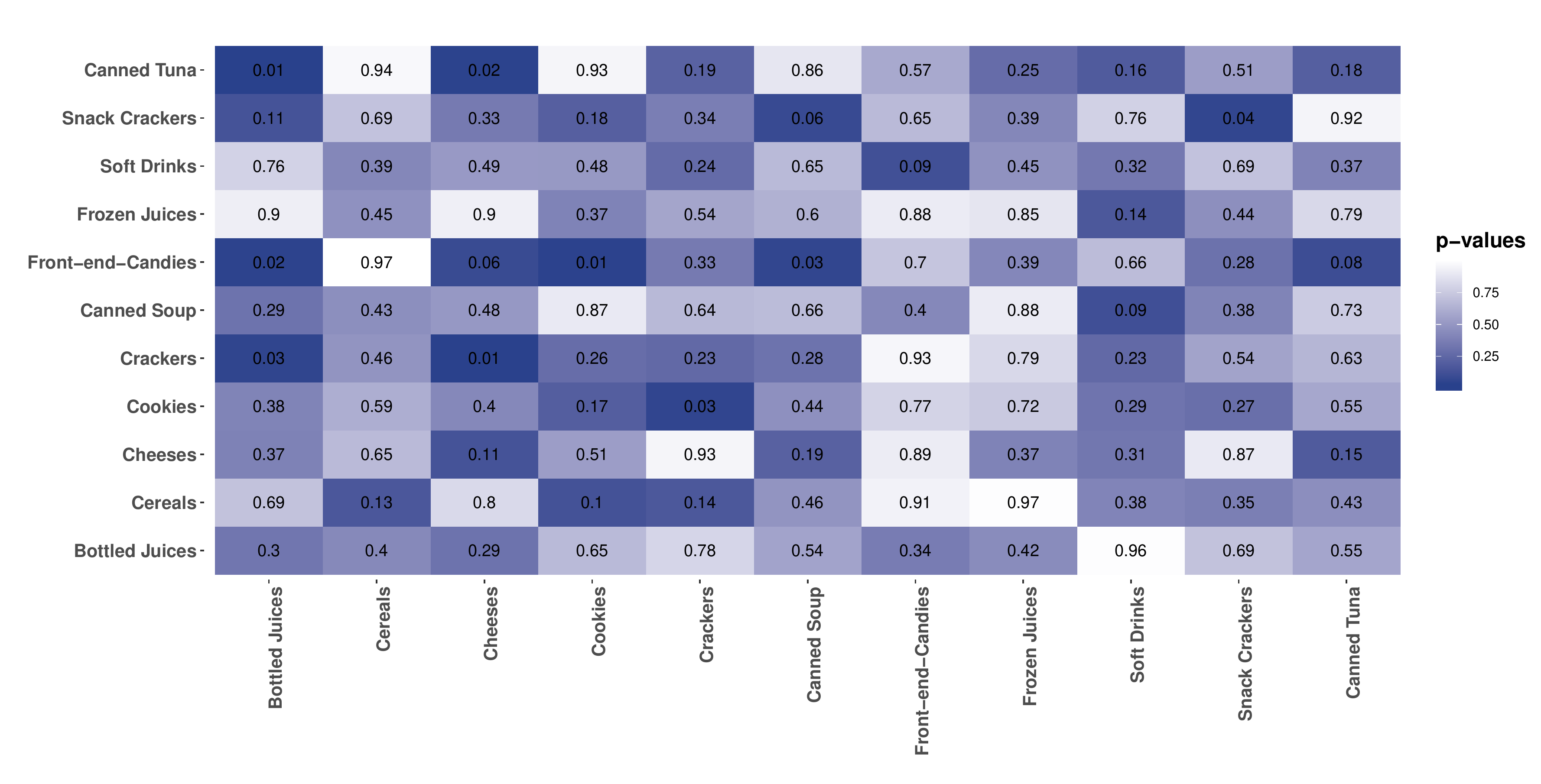}
 \caption{{\bf 1-Week} effect of {\bf promotions} of $x-$axis categories on {\bf sales} of $y-$axis categories}
\label{fig:d_1_proms_on_sales}
\end{subfigure}

\begin{subfigure}{\textwidth}
  \centering
  \includegraphics[scale=0.48]{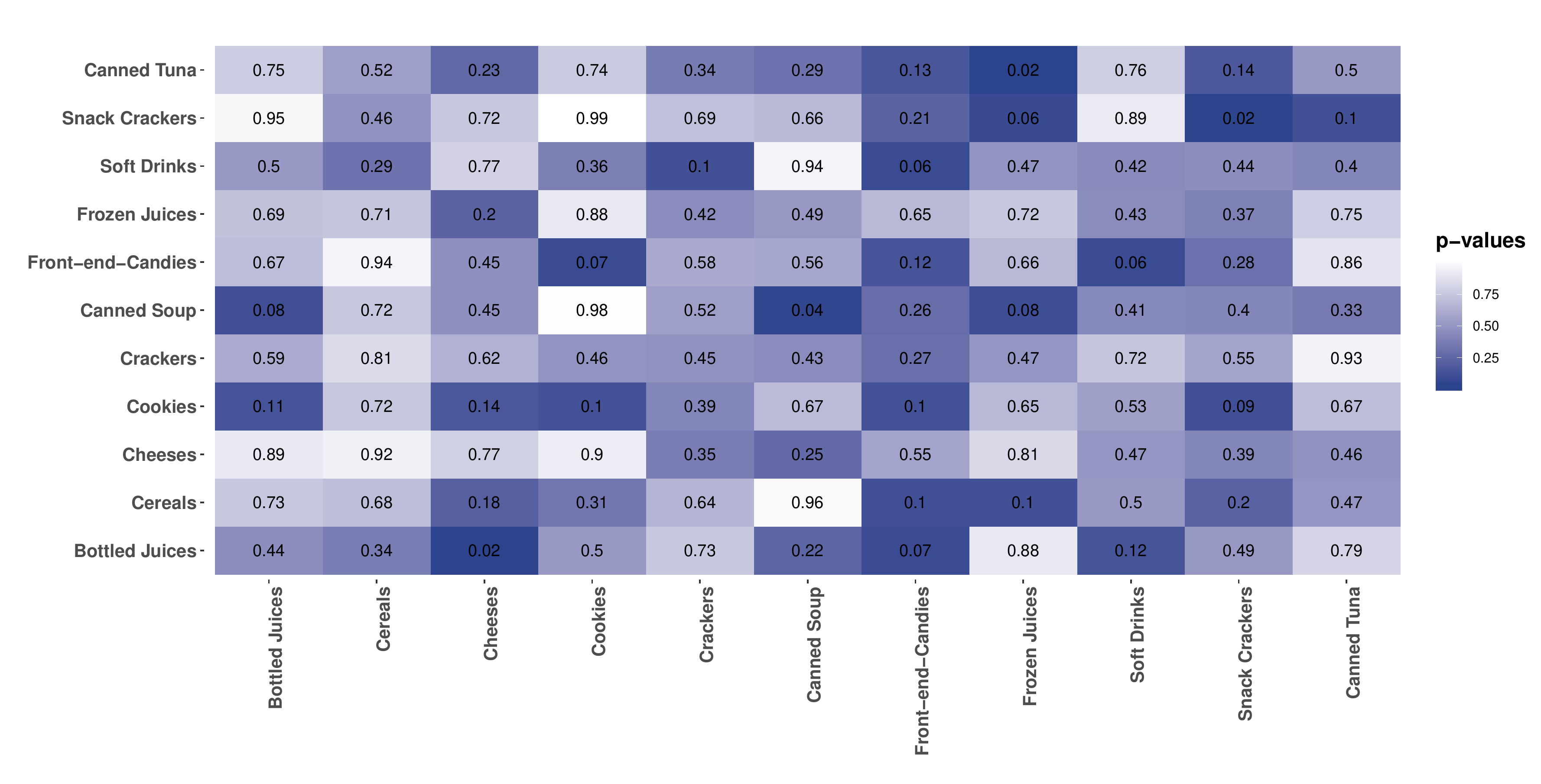}
 \caption{{\bf 2-Week} effect of {\bf promotions} of $x-$axis categories on {\bf sales} of $y-$axis categories}
\label{fig:d_2_prom_on_sale}
\end{subfigure}
\caption{{\small Figures \ref{fig:d_1_proms_on_sales}, and \ref{fig:d_2_prom_on_sale} show $p-$values for cross-category effects of promotions of $x-$axis categories on sales of $y-$axis categories, after one week and two weeks.}}\label{fig:market2}
\end{figure}

\begin{figure}
\begin{subfigure}{\textwidth}
  \centering
  \includegraphics[scale =0.49]{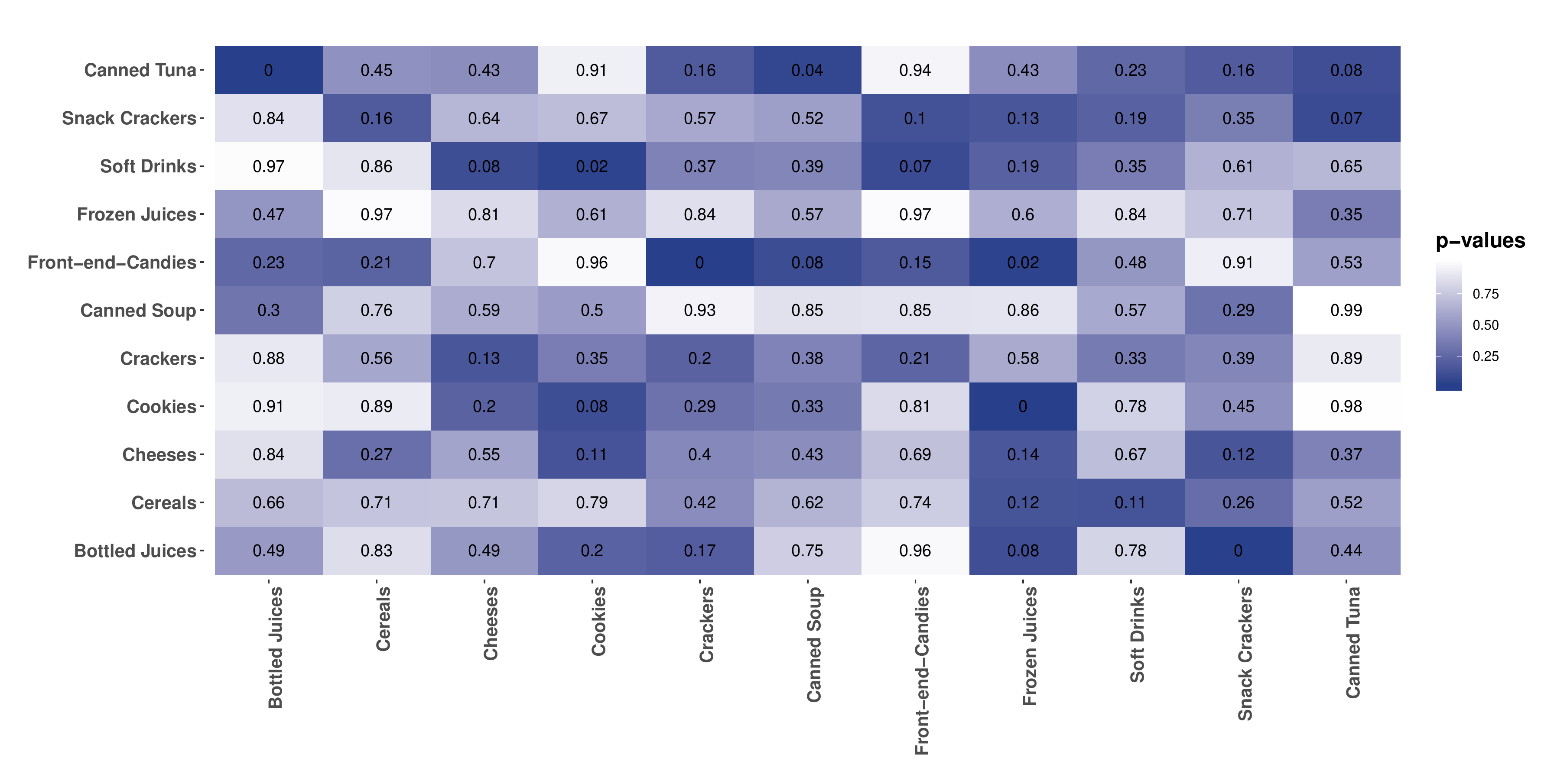}
 \caption{{\bf 2-Week} effect of {\bf sales} of $x-$axis categories on {\bf sales} of $y-$axis categories}
 \label{fig:d_2_sale_on_sale}
\end{subfigure}
\begin{subfigure}{\textwidth}
  \centering
  \includegraphics[scale=0.49]{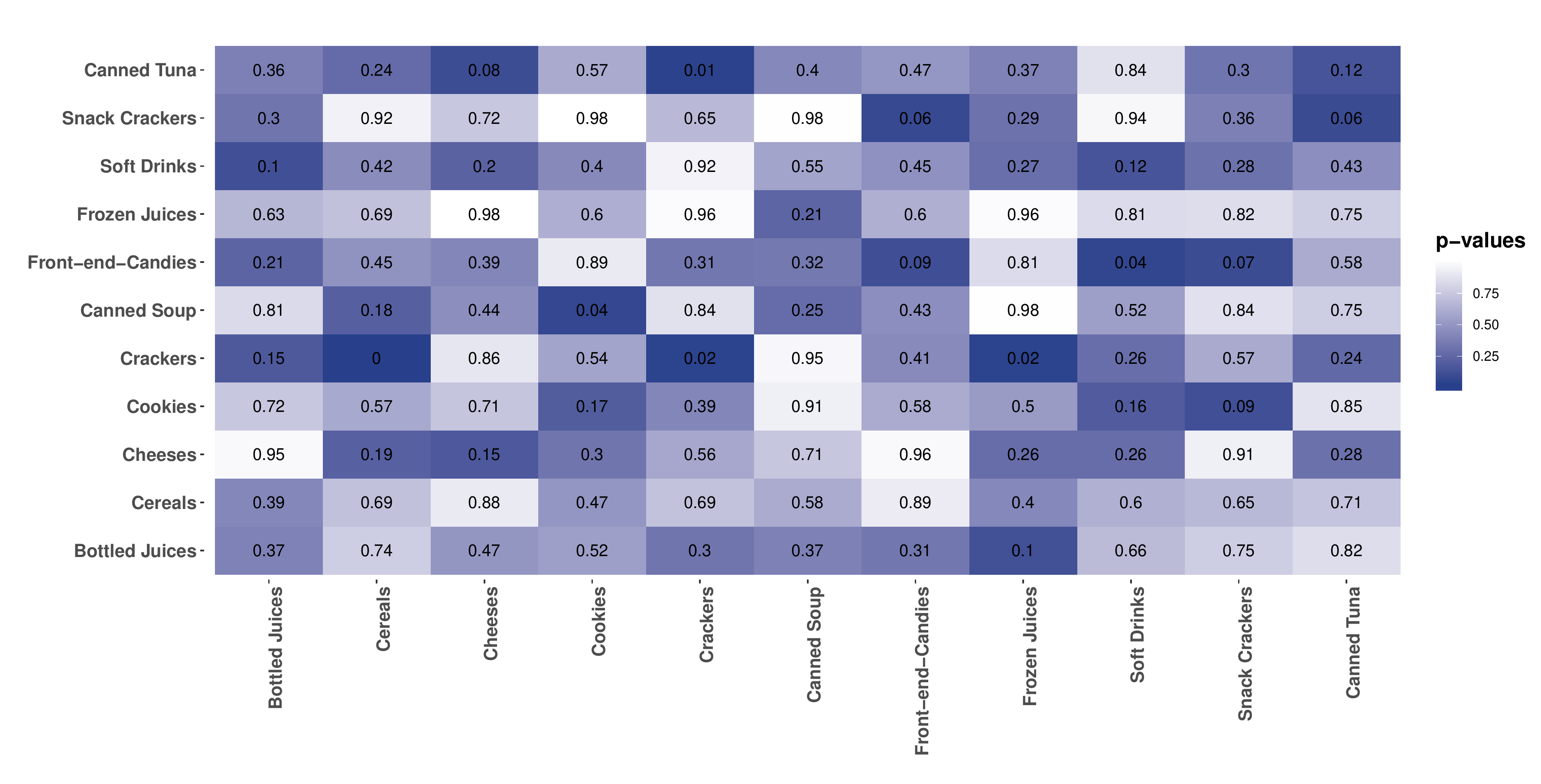}
 \caption{{\bf 2-Week} effect of {\bf prices} of $x-$axis categories on {\bf sales} of $y-$axis categories}
\label{fig:d_2_price_on_sale}
\end{subfigure}
\caption{{\small Figures \ref{fig:d_2_sale_on_sale}, and \ref{fig:d_2_price_on_sale} respectively show $p$-values for cross-category effects of sales and prices of $x$-axis categories on sales of $y-$axis categories after two weeks.}}
\label{fig:market3}
\end{figure}

\end{document}